\documentclass[reqno]{amsart}
\usepackage{amssymb}
\usepackage[dvips]{epsfig}
\usepackage{graphicx}
\usepackage{color}
\usepackage{amsmath}
\usepackage{amssymb}
\usepackage{amsfonts}
\usepackage{amsthm}
\usepackage{tikz}
\newcommand{\R}{\mathbb R}

\newcommand{\Z}{\mathbb Z}

\renewcommand{\l}{\lambda}

\newcommand{\N}{\mathbb{N}}
\newcommand{\T}{\mathbb{T}}

\newtheorem{thm}{Theorem}[section]
\newtheorem{hp}{\bf Hypothesis}
\newtheorem{tp}{\bf Type}
\newtheorem{lem}[thm]{Lemma}
\newtheorem{prop}[thm]{Proposition}

\newtheorem{rem}{\bf Remark}[section]
\theoremstyle{definition}

\theoremstyle{statement}

\numberwithin{equation}{section}
\usepackage[colorlinks=true,pdfstartview=FitV,linkcolor=magenta,citecolor=cyan]{hyperref}
\usepackage{bm}
\begin{document}
	\title[Localization for asymmetric $C^2$ potential]{On the spectrum of   quasi-periodic Schr\"odinger operators on $\Z^d$  with $C^2$-cosine type potentials}
	\author[Cao]{Hongyi Cao}
	\address[H. Cao] {School of Mathematical Sciences,
		Peking University,
		Beijing 100871,
		China}
	\email{chyyy@stu.pku.edu.cn}
	\author[Shi]{Yunfeng Shi}
	\address[Y. Shi] {School of Mathematics,
		Sichuan University,
		Chengdu 610064,
		China}
	\email{yunfengshi@scu.edu.cn}
	
	\author[Zhang]{Zhifei Zhang}
	\address[Z. Zhang] {School of Mathematical Sciences,
		Peking University,
		Beijing 100871,
		China}
	\email{zfzhang@math.pku.edu.cn}
	
	\date{\today}
	
	\maketitle
	\begin{abstract}
		In this paper,  we establish  the Anderson localization,   strong dynamical localization and  the $(\frac 12-)$-H\"older continuity of the integrated density of states (IDS)  for some  multi-dimensional discrete quasi-periodic (QP) Schr\"odinger operators with asymmetric $C^2$-cosine type  potentials. {  We extend both the iteration scheme of \cite{CSZ23a}  and the  interlacing method of \cite{FV21}   to handle  {\it asymmetric} Rellich functions with {\it collapsed gaps}.}
		%we employ the multi-scale analysis and Green's function estimates methods to  give a perturbative proof of  almost-sure Anderson  (and dynamical) localization and $(\frac{1}{2}-)$-H\"older regularity of the integrated density of states  (IDS) for quasi-periodic  (QP) discrete Schr\"odinger operators  $H=V+\varepsilon\Delta$ on $\Z^d$ with  $V$ sampled from an asymmetric,  	$C^2$-cosine-like function $v\in C^2(\T;\R)$ by a  Diophantine frequency vector $\omega$ for sufficiently small $\varepsilon\leq \varepsilon_0(v,d,\omega)$.
	\end{abstract}
	
	\maketitle 
	
	\section{Introduction}
	The spectral and dynamical theory  of QP  Schr\"odinger operators has attracted great attention over the years. Of particular importance is the phenomenon  of  Anderson localization, which means, the operator has pure point spectrum with exponential decay eigenfunctions. { The pioneering works of Sinai \cite{Sin87} and  Fr\"ohlich-Spencer-Wittwer \cite{FSW90}  have successfully established  the Anderson localization  for a class of  one-dimensional QP Schr\"odinger operators with $C^2$-cosine type potentials. The result of \cite{FSW90} was extended recently to more general potentials  without symmetry condition by  Forman-VandenBoom \cite{FV21} via introducing  the  new  interlacing method. It is well-known that the openness of gaps between Rellich functions  (i.e., parameterized eigenvalues), which restricts to the one  dimensional lattice, plays an essential in \cite{Sin87,FSW90,FV21}. In this paper, we extend the iteration scheme of \cite{CSZ23a} and the  interlacing method  of \cite{FV21}   to handle asymmetric Rellich functions with collapsed gaps. This allows us to prove   the Anderson localization,   strong dynamical localization and  the $(\frac 12-)$-H\"older continuity of the IDS for some  multi-dimensional discrete QP  Schr\"odinger operators with asymmetric $C^2$-cosine type  potentials.}
	%We offer a new perturbative proof of  localization and $(\frac 12-)$-H\"older continuity of the IDS for   multi-dimensional QP Schr\"odinger operators of such class. 
	% Removing the symmetry condition in the $\Z^d$  case requires  to exclude the resonance of a pair of  collapsed  Rellich functions in the horizon direction. In the symmetry case, this can be done by showing that the above pair of Rellich functions is also symmetric. Without this symmetry, we have to construct auxiliary functions  to avoid the horizon self-resonance of the collapsed Rellich functions, which needs  to  extend  some  ideas   (based on  the  Cauchy interlacing theorem) of \cite{FV21} to  the present case.  Obviously, the open gaps property is essential  in the iteration scheme of   \cite{FV21},  but  it is our main purpose here to show the presence of collapsed gaps are actually permitted. This may provide  a new candidate method for handling  effective  localization for multi-dimensional QP operators.  Besides the proof of  localization by establishing the Green's function estimates, we also prove  the $(\frac12-)$-H\"older continuity of the IDS by using some ideas of Bourgain \cite{Bou00} (for more results on the regularity of the IDS, we refer to \cite{CSZ23} and references therein). 
	\subsection{Main result}
	In this paper, we  are concerned with  the  QP   Schr\"odinger operators %on $\mathbb{Z}^d$
	\begin{align}\label{model}
		H(\theta)=\varepsilon \Delta+v(\theta+ x\cdot{\omega})\delta_{x, y},\ x\in\Z^d,
	\end{align}
	where $\varepsilon\geq0$ and  the discrete Laplacian $\Delta$ is defined as
	\begin{align*}
		\Delta(x, y)=\delta_{{\|x- y\|_{1}, 1}},\ \| x\|_{1}:=\sum_{i=1}^{d}\left|x_{i}\right|.
	\end{align*}
	For the diagonal part of \eqref{model},  we let $\theta\in \mathbb{T}=\mathbb{R}/\mathbb{Z},  \  \omega\in 	{\rm DC}_{\tau, \gamma}$ and $x\cdot\omega=\sum\limits_{i=1}^dx_i \omega_i$, 	with	$$	{\rm DC}_{\tau, \gamma}=\left\{\omega\in[0,1]^d:\ \|x\cdot\omega\|=\inf_{l\in\mathbb{Z}}|l- x\cdot\omega|\geq \frac{\gamma}{\| x\|_1^{\tau}},\ \forall\  x\in\mathbb{Z}^d\setminus\{0\}\right\},$$ 
	where $\tau>d,\gamma>0.$ We call $\theta$ the phase  and $\omega$ the frequency. 
	{\bf We  assume  that the  potential $v\in C^2(\mathbb{T}; \R)$ has \textit{exactly two non-degenerate critical points}}.

	{For a finite subset $\Lambda\subset \Z^d$, we denote by $H_\Lambda:=R_\Lambda H R_\Lambda$ the restriction of $H$ to $\Lambda$ with Dirichlet boundary condition, which can be identified with a $|\Lambda|\times|\Lambda|$ self-adjoint matrix}.  First, we define the integrated density of states (IDS). For a finite set $\Lambda$, denote by $\#\Lambda$ the cardinality of $\Lambda.$ Let
	$$\mathcal{N}_{\Lambda}(E;\theta)=\frac{1}{\#\Lambda}\#\{\lambda\in\sigma({H_\Lambda(\theta)}):\  \lambda\leq E\}$$
	and denote by
	\begin{align}\label{ids}
		\mathcal{N}(E)=\lim_{N\to\infty}\mathcal{N}_{\Lambda_N}(E;\theta)
	\end{align}
	the IDS, where $\Lambda_N=\{x\in\Z^d:\ \|x\|_1\leq N\}$ for $N>0$ and $\sigma(\cdot)$ denotes the spectrum.  It is well-known that the limit in  \eqref{ids} exists and is independent of $\theta$  for a.e. $\theta$.

	Our main results  of this paper are
	%\subsection{Main results} In this part, we will introduce our main results. 

	\begin{thm}[Main Theorem]\label{main}
		Let $H(\theta)$ be  given by \eqref{model}.	We assume that   $v\in C^2(\mathbb{T}; \R)$  with  exactly two non-degenerate critical points and  $\omega\in {\rm DC}_{\tau, \gamma}$. Then there exists some $\varepsilon_0=\varepsilon_0(v,d,\tau,\gamma)>0$, such that for all $0< \varepsilon\leq \varepsilon_0$,
		\begin{itemize}
			\item[1.] For (Lebesgue) a.e.  $\theta\in \T$, $H(\theta)$ satisfies Anderson localization, i.e., has pure point spectrum with exponentially decaying eigenfunctions;
			\item[2.] $H(\theta)$ satisfies  strong dynamical  localization, i.e.,  for any $q>0$, 
			% $$\sup_{t\in \mathbb{R}}\sum_{x\in \mathbb{Z}^d}(1+\|x\|_1)^q|\langle e^{itH(\theta)}{\bm e}_0, {\bm e}_x\rangle|<+\infty, $$
			
			% \\	Moreover, $H(\theta)$ satisfies  the expectation version of  Dynamic localization, i.e., 
			$$
			\int_\mathbb{T}\sup_{t\in \mathbb{R}}\sum_{x\in \mathbb{Z}^d}(1+\|x\|_1)^q|\langle e^{itH(\theta)}{\bm e}_0, {\bm e}_x\rangle|d\theta<+\infty,
			$$
			where $\{\bm e_x\}_{x\in\Z^d}$ denotes the standard basis of $\ell^2(\Z^d)$.
			\item [3.]
			%For all  $\eta>0$ and for sufficiently large $N$  (depending on $\eta$), 
			%	   \begin{align}%\label{holder}
				%	   	\nonumber&\ \ \ \sup_{\theta^*\in\T, E^*\in\R}\left(\mathcal{N}_{\Lambda_N} (E^*+\eta;\theta^*)-\mathcal{N}_{\Lambda_N} (E^*-\eta;\theta^*)\right)\\
				%	   	\label{holder}&\leq C (d)\eta^{\frac{1}{2}}\max (1,|\log\eta|^{8d}),
				%	   \end{align}	
			%	   where $C (d)>0$ depends only on $d$. In particular,
			
			The IDS of $H(\theta)$  is	  $(\frac12-)$-H\"older continuous, i.e., for all $E\in \R, \eta>0,$
			$$\mathcal{N}(E+\eta)-\mathcal{N}(E-\eta)\leq C \eta^{\frac{1}{2}}\max(1,|\log\eta|^{16d}),$$
			%of exponent $\iota$ for any $\iota\in (0, \frac{1}{2})$.  
			where $C=C(d)>0.$
		\end{itemize}
	\end{thm}
	
	\begin{rem}
		Definitely, the Anderson localization for the present  model has been obtained previously by Dinaburg \cite{Din97}. However, it is unclear if the stronger version of localization called arithmetic Anderson localization {\footnote{This means  the Anderson localization holds for phase $\theta \in\T$  in a set of full measure with explicitly arithmetic descriptions (c.f. \cite{CSZ23a}).}} still holds true if the potential is symmetric via the method of \cite{Din97}. Our method can   directly  imply  the arithmetic Anderson localization once the potential is symmetric.
	\end{rem}
	\begin{rem}
		For recent  progress on dynamical localization   for QP operators on $\Z^d$ with symmetric potentials for fixed Diophantine type frequency, we refer to \cite{GYZ23, CSZ23a, CSZ23b}. 
	\end{rem}
	\begin{rem}
		For results on the regularity of  the IDS of  QP  Schr\"odinger operators on $\Z^d$, we refer to \cite{Sch01, Bou07, Liu22, GYZ22, CSZ23a, CSZ23b}. 
	\end{rem}
	\subsection{Previous works}
	To prove the Anderson localization of QP Schr\"odinger operators, one has to deal with the small divisors difficulty. This inspires the usage of the perturbative KAM type method. Both the  aforementioned works \cite{Sin87} and \cite{FSW90} used the perturbative KAM type method  motivated by Fr\"ohlich-Spencer \cite{FS83}  in the random case, while  their employed techniques to eliminate the double resonances  are  significantly  different. First of all, Sinai  developed a perturbative iteration scheme to construct  approximate   Rellich functions. Instead, Fr\"ohlich-Spencer-Wittwer focused on  Green's function estimates  based on a multi-scale type analysis (MSA).  Second,  Sinai's construction relies on  gap  estimates  of  shifts of  Rellich functions by the Diophantine frequency,   and thus  is a  global one.  In  \cite{FSW90},  however, the authors mainly did the analysis in a neighborhood of given energy $E\in \R$.   It turns out the method  of Sinai is completely  free from the {\it symmetry} restriction of the potential $v\in C^2(\T;\R)$. While the work of \cite{FSW90} requires  the potential to be  symmetric,  it proves a more stronger result, namely, the arithmetic Anderson localization.  { Recently, Forman-VandenBoom \cite{FV21}  removed  the symmetry condition of \cite{FSW90} and  proved  both the Cantor spectrum and Anderson localization  by  combining some ideas of \cite{FSW90} and  a celebrated Cauchy interlacing argument.}  A perturbative (KAM reducibility) proof  of pure point spectrum for more general one-dimensional  QP Schr\"odinger operators was  obtained previously  by Eliasson \cite{Eli97}.  Finally, we want to mention that the breakthrough works  of Jitomirskaya \cite{Jit94, Jit99},  in which a non-perturbative  method was first  developed for the almost Mathieu operators, have  motivated  considerable progress on Anderson localization for one-dimensional QP  Schr\"odinger operators (c.f. \cite{BG00, BJ02, Bou05,  AYZ17, JL18, JL23, Liu23} for just a few).   
	
	Clearly, all the methods of \cite{Sin87,FSW90,FV21} restrict  to  the  one-dimensional case: in this case,  as mentioned above,  all  gaps  of eigenvalues of the truncated operators have   a priorly  positive lower bound, which plays an essential role in dealing with resonances. If one tries to extend the proofs of \cite{Sin87,FSW90,FV21} to multi-dimensional cases, there comes serious difficulty, namely,   some  spectral gaps may be collapsed.   In the spirit of \cite{FSW90},  Surace \cite{Sur90}  first  handled some QP  Schr\"odinger operators on $\Z^2$ with the symmetric  potentials  having exactly one non-degenerate critical point.  Very recently, the authors \cite{CSZ23a} developed the methods of \cite{FSW90, Sur90}  further,  and  completely extended the result of \cite{FSW90} to $\Z^d$ for arbitrary $d\geq1$. Definitely, in \cite{CSZ23a}, we also require  the potentials to satisfy the  symmetry  condition.  We should remark that the ideas of Sinai \cite{Sin87} have been extended by Chulaevsky and Dinaburg \cite{CD93, Din97}  to the $\Z^d$ and long-range case without the symmetry restriction, so it is natural to ask whether the symmetry condition is necessary or not in the analysis of \cite{CSZ23a}. In the present, we completely remove the symmetry condition of \cite{CSZ23a}, thus extend  the work of \cite{FSW90} to both multi-dimensional lattice and asymmetric potentials  cases. 
	We also mention the works  Bourgain-Goldstein-Schlag \cite{BGS02}, Bourgain \cite{Bou07} and Jitomirskaya-Liu-Shi \cite{JLS20} (c.f. \cite{Shi22,Liu22} for more recent results), in which  the  Anderson localization was proved  for multi-dimensional   QP  Schr\"odinger operators with general analytic potentials defined on multi-dimensional torus by combining the MSA method with subharmonic function estimates and semi-algebraic geometry theory (the results are a bit weaker in the sense that ``good'' frequencies are not always Diophantine type).  The arithmetic Anderson localization was  also  previously obtained for QP  Schr\"odinger operators on $\Z^d$ with the cosine potential via Aubry dual type localization-reducibility method by Ge-You  \cite{GY20} and direct Green's function estimates method \cite{CSZ23b} by the authors.   
	\subsection{New ingredients  of the proof}
	Our  proof relies heavily on ideas of  \cite{CSZ23a, Sur90, FSW90,FV21},  which  is  based  on inductive analysis on scales and  Green’s function estimates.  Once such estimates were obtained, the proof of both
	localization and  the  $(\frac{1}{2}-)$-H\"older regularity of the IDS just follows in a standard way.  
	
	{ In \cite{FSW90}, the authors assumed $d=1$ and the potential $v$ is even. They  developed a perturbative iteration scheme depending on locally defined, symmetric Rellich functions to analyze the resonances. Both the symmetry condition and $d=1$ in their proof are crucial since the  symmetry of $v$ forces the resonances to be located at $-k\omega/2$ mod $1$ (the rotation being $\theta\mapsto \theta+\omega$).   Then  the Rellich curves  are well-separated by an  eigenvalue separation lemma of one-dimensional Schr\"odinger operators since the critical points of these curves are located exactly  at these points. Combining the similar technique of \cite{FSW90} and some ideas of \cite{Sur90}, \cite{CSZ23a} successfully extended \cite{FSW90} to $d>1$.  Definitely, dropping the symmetry condition of $v$, even for $d=1$,   is a highly nontrivial matter of which the proof   needs new ideas.  This was recently  accomplished by Forman-VandenBoom \cite{FV21},  where they developed a new inductive procedure by constructing the Rellich functions  tree to handle the resonances, combined with  the  interlacing idea.}  Compared to one-dimensional case, the inductive analysis on scales for  multi-dimensional  operators  requires  to deal with   the resonance of a pair of  collapsed  Rellich functions in the vertical direction (if the potential $v$ is symmetric, the argument can be  largely simplified in \cite{CSZ23a}  as we will specify below).   Here we construct auxiliary functions  to avoid the horizon self-resonance of the these two collapsed Rellich functions by the Diophantine condition.  { Such a construction (of  auxiliary functions)  needs to extend  the interlacing argument of  \cite{FV21}  across multiple scales. We will prove  that the interlacing argument  can be applied in some sense to  level crossings in multi-dimensional setting, and that in this case   the presence of collapsed gaps in the iteration procedure are actually permitted.} This may provide  a new candidate method for handling  effective  localization for multi-dimensional QP operators.  Besides the proof of  localization by establishing the Green's function estimates, we also prove  the $(\frac12-)$-H\"older continuity of the IDS by using some ideas of Bourgain \cite{Bou00} (for more results on the regularity of the IDS, we refer to \cite{CSZ23a} and references therein). 
	
	{ Precisely, in contrast  to the works  of \cite{CSZ23a,FV21},  our new   accomplishments   are}: 
	\subsubsection{Removing symmetry}
	In  \cite{CSZ23a}, the authors proved  analogous results  (even stronger arithmetic  localization) under the additional symmetry  assumption of the sampling function: $v(\theta)=v(-\theta)$.  This assumption is crucial in our  proofs because the symmetry can be inherited by the Rellich functions    $E_n^i(\theta)$ of  Dirichlet restriction $H_{B_n^i}(\theta)$, where $B_n^i$ is an $n^{\rm th}$ scale block symmetric about $c_n^i$. With this symmetry,  we  can analyze the eigenvalue separation case and the  level crossing case of  Rellich curves respectively basing on the predetermined symmetrical points   $\theta_s=-c_n^i\cdot \omega $ of any Rellich functions of $H_{B_n^i}(\theta)$.  We  have proven  $$m(c_n^i,c_n^j):=\min(\|(c_n^i-c_n^j)\cdot\omega\|,\|2\theta^*+(c_n^i+c_n^j)\cdot\omega\| )$$ is small for any two $n$-resonant sites $c_n^i$ and $c_n^j$ relative to $\theta^*,E^*$. On one hand, the smallness of the above $m$ function together with Diophantine condition provides crucial separation of any three resonant sites at scale $n$, which enables the induction to  proceed. On the other hand, noting that this $m$ function does not depend on energy $E^*$,  one can define uniformly in energy the bad sets  of double resonant phases  explicitly  at each scale  and prove the  arithmetic version of localization. 
	Without symmetry, no such function $m$ can be  defined to be independent of $E^*$ and the above argument seems to  lapse. 
	
	In our present construction, the two-monotonicity interval structure   (i.e., whose domain can be divided into two closed intervals with disjoint interiors and  opposite-signed derivatives) and Morse condition  (i.e., lower bound of  $|E'|+|E''|$)  of Rellich functions, together with the  Diophantine condition provide the crucial separation of resonant sites. The bad set of double resonant phases is   defined via the preimage of double resonant energy interval of each Rellich function $E_n$. By carefully controlling the number of Rellich functions of   each scale,  we can again construct a uniform bad set of phases  at each scale  (c.f. \eqref{bad}), whose limit superior has zero Lebesgue measure such that localization holds on its complement. Moreover, if $v$ satisfies the symmetry  condition,  all the Rellich functions can be constructed symmetrically and the bad set of phases we ultimately  eliminate is  exactly that in \cite{CSZ23a}, which gives another proof of the arithmetic  Anderson localization for the symmetric case. However,  without the symmetry  assumption of $v$ and since the symmetry of the Rellich functions is not clear, one may not hope an  explicitly arithmetic relation of the bad set of phases with the frequency.
	\subsubsection{Collapsed gaps of  Rellich curves in  multi-dimensions}
	In \cite{FV21}, the authors proved almost-sure Anderson localization and Cantor spectrum in the one-dimensional  case. In their paper, they established  a uniform local eigenvalue separation (i.e., the openness of  gaps)
	$$\inf E_{n+1,>}-\sup E_{n+1,<}\gg \delta_{n+1}$$
	for the Rellich children $E_{n+1,>}>E_{n+1,<}$ resolved from a double resonant region of a Rellich function $E_n$ via a Cauchy interlacing argument and the classical eigenvalue separation lemma for one-dimensional  Schr\"odinger operators (c.f. \cite{FSW90}).
	
	Forman-VandenBoom \cite{FV21}  emphasized  two  crucial benefits  of  the above  uniform local separation property:  {``First, this separation guarantees that different Rellich curves cannot resonate with one another, which enables the induction to proceed on each single Rellich function. The second important consequence of the separation is that the size of  gap  from it is sufficiently large to remain open throughout the inductive procedure, which yields Cantor spectrum.''}
	{    Obviously, the uniform local  separation is necessary for Cantor spectrum. Since Cantor spectrum is not a generic phenomenon in the  multi-dimensional case   (c.f. Bourgain \cite{Bou05}),  one may not hope such uniform local separation still holds in higher  dimensions. However, we will show that without  the uniform local  separation, the iteration procedure still works and  the   two-monotonicity interval structure and   Morse condition   of  Rellich functions are enough for  the iteration  and  localization.} In our induction, we  establish  the Morse condition and two-monotonicity interval structure of  Rellich functions as in \cite{FV21}, however, we allow a Rellich function may resonate with another one  (called its Rellich brother, c.f. \textbf{Type} \ref{t3}) in the vertical direction (i.e., gap collapsed). {Noticing that this type of  resonance  (from different curves) is caused by a double resonance of a single Rellich curve in a previous scale  (c.f. Lemma \ref{Ta}), we  consider
		the scale $m$ when the resonance first appears and carefully  show the next  generation (scale $n+1$) of  Rellich functions maintains  the  inductive structure (c.f. Remark \ref{xin}). The main technical point  is the ``multi-scale interlacing argument'' (Proposition \ref{mul}), where we extend Forman-VandenBoom's  interlacing argument to  a new context.} % In addition, in their paper \cite{FV21},   the construction of regular sets when  establishing Green's function decay  (c.f. Lemma 5.7) and the bootstrapped Green's function decay in appendix seem  to depend heavily on the one-dimensional interval structure,  which  are not  totally obvious in higher  dimensions  where  we offer  new proofs  (resp., c.f. Appendix \ref{chouti}, Theorem \ref{rep}).
	\subsection{Outline of the proof}
	{ In the remainder of this section, we will show the  general  perturbative scheme of the proof.}
	By the perturbation theory of differentiable one-parameter self-adjoint
	operator families, $H_\Lambda(\theta)$ exhibits $|\Lambda|$  Rellich functions with $C^1$ smoothness, which are locally $C^2$ when being simple eigenvalues of $H_\Lambda(\theta)$. 
	
	For localization, from the well-known Schnol's Lemma, generalized eigenvalues of   $H$  (i.e., energy $E$ with solutions $\psi$ to $H\psi=E\psi$ with polynomial bound) are spectrally dense. Thus to prove Anderson localization, it suffices to show  every generalized eigenfunction $\psi$ in fact decays exponentially. Dynamical  localization can be  proved via a Hilbert-Schmidt argument if we have established  quantitative exponential decay of  eigenfunctions. We recall the Poisson formula for discrete Schr\"odinger operators 
	$$\psi(x)=-\sum_{z,z'}G_{\Lambda}(\theta,E;x,z)\Gamma_{z,z'} \psi(z'), \ x\in \Lambda , E\notin \sigma(H_\Lambda(\theta)),$$  
	where 
	the Green's function 
	$$G_{\Lambda}(\theta,E;x,z):=\left\langle \textbf{e}_x, G_{\Lambda}(\theta,E)\textbf{e}_z\right\rangle= \langle \textbf{e}_x, (H_\Lambda(\theta)-E)^{-1}\textbf{e}_z\rangle $$
	and the coupling  operator 
	\[\Gamma_{z,z'}:=\left\{\begin{aligned}
		&\varepsilon \quad \text{if } z\in \Lambda, z'\notin\Lambda, \|z-z'\|_1=1, \\
		&0 \quad \text{otherwise.} 
	\end{aligned}\right. \]
	If we can inductively construct ``good'' sets of increasing size on which we have  exponential off-diagonal  decay of Green's function for all generalized eigenvalues $E$, we can thus prove exponential decay of generalized eigenfunctions since the at most polynomial growth of generalized eigenfunctions can be absorbed by the  exponential   decay of Green's function, hence localization.
	
	For regularity of IDS,  we use ideas of Bourgain \cite{Bou00}:  By  a Hilbert-Schmidt argument, one can relate the operator norm of Green's function    to the number of eigenvalues  inside a certain interval $[E-\eta,E+\eta]$ of a finite restriction of $H$   (c.f. Lemma \ref{IDSL}). If we can construct ``good'' sets of increasing size on which we can handle the operator norm  of Green's function, then we can  prove the regularity of IDS.
	
	We recall the by-now classical resolvent identity argument of proving Green's function estimates   (results of this type  are well established in \cite{FS83} and \cite{Bou05}).  If one can cover $ \Lambda$ by  blocks   with resolvent bound (both operator norm growth and off-diagonal exponential decay estimates),  then one  can iterate the resolvent identity to prove  exponential off-diagonal  decay and operator norm bound of Green's function on  $\Lambda$.
	
	In our construction, the blocks are chosen to be translates of $B_n \ (n\geq 0)$, where $B_n$ is a $n^{\rm th}$ scale block centered at origin by carefully constructed.  Noticing  that  the  transition invariance of our model
	$$H_{B_n+x}(\theta)=H_{B_n}(\theta+x\cdot \omega )$$ and the  resolvent bound for self-adjoint operators
	$$\|(H_{B_n}(\theta)-E)^{-1}\|=\frac{1}{\operatorname{dist}(\sigma(H_{B_n}(\theta)), E)},$$
	we can relate the resolvent bound of $B_n+x$ to  the  Rellich functions of $H_{B_n}$ near the energy $E$. This is why we spend lots of efforts to construct such collection of Rellich functions.   
	
	At the $0^{\text{th}}$ scale, the block $B_0$ is exactly the single origin  site and the Rellich function  of $H_{B_0}(\theta)$ is exactly the sampling function $v(\theta)$. In each  step of  our induction, we  construct a bigger block $B_{n+1}$ and extend the locally defined  Rellich function $E_n$ of $H_{B_n}$ by a small perturbation to $E_{n+1}$ of $H_{B_{n+1}}$. The difficult aspect of the induction is verifying how the  Rellich functions inherit the Morse condition and two-monotonicity interval structure from the sampling function $v$. 
	
	At the first scale, we choose a scale length $l_1^{(1)}=|\log\varepsilon_0|^4$  and divide the energy axis into several overlapping intervals of size $\delta_0^{\frac{1}{100}}$, where $\delta_0:=\varepsilon_0^\frac{1}{20}$. Each energy interval can be characterized as double resonant, if it contains some $e_k$ satisfying $e_{k}=v(\theta)=v(\theta+k\cdot \omega )$ for some $\theta\in \T$ and $0<\|k\|\leq 10l_1^{(1)}$, or simple resonant if it  does not  (c.f. Proposition \ref{coverpr}). By the  two-monotonicity interval structure and   Morse condition of $v$,   it follows that any third $x$ with $\|x\|_1\leq (l_1^{(1)})^4$ such that $v(\theta+x\cdot \omega )$ recurs  to a $\delta_0^{\frac{1}{400}}$-neighborhood of $v(\theta)$ is excluded.    Each function $E_1$ is a Rellich function  of $H_{B_1}$ where $B_1\subset\Z^d$ is a cube of length $l_1^{(1)}$ if the energy interval is simple resonant, or $l_1^{(2)}=(l_1^{(1)})^2$ if the energy interval is double resonant.   For simple resonant interval, we show that any eigenfunction of $H_{B_1}$ with eigenvalue near $v$ is almost localized at origin. By Feynman-Hellman type  formulas  (c.f. Lemma \ref{daoshu}), it follows   that $H_{B_1}$ has a single Rellich function $E_1$ which is $C^2$ well approximated by $v$. For double resonant interval, we show that any eigenfunction of $H_{B_1}$ with eigenvalue near $v$ is almost localized at origin and the above $k$ satisfying $v(\theta+k\cdot \omega)=v(\theta)$. In this case, $v$ resolves as a pair of cosine-like Rellich functions $E_{1,>}\geq E_{1,<}$ of $H_{B_1}$. Unlike  one-dimensional case  as in \cite{FV21}, we are not able to show  the uniform local separation of these two Rellich functions. However, relying on the interlacing idea of \cite{FV21}, we can prove $E_{1,>}$ and  $E_{1,<}$ are interlaced by two auxiliary curves $\lambda_{\pm}$  (c.f. Proposition \ref{inter}), which enables our induction. 
	
	Assuming the constructed Rellich functions satisfy the  Morse condition and two-monotonicity interval structure, we classify  these functions into three types  (c.f. Hypothesis \ref{h4}).  The inductive argument analogous to the first scale proceeds on each such constructed Rellich function. We analyze all  cases to prove the next generation of Rellich functions maintain  the  Morse condition and two-monotonicity interval structure and belong to one of  the above three types, which  thus completes  the induction  (c.f. Figure \ref{f1}).  
	
	Having constructed the $n^{{\rm th}}$ generation of Rellich functions  class $\mathcal{C}_n$, we  give a criteria  of  $n$-good set $\Lambda\subset\Z^d$ by these Rellich functions  (c.f. Hypothesis \ref{h6}) and establish Green's function estimates on any finite $n$-good set via a multi-scale resolvent iteration. 
	
	The  final step to  establish localization is  eliminating  double resonances. At each scale,  we define the bad sets of phases where  double resonances occur  (c.f. \eqref{bad}). By utilizing the Morse condition of Rellich functions and carefully controlling the number of Rellich functions of  each scale, we ensure the bad sets of phases have summable measures.   Thus their limit   superior $K$ where double resonances occur infinitely often  has zero Lebesgue measure by Borel-Cantelli lemma. Our full-measure set on which localization occurs is the complement $\Theta=\T\setminus K$.  Given a phase $\theta\in \Theta$ and  a generalized eigenvalue $E$ of $H(\theta)$, we show that eventually there exists  an $n$-resonant site in the  $n$-scale cube centered at origin. Since double resonances have been eliminated, we show that there exists an $n$-good annulus of $(n+2)$-scale on which we have off-diagonal Green's function decay. Thus we can employ Poisson formula on this  annulus to exhibit the exponential decay of generalized eigenfunctions and complete the proof of  localization by Schnol's Lemma.  
	
	Moreover, fixing $\theta\in \T,  E\in \R$ and $\eta>0$,  for any sufficiently large cube $\Lambda_N$, we can  construct  an  $\eta^{\frac{1}{2}-}|\Lambda_N|$-deformation $\Lambda'$ with resolvent bound $\|G_{\Lambda'}(\theta,E)\|\leq \frac{1}{2}\eta^{-1}$, which together with Lemma \ref{IDSL} shows the  (finite volume version) $(\frac{1}{2}-)$-H\"older regularity of the IDS.
	\begin{figure}[htp]
		\begin{tikzpicture}[>=latex, scale=0.95]
			\draw  (-7,0)node{$\mathcal{C}_0$};
			\draw  (-7,-2)node{$\mathcal{C}_1$};
			\draw  (-7,-4)node{$\mathcal{C}_2$};
			\draw  (-7,-6)node{$\mathcal{C}_3$};
			\draw  (0,0)node{{\Large $v$}};
			\draw  (-4,-0.9)node{$\mathcal{C}_1^{(1)}(L=l_1^{(1)})$};
			\draw  (4,-0.9)node{{$\mathcal{C}_1^{(2)}(L=l_1^{(2)})$}};
			\draw  (-3,-2) ellipse  (1.5 and 0.8);
			\draw[dashed]  (-3,-1.2)-- (-3,-2.8);
			\draw  (-3.6,-1.3)[below]node{{\footnotesize \textbf{Type}1}};
			\draw  (-2.4,-1.3)[below]node{{\footnotesize \textbf{Type}2}};
			\draw  (0.2,-0.1)to (2.8,-1.2);
			\draw  (-0.2,-0.1)to (-2.8,-1.2);
			\draw  (-3.8,-2.1)node{{\footnotesize  $E_1^{(1)},\cdots$}};
			\draw  (-2.2,-2.1)node{{\footnotesize $\tilde{E}_1^{(1)},\cdots $}};		
			\draw  (3,-2) ellipse  (2 and 0.8);
			\draw[dashed]  (3,-1.2)-- (3,-2.8);
			\draw  (2.4,-1.3)[below]node{{\footnotesize \textbf{Type}2}};
			\draw  (3.6,-1.3)[below]node{{\footnotesize \textbf{Type}3}};
			\draw  (2,-2.1)node{{\footnotesize  $E_1^{(2)},\mathcal{E}_1^{(2)},\cdots$}};
			\draw  (4.05,-2.1)node{{\footnotesize $E_{1,<}^{(2)},E_{1,>}^{(2)}\ \cdots$}};
			\draw (3.03,-1.8 ) rectangle (4.5,-2.4);
			\draw  (-5.5,-4) ellipse  (0.7 and 0.4); \draw  (-3.8,-4) ellipse  (0.7 and 0.4); \draw  (-1.5,-4) ellipse  (0.7 and 0.4);  \draw  (1,-4) ellipse  (0.7 and 0.4);  \draw  (4,-4) ellipse  (1.2 and 0.5);   \draw  (5.4,-4) ellipse  (1.2 and 0.5);  
			\draw  (-5.5,-4)node{{\scriptsize $\mathcal{C}^{(1)} (E_1^{(1)})$}};		
			\draw  (-3.8,-4)node{{\scriptsize $\mathcal{C}^{(2)} (E_1^{(1)})$}};		
			\draw  (-4.2,-2.25)to (-5.4,-3.7);\draw  (-4.15,-2.25)to (-3.7,-3.7);		
			\draw (-1.5,-4)node{{\scriptsize$\mathcal{C}^{(1)} (\tilde{E}_1^{(1)})$}};
			\draw (-2.57,-2.28)to (-1.6,-3.7);
			\draw (1,-4)node{{\scriptsize$\mathcal{C}^{(1)} (E_1^{(2)})$}};
			\draw (1.25,-2.25)to (0.8,-3.7); \draw (2,-2.35)to (2.25,-3.9);
			\draw (2.25,-4)node{$\cdots$};
			\draw (3.5,-4)node{ {\tiny $\mathcal{C}^{(1)} (E_{1,<}^{(2)})$}};
			\draw (5.9,-4)node{ {\tiny $\mathcal{C}^{(1)} (E_{1,>}^{(2)})$}};
			\draw (4.7,-4)node{{\tiny\textbf{Type}3} };
			\draw (3.57,-2.32)to (3.8,-3.55);
			\draw (4.27,-2.32)to (5.5,-3.55);
			\draw (-5.7,-4.35)to (-6,-5.8)node[below]{$\cdots$}; \draw (-5.4,-4.35)to (-5.1,-5.8)node[below]{$\cdots$}; 
			\draw (-3.8,-4.35)to (-3.8,-5.8)node[below]{$\cdots$};	\draw (-1.5,-4.35)to (-1.5,-5.8)node[below]{$\cdots$};	
			\draw (1.1,-4.35)to (1.1,-5.8)node[below]{$\cdots$}; \draw (3.8,-4.5)to (3.5,-5.8)node[below]{$\cdots$};
			\draw (4.7,-4.3)to (4.7,-5.8)node[below]{$\cdots$}; \draw (5.6,-4.5)to (5.9,-5.8)node[below]{$\cdots$}; 
		\end{tikzpicture}
		\caption{After completing the construction, we have a tree  of Rellich functions $\{\mathcal{C}_n\}_{n\geq 0}$.  The superscript $ (j),j\in \{1,2\}$, of each child indicates the resonance  type relative to the parent.}
		\label{f1}
	\end{figure}
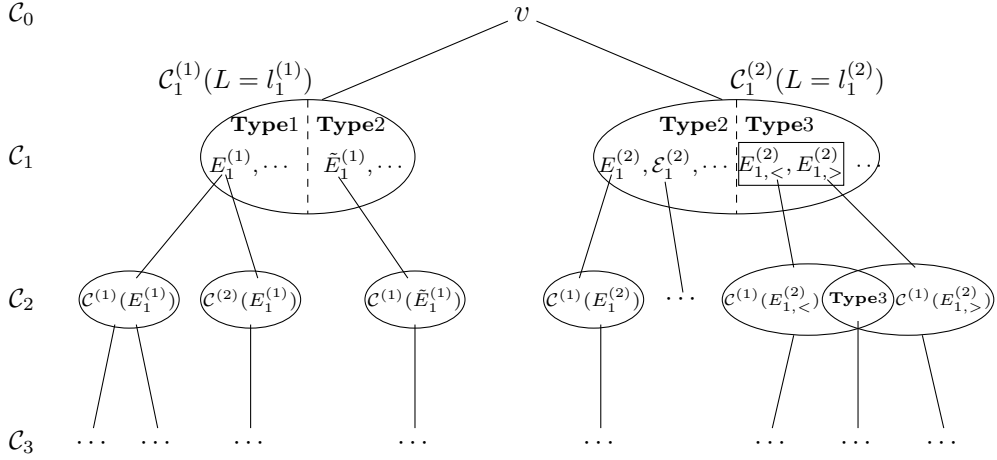
	\subsection{Organization of the paper}
	The paper is organized  as follows. We will introduce some basic notation and  foundational lemmas for use throughout the paper in \S \ref{pre}. The center arguments of this paper, constructing the Rellich functions and establishing Green's function estimates via a multi-scale induction analysis are   presented in \S \ref{multi}. In \S \ref{mainp}, we employ the results from \S \ref{multi} to finish the proof of the Main Theorem.
	\section{Preliminaries}\label{pre} 
	In this section, we explain the meaning of some notation  and  collect the foundational lemmas in this paper. The proofs of these lemmas can be found in Appendix. 
	
	\begin{rem}\label{ftnote}
		Without loss of  generality, we  assume that $\theta_M$ is the maxima point and $\theta_m$ the minima one of  $v$. Since we are considering  small $\varepsilon$, we  further assume that there exists some $a $ with   $0<a<\|\theta_M-\theta_m\|/10$, such that $|v''(\theta)|>3$ in $\{\theta\in \mathbb{T}:\ \|\theta-\theta_M\|<a\}\cup\{\theta\in \mathbb{T}:\ \|\theta-\theta_m\|<a\}$,  and $|v'(\theta)|>3$ in  $\{\theta\in \mathbb{T}:\ \|\theta-\theta_M\|\geq a\}\cap\{\theta\in \mathbb{T}:\ \|\theta-\theta_m\|\geq a\}$. Under these assumptions, we denote $$D=\sup_{\theta\in\T}\max(|v(\theta)|,|v'(\theta)|,|v''(\theta)|).$$ 
	\end{rem}

	We   denote by $\Lambda_L$  the set $\{x\in \Z^d:\ \|x\|_1\leq L\}$ for any $L>0$,  and   $A\lesssim B$ if $A\leq C (d,\tau,\gamma,D) B$, where  $C (d,\tau,\gamma,D)$ is an absolutely constant only depending on the dimension $d$, {Diophantine parameters} $\tau,\gamma$ and the $C^2$ norm of the sampling function $v$.
	
	We now list  some important lemmas, whose proofs   are appended.
	
	We begin with a lemma for finite-dimensional self-adjoint operators, which is useful to show the existence of  eigenvalue near a certain energy $E^*$.  
	\begin{lem}[Trial wave function]\label{ll} Let $H$ be a self-adjoint operator on a finite dimensional   Hilbert space  and $E^*\in \R$. If there exist $m$ normalized orthonormal functions $\psi_k\  (1\leq k\leq m )$ such that  $\| (H-E^*)\psi_k\|\leq\delta$ for some $\delta>0$ and all $1\leq k\leq m$, then $H$ has $m$ eigenvalues $ E_{k}\  (1\leq k\leq m )$ counted in multiplicity satisfying $\sum_{k=1}^{m} (E_k-E^*)^2\leq m\delta^2$. We call these  $\psi_k$   trial functions.
	\end{lem}
	
	As a corollary,  Lemma \ref{ll} immediately gives us 	  
	\begin{lem}\label{trialcor}
		If there  exists a trial function such that $\|\psi\|=1$ and   $\| (H-E^*)\psi\|\leq\delta$, then $H$ has at least one eigenvalue in $|E-E^*|\leq\delta.$ If there exist two orthogonal trial functions such that $\|\psi_1\|=\|\psi_2\|=1,\langle\psi_1,\psi_2\rangle=0$ and $\| (H-E^*)\psi_1\|\leq\delta,
		\| (H-E^*)\psi_2\|\leq\delta$,  then $H$ has at least two eigenvalues in $|E-E^*|\leq\sqrt{2}\delta$.
	\end{lem}
	Next we provide a lemma concerning  $C^2$ Morse functions. This  lemma,  not only together with Diophantine condition, provides   crucial separation of the resonant sites in our induction but also gives an upper bound estimate of  the measure of bad phases in proving localization. 
	\begin{lem}[Properties of Morse function]\label{C2}
		Let $E (\theta)\in C^2([a,b])$ with a critical point $\theta_s\in  (a,b)$.  Suppose  that there exists $\delta>0$ such that  $|E' (\theta)| \leq \delta$ implies $|E'' (\theta)| \geq 2$ with  a unique sign for these $\theta${,
			%			xiugai15
			that is, if we denote $X:=\{\theta\in [a,b]: \  |E' (\theta)| \leq \delta \}$, then either $\inf_{\theta\in X}E'' (\theta)\geq 2$ or $\sup_{\theta\in X}E'' (\theta)\leq -2$}. If $\theta_1$ and $\theta_2$ are both on the right of $\theta_s$ or both on the left of $\theta_s$,   then 
		$$
		\left|E\left(\theta_2\right)-E\left(\theta_1\right)\right| \geq \frac{1}{2} \left|\theta_2-\theta_1\right|^2
		$$
		provided $|\theta_1-\theta_2|\leq\delta$.
		Moreover,
		$$
		|E'(\theta)| \geq \min 
		(\delta,|\theta-\theta_s|) .
		$$\end{lem}
	
	The final lemma  that we recall is a classical computation of  derivatives of simple eigenvalue for $C^2$ one-parameter families of symmetric matrices.
	\begin{lem}[Feynman-Hellman formulas]\label{daoshu}
		Let  $H(\theta)$ be a family of  finite dimensional self-adjoint operators with $C^2$ parametrization. Assuming  $E(\theta^*)$ is a simple eigenvalue of $H(\theta^*)$ and $\psi(\theta^*)$ is its corresponding   eigenfunction, then $E(\theta),\psi(\theta)$ can be $C^2$ parameterized in a neighbourhood of $\theta^*$. Moreover, for $\theta$ belonging to  the neighborhood, we have:
		\begin{itemize}
			\item[\textbf{(1)}.] $E'=\langle\psi,H'\psi\rangle$.
			\item[\textbf{(2)}.] $E''=\langle\psi,H''\psi\rangle-2\langle H'\psi,G^\perp(E)H'\psi\rangle$, where $G^\perp(E)$ denotes the Green's function $(H-E)^{-1}$ on the orthogonal complement of $\psi$.
			\item[\textbf{(3)}.] Let  $\mathcal{E}\neq E$ is another simple eigenvalue and $\Psi$ its   eigenfunction, then  we have 
			$$\langle H'\psi,G^\perp(E)H'\psi\rangle=-\frac{\langle\Psi,H'\psi\rangle^2}{E-\mathcal{E}}+\langle H'\psi,G^{\perp\perp}(E)H'\psi\rangle$$
			where $G^{\perp\perp}(E)$ denotes the Green's function $(H-E)^{-1}$ on the orthogonal complement of the space spanned by $\psi$ and $\Psi$.
		\end{itemize} 
	\end{lem}
	\section{Multi-scale   analysis}\label{multi}
	In this section, we will inductively construct collections  of Rellich functions $\{\mathcal{C}_n\}_{n\geq0}$. Each element  $E_n\in \mathcal{C}_n$ is a locally defined Rellich function of the $n$-scale restriction operator $H_{B_n}$. The definition of  $E_n(\theta)\in \sigma(H_{B_n}(\theta))$ depends on a certain element  $E_{n-1}\in \mathcal{C}_{n-1}$, which is called the parent of $E_n$. %Let  $E_{n,i}\in\mathcal{E}_n $.  
	%Then $E_{n,i}$ maps  $I_{n,i}$ to  $J_{n,i}$ satisfying that   for any $\theta\in I_{n,i}$, $E_{n,i} (\theta)$ is  a certain eigenvalue of the Dirichlet restrictions $H_{\Lambda_n} (\theta)$. The eigenvalue  is chosen to be closed to   $E_{n-1,i'} (\theta)$ less than the $n$-th resonance distance  $\delta_n$, where  $E_{n-1,i'}\in \mathcal{E}_{n-1}$ is the ancestor of $E_{n,i}$.
	The main aims  of this section are  \begin{itemize}
		\item[\textbf{(1).}]   Construct the collections of Rellich function  $\{\mathcal{C}_n\}_{n\geq0}$ by a multi-scale induction argument.
		\item[\textbf{(2).}]  Verify that  the  Rellich functions inherit the Morse condition and two-monotonicity interval structure from the sampling function $v$,  thus we can employ the properties of   Morse functions from Lemma \ref{C2}.
		\item[\textbf{(3).}]   Describe  the resonance  (operator norm of Green's function) by  these Rellich functions and  establish Green's function estimates  for nonresonant sets with  regularity condition. 
	\end{itemize}   
	Let $\varepsilon_0>0$ be sufficiently small and $0<\varepsilon\leq \varepsilon_0$.  Define $ \delta_{0}:=\varepsilon_0^{\frac{1}{20}}$.
	%for $n=0$ and  $l_n^{(j)}=|\log\varepsilon_0|^{4^{n}j}, \ \delta_n^{(j)}=e^{- (l_n^{(j)})^{2/3}}$ for $n\geq1$, where $j=1$ if the resonance is simple and $j=2$ if the resonance is double.
	\subsection{The initial scale}
	The $0^{\text{th}}$ generation of  Rellich functions  is defined as $\mathcal{C}_0:=\{v\}$.   We first establish Green's function estimates on $0$-nonresonant  sets   by a Neumann series argument.
	\begin{thm}\label{0ge}
		Given  $\theta^*\in \T, E^*\in \R$, we define \begin{equation}\label{S0}
			S_0 (\theta^*,E^*)=\{x\in \Z^d:\ |v (\theta^*+x\cdot \omega)-E^*|<\delta_0\}.
		\end{equation} Let $\Lambda\subset \Z^d$ be a finite set.  If 	$\Lambda\cap S_0 (\theta^*,E^*)=\emptyset$,
		then  for  $|\theta-\theta^*|<\delta_0/ (10D)$, $|E-E^*|<\delta_0/5$ and  $\gamma_0=\frac{1}{2}|\log\varepsilon|,$
		\begin{align*}
			\|G_\Lambda (\theta,E)\|:=\|\left ( H_\Lambda (\theta)-E\right) ^{-1}\|&\leq10\delta_0^{-1}, \\
			|G_\Lambda (\theta,E;x,y)|:=|\left ( H_\Lambda (\theta)-E\right) ^{-1} (x,y)|&\leq e^{-\gamma_0\|x-y\|_1}, \ \|x-y\|_1\geq 1.  \end{align*}
	\end{thm}
	\begin{proof}
		For $x\in \Lambda$,  $|\theta-\theta^*|<\delta_0/ (10D)$ and  $|E-E^*|<\delta_0/5$, one has 
		$$|v (\theta+x\cdot \omega) -E|\geq |v (\theta^*+x\cdot \omega) -E^*|-|E-E^*|-D|\theta-\theta^*|>\delta_0/2. $$ Thus $V (\theta)$,  the diagonal term of $H_\Lambda (\theta)$, satisfies   $\| (V (\theta)-E)^{-1}\|\leq2\delta_0^{-1}$.  Since $\|\Delta\|\leq2d, \ 4d\varepsilon\delta_0^{-1}\leq\frac{1}{2}$,  we have by the Neumann series  argument  
		\begin{align*}
			\left(H_\Lambda (\theta)-E\right)^{-1} & =\left (\varepsilon \Delta+V (\theta)-E\right)^{-1} \\
			& =\sum_{l=0}^{\infty} (-1)^l \varepsilon^l\left[\left (V (\theta)-E\right)^{-1} \Delta\right]^l\left (V (\theta)-E\right)^{-1}.
		\end{align*}
		Thus
		$$\|\left(H_\Lambda (\theta)-E\right)^{-1}\|\leq2\|\left (V (\theta)-E\right)^{-1}\|\leq4\delta_0^{-1}$$
		and 
		\begin{align*}
			\left|\left(H_\Lambda (\theta)-E\right)^{-1} (x, y)\right| & \leq \frac{4}{\delta_0}\left (\frac{4 d\varepsilon}{\delta_0}\right)^{\|x-y\|_1}\leq\sqrt{\varepsilon}^{\|x-y\|_1}=e^{-\gamma_0\|x-y\|_1}
		\end{align*} provided $\|x-y\|_1\geq1$.
	\end{proof}
	\subsection{The first scale}\label{n=1} In this section, we will construct the first  generation of   Rellich functions $\mathcal{C}_1=\mathcal{C}_1^{(1)}\cup\mathcal{C}_1^{(2)}$ to describe the resonance  at this scale,  where the superscript ``$ (1)$'' indicates the simple resonance case and ``$ (2)$''   the  double resonance one. %, \ \delta_1=e^{-l_1^{2/3}}%
	
	For this purpose, the first step is to  divide the interval $J (v):=\operatorname{Image}v$ (i.e., the range of $v$) into simple resonant intervals $\mathcal{J}^{(1)}$  and    double  resonant intervals  $\mathcal{J}^{(2)}$. By the assumption of $v$, there are two closed intervals $I_\pm$ with disjoint interiors, such that 
	$$v_\pm:=v|_{I\pm} ,\ \  \pm v'_\pm\geq0, \ \ v_\pm (I_\pm)=J (v). $$
	Define $$l_1^{(1)}:=|\log\varepsilon_0|^{4},\  l_1^{(2)}:= (l_1^{(1)})^2\gg 10l_1^{(1)}.$$
	Since $v$ is monotonic on each $I_\pm$, for $0<\|k\|_1\leq 10 l_1^{(1)}$, there is a unique point $\theta_{k,-}\in I_-$ and its relevant translate $\theta_{k,+}:=\theta_{k,-}+k\cdot \omega$,  such that 
	$$e_{k}:=v (\theta_{k,-})=v (\theta_{k,+}).$$
	To construct the double resonant intervals $\mathcal{J}^{(2)}$, we need the following lemmas which give the lower bound of $|e_k-e_{k'}|$ for $k\neq k'$  and $|e_k-e_c|$, where $e_c$ is a critical value of $v$.
	\begin{lem}\label{a}If   $\theta_1,\theta_2$ belong to the same monotonicity interval $I_\pm$,  then 
		\begin{equation}\label{wh}
			|v (\theta_1)-v (\theta_2)|\geq \|\theta_1-\theta_2\|^2.
		\end{equation}
	\end{lem}
	\begin{proof}[Proof of Lemma \ref{a}]
		Without loss of  generality, we assume $\theta_1,\theta_2\in I_+$ and  $\theta_1<\theta_2$. By our assumption,    $v$ is strictly increasing  on $I_+=[\theta_m,\theta_M]$ and satisfies   $v' (\theta)>2$ for $\theta\in[\theta_m+a,\theta_M-a]$ and  $v'' (\theta)>2$  (resp. $<2$) for $\theta\in [\theta_m,\theta_m+a]$  (resp. $[\theta_M-a,\theta_M]$).\\
		{\it Case }1: $\theta_m\leq\theta_1<\theta_2\leq \theta_m+a$. We have  in this case $$v (\theta_2)-v (\theta_1)=v' (\theta_1) (\theta_2-\theta_1)+\frac{1}{2}v'' (\xi) (\theta_2-\theta_1)^2\geq (\theta_2-\theta_1)^2.$$\\
		{\it Case }2: $\theta_m\leq\theta_1\leq\theta_m+ a\leq\theta_2\leq\theta_M-a$. We have in this case
		\begin{align*}
			v (\theta_2)-v (\theta_1)&= (v (\theta_2)-v (a))+ (v (a)-v (\theta_1))\\
			&\geq2 (\theta_2-a)+ (a-\theta_1)^2\\&\geq (\theta_2-\theta_1)^2.
		\end{align*}\\
		{\it Case} 3: $\theta_m+ a\leq\theta_1<\theta_2\leq\theta_M-a$. We have 
		$$	v (\theta_2)-v (\theta_1)=v' (\xi) (\theta_2-\theta_1)\geq (\theta_2-\theta_1)^2.$$\\
		{\it Case} 4: $\theta_m<\theta_1\leq \theta_m+ a<\theta_M-a\leq\theta_2$. We have 
		$$	v (\theta_2)-v (\theta_1)\geq v (\theta_M-a)-v (\theta_m+ a)\geq2 (|I_+|-2a)\geq (\theta_2-\theta_1)^2.$$\\
		{\it Case} 5: $\theta_M-a\leq\theta_1<\theta_2\leq\theta_M$. This case  is similar to {\it Case} 1.\\
		{\it Case} 6: $\theta_m+a\leq\theta_1\leq \theta_M-a\leq\theta_2\leq\theta_M$. This case  is similar to {\it Case} 2.
	\end{proof}
	\begin{lem}\label{sep}
		For $0<\|k\|_1, \|k'\|_1\leq 10 l_1^{(1)}$,  $k\neq k'$, we have 
		\begin{equation}\label{16}
			|e_k-e_{k'}|\gtrsim{(l_1^{(1)})^{-2\tau}}\gg\delta_{0}^{\frac{1}{100}},
		\end{equation}
		\begin{equation}\label{160}
			|e_k-e_{c}|\gtrsim{(l_1^{(1)})^{-2\tau}}\gg\delta_{0}^{\frac{1}{100}}.
		\end{equation}
		
	\end{lem}
	
	\begin{proof}
		First we prove \eqref{16}.	By triangle inequality and Diophantine condition, we have $$\max (\|\theta_{k,-}-\theta_{k',-}\|,\|\theta_{k,+}-\theta_{k',+}\|)\geq\| (k-k')\cdot \omega\|/2\gtrsim\|k-k'\|^{-\tau}\gtrsim  (l_1^{(1)})^{-\tau}.$$
		If $\|\theta_{k,-}-\theta_{k',-}\|\gtrsim  (l_1^{(1)})^{-\tau}$ holds, by Lemma \ref{a} we get 	$$|e_k-e_{k'}|=|v (\theta_{k,-})-v (\theta_{k',-})|\geq\|\theta_{k,-}-\theta_{k',-}\|^2\gtrsim  (l_1^{(1)})^{-2\tau}.$$ 
		Likewise, if $\|\theta_{k,+}-\theta_{k',+}\|\gtrsim  (l_1^{(1)})^{-\tau}$  holds, we get $$|e_k-e_{k'}|=|v (\theta_{k,+})-v (\theta_{k',+})|\geq\|\theta_{k,+}-\theta_{k',+}\|^2\gtrsim  (l_1^{(1)})^{-2\tau}.$$ 
		To see \eqref{160}, without loss of  generality, we assume $e_c=v (\theta_m)$ is the minima of $v$. As above, we have $$\max (\|\theta_{k,-}-\theta_m\|,\|\theta_{k,+}-\theta_m\|)\geq\| (k-k')\cdot \omega\|/2\gtrsim\|k-k'\|^{-\tau}\gtrsim  (l_1^{(1)})^{-\tau}.$$
		Thus 
		$$|e_k-e_{c}|\geq\max (\|\theta_{k,-}-\theta_m\|,\|\theta_{k,+}-\theta_m\|)^2\gtrsim  (l_1^{(1)})^{-2\tau}.$$
	\end{proof}
	\begin{figure}[htp]
		\begin{tikzpicture}[>=latex, scale=1]
			\draw (0,0)to (3*pi,0);
			\draw[->]  (0,-1)-- (0,6)node[right]{$y$};
			\draw  (0,0);
			\draw[domain=0:2*pi, samples=100] plot  ({\x},{2*cos (\x/2 r) +3});
			\draw[domain=2*pi:3*pi, samples=100] plot  ({\x},{3-2*cos ((\x-2*pi) r) });
			\draw[dashed]  (2*pi,0)-- (2*pi,1);
			\draw[<->] (0,0)-- (2*pi,0); \draw[<->] (2*pi,0)-- (3*pi,0);
			\draw (pi,0)node[below]{$I_-$};\draw (2.5*pi,0)node[below]{$I_+$};
			\draw[dashed]  (0.45*pi,0)-- (0.45*pi,4.5); 	\draw[dashed]  (0,4.5)-- (2.775*pi,4.5);
			\draw[dashed]  (0.68*pi,0)-- (0.68*pi,4);	\draw[dashed]  (0,4)-- (2.66*pi,4); \draw[dashed]  (2.66*pi,0)-- (2.66*pi,4); \draw[dashed]  (2.775*pi,0)-- (2.775*pi,4.5);
			\draw[thick,red] (0.45*pi,0)-- (0.68*pi,0); 	\draw[thick,red] (2.66*pi,0)-- (2.775*pi,0);
			\draw (0.56*pi,0)node[below]{\color{red}{\scriptsize $I_{k,-}^{(2)}$ }};
			\draw (2.715*pi,0)node[below]{\color{red}{\scriptsize $I_{k,+}^{(2)}$ }};
			\fill  (0.57*pi,4.25)circle (1pt); 	\fill  (2.715*pi,4.25)circle (1pt);
			\draw[dashed]  (0.57*pi,4.25)-- node[fill=white]{$k\cdot \omega $} (2.715*pi,4.25);
			\draw[dashed,<-]  (0.57*pi,4.25)--  (0.7*pi,5); \draw (0.67*pi,5.1)node[right]{{ $ (\theta_{k,-},e_k)$}};
			\draw[thick,red] (0,4)-- (0,4.5); 	\draw (0,4.25)node[left]{\color{red}{\tiny $J_k^{(2)}$}}; 
			\draw[dashed]  (1.25*pi,2.25)-- node[fill=white]{$k'\cdot \omega $} (2.375*pi,2.25);
			\fill  (1.25*pi,2.25)circle (1pt); 	\fill  (2.375*pi,2.25)circle (1pt); 
			\draw[dashed,<-]  (1.25*pi,2.25)--  (1.37*pi,3); \draw (1.34*pi,3)node[right]{{ $ (\theta_{k',-},e_{k'})$}};		\draw (1.5*pi,1.3)node{$v$}; 
		\end{tikzpicture}
		\caption{Separation of the double resonant intervals and critical values.}
	\end{figure}
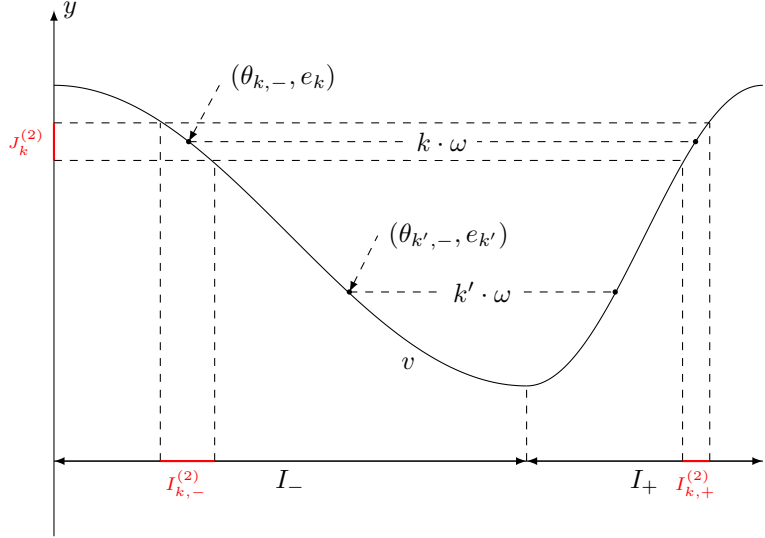
	Let $B_{\delta}(x_0)=\{x\in\R: |x-x_0|<\delta\}$ and let $\bar B_{\delta}(x_0)$ be its closure. Define the double resonant intervals $$J^{(2)}_k:=\bar{B}_{\delta_{0}^{\frac{1}{100}}} (e_k), \ 0<\|k\|_1\leq10l_1^{(1)}$$
	and   the complementary region 
	\begin{equation}\label{1SR}
		J^{SR}:=J\setminus \bigcup_{0<\|k\|\leq10l_1^{(1)}}B_{\delta_{0}^{\frac{1}{100}}-3\delta_0} (e_k).
	\end{equation}
	\begin{rem}\label{se}
		The separation from  Lemma \ref{sep} ensures that   $J^{(2)}_k\cap J^{(2)}_{k'}=\emptyset$ for $k\neq k'$ and the length of  each  connected component of $J^{SR}$ exceeds $\delta_{0}^{\frac{1}{100}}$.
	\end{rem}		
	From this observation, we  deduce
	\begin{prop}\label{coverpr}
		$J (v)$ can be covered by closed intervals $J_i^{(1)}$ and $J^{(2)}_k$, such that 
		\begin{itemize}
			\item[\textbf{(1)}.]  $\delta_{0}^{\frac{1}{100}}\leq |J_i^{(1)}|,  |J^{(2)}_k|\leq 2\delta_{0}^{\frac{1}{100}}$. 
			\item[\textbf{(2)}.] $J_i^{(1)}\subset J^{SR}.$
			\item[\textbf{(3)}.] $J_i^{(1)}\cap J_{i'}^{(1)}\neq\emptyset\Rightarrow |J_i^{(1)}\cap J_{i'}^{(1)}|=3\delta_0.$
			\item[\textbf{(4)}.] $ J_i^{(1)}\cap J_{k}^{(2)}\neq\emptyset\Rightarrow |J_i^{(1)}\cap J_{k}^{(2)}|=3\delta_0. $
			\item[\textbf{(5)}.] The total number of such intervals does not exceed $2|J (v)|\delta_0^{-\frac{1}{100}}$.
		\end{itemize}
	\end{prop}
	\begin{proof}
		By Remark \ref{se}, we can construct a collection of closed intervals $J_i^{(1)}$ of size between $\delta_{0}^{\frac{1}{100}}$ and $2\delta_{0}^{\frac{1}{100}}$ overlapping with only their nearest neighbors by exactly $3\delta_0$ covering $J^{SR}$. Together with the collection of  intervals	$J_{k}^{(2)} (0<\|k\|_1\leq10l_1^{(1)})$, we have a total of at most  $2|J (v)|\delta_0^{-\frac{1}{100}}$ intervals satisfying the conditions laid
		out in the proposition. 
	\end{proof}
	We denote by  $\mathcal{J}^{(1)}$  the collection of  simple resonant intervals $J_i^{(1)}$   and  $\mathcal{J}^{(2)}$ the collection of  double  resonant intervals $J_k^{(2)}$. The corresponding preimages of $J^{(j)}_{\bullet}$ are 
	$$I^{(j)}_{\bullet,\pm}:=v_\pm^{-1} (J^{(j)}_{\bullet}), \ I^{(j)}_{\bullet}:=I^{(j)}_{\bullet,-}\cup I^{ (j)}_{\bullet,+},  $$ where  $j\in \{1,2\}$ and $\bullet\in\{i,k\}.$ 
	
	By Lemma \ref{a} and the uniform bound $|v'|\leq D$, we have 
	$$|I^{(j)}_{\bullet,\pm}|^2\leq |v_\pm (I^{(j)}_{\bullet,\pm})|\leq D|I^{(j)}_{\bullet,\pm}|.$$ 
	Since $|v_\pm (I^{(j)}_{\bullet,\pm})|=|J^{(j)}_{\bullet}|\in [\delta_{0}^{\frac{1}{100}},2\delta_{0}^{\frac{1}{100}}]$, 
	we deduce  \begin{equation}\label{I1len}
		\frac{1}{D}	\delta_{0}^{\frac{1}{100}}\leq  |I^{(j)}_{\bullet,\pm}|\leq2\delta_{0}^{\frac{1}{200}}.
	\end{equation}
	
	The remainder of this section  is devoted to constructing  the $1^{{\rm st}}$ generation of  Rellich functions   $E^{(j)}_{1}$. We do so by case analysis.
	%		 Moreover the  collection satisfies the following properties:
	%	 \begin{itemize}
		%		\item Let  $E^{(j)}_{1,i}\in \mathcal{E}_1^{(j)}$. Then $E^{(j)}_{1,i}$ maps  $I^{(j)}_{1,i}$ to  $J^{(j)}_{1,i}$. For $\theta\in I^{(j)}_{1,i}$,   $E^{(j)}_{1,i} (\theta)$  is a certain eigenvalue of  $H_{\Lambda_1^{(j)}} (\theta)$ satisfying $|E^{(j)}_{1,i} (\theta)-v (\theta)|<\delta_{0}$, where $\Lambda_1^{(j)}=[-l_1^{(j)},l_1^{(j)}]^d$.   
		%		\item 
		%	\end{itemize}
	\subsubsection{The simple resonance case}
	In this section, we fix an interval $J_i^{(1)}\in \mathcal{J}^{(1)}$ and its preimage $I_i^{(1)}$.
	\begin{prop}\label{SRpro}
		If $\theta\in I^{(1)}_i$, then for any $0<\|x\|_1\leq l_1^{(1)}$, we have $$|v (\theta)-v (\theta+x\cdot \omega)|\geq 3\delta_0.$$
	\end{prop}
	\begin{proof}
		Suppose $\theta\in I_{i,-}^{(1)}$ (the case $\theta\in I_{i,+}^{(1)}$ is completely analogous). If $\theta+x\cdot \omega\in I_-$,  since $0<\|x\|_1\leq l_1^{(1)}$, we can apply Lemma \ref{a} to get $$|v (\theta)-v (\theta+x\cdot \omega)|\geq \|x\cdot \omega\|^2\gtrsim  (l_1^{(1)})^{-2\tau}\gg 3\delta_0.$$
		
		We thus suppose $\theta+x\cdot \omega\in I_+$.
		
		Assume $\theta\leq \theta_{x,-}$. Then  $\theta+x \cdot \omega \leq \theta_{x,+}$. By the monotonicity of $v$ on $I_{\pm}$, $v (\theta)\geq e_{x} $ and $v (\theta+x\cdot \omega )\leq e_{x} $. Thus by \eqref{1SR} and item \textbf{(2)} of  Proposition \ref{coverpr},  
		$$v (\theta)-v (\theta+x\cdot \omega)\geq v (\theta)-e_x> \delta_{0}^{\frac{1}{100}}-3\delta_0  >3\delta_0.$$
		If $\theta\geq \theta_{x,-}$, then similarly $v (\theta)\leq e_{x} $ and   $v (\theta+x\cdot \omega )\geq e_{x}$. Thus  	$$|v (\theta)-v (\theta+x\cdot \omega)|=v (\theta+x\cdot \omega)-v (\theta)\geq e_x -v (\theta)  >3\delta_0.$$
	\end{proof}
	\begin{rem}
		From Proposition \ref{SRpro} it follows that for $\theta\in I^{(1)}_i,E=v (\theta)$, the origin  $o\in S_0 (\theta,E)$ is the only resonant site at $l_1^{(1)}$-scale. 
	\end{rem}
	For the simple resonance case, we define  $B_1^{(1)}:=\Lambda_{l_1^{(1)}}$. The following proposition shows that the restriction operator  $H_{B_1^{(1)}} (\theta)$ has   a unique  Rellich function (of  order $\delta_0$),  $E_1^{(1)} (\theta)$ defined on $ I^{ (1)}_i$ near $v (\theta)$. 
	\begin{prop}\label{k1}
		For $\theta\in I^{(1)}_i$, we have 
		\begin{itemize}
			\item[\textbf{(a).}]  $H_{B_1^{(1)}} (\theta)$ has a unique eigenvalue  $E_1^{(1)} (\theta)$ such that $|E_1^{(1)} (\theta)-v (\theta)|\lesssim\varepsilon$. Moreover, any other $\hat{E}\in\sigma (H_{B_1^{(1)}} (\theta)) $ must obey $|\hat{E}-v (\theta)|\geq2\delta_0$.
			\item[\textbf{(b).}] The  eigenfunction $\psi_1$ corresponding to  $E_1^{(1)}$ satisfies  
			\begin{equation}\label{decay1}
				|\psi_1 (x)|\leq e^{-\gamma_0\|x\|_1} .
			\end{equation}
			\item[\textbf{(c).}] $\|G_1^\perp (E_1^{(1)})\|\leq \delta_0^{-1}$, where $G_1^\perp$ is the Green's function for $B_1^{(1)}$ on the orthogonal complement of $\psi_1$.
		\end{itemize} 
	\end{prop}
	\begin{proof}
		The existence of $E_1^{(1)} (\theta)$ follows from Lemma 	 \ref{trialcor} and the following observation {
			%			xiugai14
			$$\| (H_{B_1^{(1)}} (\theta)-v (\theta)){\bm e}_{o}\|=\|\sum_{\|y\|_1=1}\varepsilon {\bm e}_y\|\lesssim\varepsilon.$$}
		Denote $\Lambda=B_1^{(1)}\setminus\{o\}$. Let  $|E-v (\theta)|\leq 2\delta_0$.  By Proposition \ref{SRpro}, 
		we have  $\Lambda\cap 
		S_{0} (\theta,E)=\emptyset$. 	Thus by Theorem \ref{0ge}, $$|G_{\Lambda} (\theta,E; x, y)|\leq\delta_0^{-1}e^{-\gamma_0\|x-y\|_1}.$$  
		Let $E\in \sigma (H_{B_1^{(1)}} (\theta))$ such that $|E-v (\theta)|\leq 2\delta_0$.  We determine the value of $\psi_1 (x)$ for $x\neq o$  by 
		$$\psi_1 (x)= \sum_{\|y\|_1=1} G_{\Lambda} (\theta,E; x, y) \Gamma_{y,o} \psi_1\left(o\right).$$
		Thus \begin{equation}\label{g}
			|\psi_1 (x)|\lesssim\frac{\varepsilon}{\delta_0}e^{-\gamma_0 (\|x\|_1-1)}\leq e^{-\gamma_0\|x\|_1}
		\end{equation} and we finish the proof of   \textbf{(b)}.
		
		If there is another $\hat{E}\in \sigma (H_{B_1^{(1)}} (\theta))$ satisfying $|\hat{E}-v (\theta)|\leq2\delta_0$, its   eigenfunction $\hat{\psi}$ also satisfies \eqref{decay1}, which violates the orthogonality of $\psi_1 $ and $\hat{\psi}$. Thus we prove the uniqueness part of \textbf{(a)}. The item \textbf{(c)} follows from any other $E_1^{(1)} (\theta)\neq\hat{E}\in \sigma (H_{B_1^{(1)}} (\theta))$ must obey $$|\hat{E}-E_1^{(1)} (\theta)|\geq |\hat{E}-v (\theta)|-|v (\theta)-E_1^{(1)} (\theta)|\geq2\delta_0-O (\varepsilon)\geq\delta_0.$$
	\end{proof}
	%Next proposition shows that $E_1^{(1)} (\theta)$ is a perturbation of $E_0^i (\theta)$ for $|\theta-\theta^\bullet|<\delta_0/ (10M_1)$.
	\begin{prop}\label{ap} We have for $\theta\in I_i^{(1)},$   
		$$|\frac{d^s}{d\theta^s} (E_1^{(1)} (\theta)-v (\theta))|\lesssim\frac{\varepsilon}{\delta_0^s}\ll\delta_0^{10} \ {\rm for\ }s=0,1,2.$$
	\end{prop}
	\begin{proof}
		By the previous proposition, we have 
		\begin{equation}\label{C0}
			|E_1^{(1)} (\theta)-v (\theta)|\lesssim\varepsilon.
		\end{equation} Recalling \eqref{g}, we have \begin{equation}\label{d1}
			\|\psi_1-\psi_0\|\lesssim\frac{\varepsilon}{\delta_0}.
		\end{equation}
		For $s=1,2$,  we use the Feynman-Hellman formulas \textbf{(1)} and \textbf{(2)} in Lemma \ref{daoshu}. Thus  
		\begin{align*}|(E^{(1)}_1)' (\theta)- v' (\theta)|&=|\left\langle\psi_1,V' \psi_1\right\rangle-\left\langle\psi_0,V' \psi_0\right\rangle|\lesssim\frac{\varepsilon}{\delta_0},\\
			(E^{(1)}_1)'' (\theta)&=\left\langle\psi_1, V^{\prime \prime} \psi_1\right\rangle-2\left\langle\psi_1, V' G_1^{\perp} (E_1) V' \psi_1\right\rangle. 
		\end{align*}
		Since \eqref{d1}, we can write $$V' \psi_1=A\psi_1+O (\frac{\varepsilon}{\delta_0}).$$ Thus we get $$\|P_1^\perp  (V'\psi_1)\|\lesssim \frac{\varepsilon}{\delta_0},$$ where $P_1^\perp$ denotes projection  onto  the orthogonal complement  of  $\psi_1$.   Thus 
		\begin{align*}
			|  (E^{(1)}_1)'' (\theta)- v'' (\theta)|&=|\left\langle\psi_1, V^{\prime \prime} \psi_1\right\rangle-\left\langle\psi_0, V^{\prime \prime} \psi_0\right\rangle-2\left\langle\psi_1, V' G_1^{\perp} (E_1^{(1)}) V' \psi_1\right\rangle|\\
			&\lesssim  \frac{\varepsilon}{\delta_0} +\|G_1^\perp (E_1^{(1)})\| \cdot\|P_1^\perp  (V'\psi_1)\|\\
			&\lesssim \frac{\varepsilon}{\delta_0^2},
		\end{align*}
		where we have  used the estimate from Proposition \ref{k1} \textbf{(c)} to bound the term $\|G_1^\perp (E_1^{(1)})\|$.
	\end{proof}
	The next proposition verifies the Morse condition of $E_1^{(1)} (\theta)$.
	\begin{prop}\label{31}
		Assume $\theta\in I_{i}^{(1)}$ such that $| (E^{(1)}_1)' (\theta)|\leq\delta_0^{\frac{1}{1000}}$. Then $| (E^{(1)}_1)'' (\theta)|\geq3-\delta_0^{10}>2$ with a unique sign for all such $\theta$. 
	\end{prop}
	\begin{proof}
		If $| (E^{(1)}_1)' (\theta)|\leq\delta_0^{\frac{1}{1000}}$, then  by Proposition \ref{ap}, we have 
		$$|v' (\theta)|\leq| (E^{(1)}_1)' (\theta)|+\delta_0^{10}\leq 2\delta_0^{\frac{1}{1000}}.$$
		By  our assumption of $v$  (c.f. Remark \ref{ftnote}) and since  $\delta_0^{\frac{1}{1000}}\ll a$, we deduce that 
		$$|v'' (\theta)|\geq3$$ with a unique sign for all such $\theta$. Thus by Proposition \ref{ap}, $$| (E^{(1)}_1)'' (\theta)|\geq|v'' (\theta)|-\delta_0^{10}\geq3-\delta_0^{10}>2$$ with a unique sign for all such $\theta$.
	\end{proof}
	At the end of this part, we show each Rellich  child $E_1^{(1)}$ defined on $I^{(1)}_i$ inherits the two-monotonicity interval structure. {
		%		xiugai16
		The two-monotonicity interval structure  is crucial in the whole induction since it ensures that there are at most two resonant Rellich functions at
		each scale.}  For this purpose, we need a lemma concerning the derivative of $v$.  
	\begin{lem}\label{derv}
		For any $\theta\in \T$, we have 
		$$|v' (\theta)|\geq\min (\|\theta-\theta_m\|,\|\theta-\theta_M\|).$$
	\end{lem}
	\begin{proof}
		By the assumption of $v$,  
		if $\theta\in\{\theta\in \mathbb{T}:\ \|\theta-\theta_M\|\geq a\}\cap\{\theta\in \mathbb{T}:\ \|\theta-\theta_m\|\geq a\}$, we have 
		$$|v' (\theta)|\geq3\geq\min (\|\theta-\theta_m\|,\|\theta-\theta_M\|).$$
		If $\|\theta-\theta_M\|<a$, by the mean value theorem,  we have 
		$$|v' (\theta)|\geq|v'' (\xi)||\theta-\theta_M|\geq\|\theta-\theta_M\|.$$
		Likewise, if $\|\theta-\theta_m\|<a$,  then 
		$$|v' (\theta)|\geq|v'' (\xi)||\theta-\theta_m|\geq\|\theta-\theta_m\|.$$
	\end{proof}
	\begin{prop}\label{str1} We have the following:
		\begin{itemize}
			\item [\textbf{(a).}] Suppose   $v$ has a critical point in $I^{(1)}_i$.  (Thus $I^{(1)}_i$ is a closed interval.) Then $E_1^{(1)}$ also has a critical point in $I^{(1)}_i$. Moreover, $E_1^{ (1)}$ has  the two-monotonicity interval structure, that is,   $I^{(1)}_i$ can be divided into two closed intervals $\tilde{I}^{ (1)}_{i,+}$  and  $\tilde{I}^{ (1)}_{i,-}$ with disjoint interiors, such that 
			$$\pm  (E^{(1)}_1)'|_{\tilde{I}^{(1)}_{i,\pm}}\geq0, \ \tilde{I}^{(1)}_{i,+}\cup\tilde{I}^{(1)}_{i,-}=I^{(1)}_i.$$
			%	Moreover, $$| (E^{(1)}_1)'|_{\partial I^{(1)}_i}|\gtrsim\delta_0^{\frac{1}{100}}$$ with different sign on the edges of $I^{(1)}_i$.
			\item [\textbf{(b).}] Suppose   $v$ has no  critical point in $I^{(1)}_i$  (Thus  $I^{(1)}_i=I^{(1)}_{i,+}\cup I^{ (1)}_{i,-}$ is a union of two closed disjoint intervals.) Then 
			$$\pm  (E^{(1)}_1)'|_{I^{(1)}_{i,\pm}}\gtrsim \delta_0^{\frac{1}{100}}.$$
		\end{itemize}
	\end{prop}
	\begin{proof}
		To prove \textbf{(a)}, assume that  $\theta_c$ is   the critical point of $v$ in $I^{(1)}_i$. By \eqref{I1len}, Lemma \ref{derv} and Proposition \ref{ap}, we deduce  $| (E^{(1)}_1)'|_{\partial I^{(1)}_i}|\gtrsim\delta_0^{\frac{1}{100}}$ with different sign on the edges of $I^{(1)}_i$. Since $ (E^{(1)}_1)'$ is continuous and  $I^{(1)}_i$ is connected, there is point $\theta_s\in I^{(1)}_i$ such that $ (E^{(1)}_1)' (\theta_s)=0$.  The two monotonicity structure immediately follows from Proposition \ref{31} and Lemma \ref{C2}. { 
			%			xiugai2
			The uniqueness of $\theta_s$ on $I_i^{(1)}$ follows from the two monotonicity structure of $E_1^{(1)}$ restricted to  $I_i^{(1)}$.}
		
		To prove \textbf{(b)}, we  notice that by the construction (c.f. Proposition \ref{coverpr}), if $J_i^{(1)}$ does not contain a critical value of $v$, then $J_i^{(1)}$ has a  distance exceeding $\delta_0^{\frac{1}{100}}-3\delta_0$  from the critical value. Thus its preimage $I_i^{(1)}$  has a distance on order $\delta_0^{\frac{1}{100}}$  from the critical point. By  Lemma \ref{derv}, we get $$ \pm v|_{I^{ (1)}_{i,\pm}}\gtrsim \delta_0^{\frac{1}{100}}.$$   The desired estimate for  $\pm  (E^{(1)}_1)'|_{I^{(1)}_{i,\pm}}$ follows from the above  estimate and Proposition \ref{ap} immediately.
	\end{proof}
	We finish the proof in the simple resonance case and consider the double resonance case in the next part.
	\subsubsection{The double resonance case} In this section, we fix $J_k^{(2)}\in \mathcal{J}^{(2)} \ (0<\|k\|_1\leq10l_1^{(1)})$ and define the domain 
	$$ I^{(2)}_{k,\cup}:= I^{(2)}_{k,-}\cup  (I^{(2)}_{k,+}-k\cdot \omega), $$
	which is a single interval with  $\delta_0^{\frac{1}{100}}\lesssim|I^{(2)}_{k,\cup}|\leq 4\delta_0^{\frac{1}{200}}$ and $\operatorname{dist}(k,\partial I^{(2)}_{k,\cup})\gtrsim \delta_0^{\frac{1}{100}}$. 
	\begin{prop}\label{DRpro}
		If $\theta\in I^{(2)}_{k,\cup}$, then for any $\|x\|_1\leq  (l_1^{(1)})^4$, $x\notin\{o,k\}$,  we have 
		$$|v (\theta)-v (\theta+x\cdot \omega)|\geq 2\delta_0^{\frac{1}{400}}.$$
	\end{prop}
	\begin{proof}
		We first consider the case that $\theta^*=\theta_{k,-}$. If $\theta^*+x\cdot \omega\in I_-$, then by Lemma \ref{wh} and Diophantine condition, we have 
		$$|v (\theta^*)-v (\theta^*+x\cdot \omega)|\geq \|x\cdot \omega\|^2 \gtrsim  (l_1^{(1)})^{-8\tau}\gg \delta_0^{\frac{1}{400}}.$$
		If instead $\theta^*+x\cdot \omega\in I_+$, then 
		$$|v (\theta^*)-v (\theta^*+x\cdot \omega)|=|v (\theta^*+k\cdot \omega)-v (\theta^*+x\cdot \omega)|\geq\| (k-x)\cdot \omega\|^2\gtrsim  (l_1^{(1)})^{-8\tau} \gg\delta_0^{\frac{1}{400}}.$$
		The claim for $\theta\in I^{(2)}_{k,\cup}$ follows from the mean value theorem, \eqref{I1len} and the uniform bound $|v'|\leq D$.
	\end{proof}
	\begin{rem}
		Proposition \ref{DRpro} ensures that for $\theta\in I^{(2)}_{k,\cup},E=v (\theta)$, the origin  $o$ and $k$  are the only resonant sites at $ (l_1^{(1)})^4$-scale. 
	\end{rem}
	For the double resonance case, we define  $B_1^{(2)}:=\Lambda_{l_1^{(2)}}$. The following proposition shows that the restriction operator  $H_{B_1^{(2)}} (\theta)$ has only  two  Rellich functions (of  order $\delta_0^{\frac{1}{400}}$),  $E_1^{(2)} (\theta)$ and $\mathcal{E}_1^{(2)} (\theta)$ defined on $I^{(2)}_{k,\cup}$  near $v (\theta)$.
	\begin{prop}\label{k2}
		The following holds  for $\theta\in I^{(2)}_{k,\cup}$:
		\begin{itemize}
			\item[\textbf{(a).}]  $H_{B_1^{(2)}} (\theta)$ has exactly two  eigenvalues $E_1^{(2)} (\theta)$ and $\mathcal{E}_1^{(2)} (\theta)$ in  $[v (\theta)-5D\delta_0^{\frac{1}{200}},v (\theta)+5D\delta_0^{\frac{1}{200}}]$. Moreover, any other $\hat{E}\in\sigma (H_{B_1^{(2)}} (\theta)) $ must obey $|\hat{E}-v (\theta)|\geq2\delta_0^{\frac{1}{400}}$.
			\item [\textbf{(b).}] The    eigenfunction $\psi_1$\ {\rm (resp. $\Psi_1$)} corresponding to $E_1^{(2)}$ {\rm  (resp. $\mathcal{E}_1^{(2)}$)}  decays exponentially fast away from $o$ and $k$, i.e.,
			$$|\psi_1 (x)|\leq e^{-\gamma_0\|x\|_1}+e^{-\gamma_0\|x-k\|_1},$$
			$$|\Psi_1 (x)|\leq e^{-\gamma_0\|x\|_1}+e^{-\gamma_0\|x-k\|_1}.$$
			Moreover, the two eigenfunctions can be expressed as 
			\begin{align}\label{y}
				\begin{split}
					\psi_1=A\delta (x-o)+B\delta (x-k)+O (\delta_0^{10}),\\
					\Psi_1=B\delta (x-o)-A\delta (x-k)+O (\delta_0^{10}),
				\end{split}	
			\end{align}
			where $A,B$ depend on $\theta$ satisfying $A^2+B^2=1$.
			\item [\textbf{(c).}] $\|G_1^{\perp\perp} (E_1^{(2)})\|,\|G_1^{\perp\perp} (\mathcal{E}_1^{(2)})\|\leq\delta_0^{-\frac{1}{400}}$, where $G_1^{\perp\perp}$ denotes the Green's function for $B_1^{(2)}$ on the orthogonal complement of the space spanned by $\psi_1$ and $\Psi_1$.
		\end{itemize}
	\end{prop}
	\begin{proof}
		We first consider the case $\theta^*=\theta_{k,-}$.  In this case, since  $v (\theta^*)=v (\theta^*+k\cdot \omega)$,  $\delta (x-o)$ and $\delta (x-k)$ serve as two trial wave functions  for the operator $H_{B_1^{(2)}} (\theta^*)-v (\theta^*)$. By Lemma \ref{trialcor}, we deduce  $H_{B_1^{(2)}} (\theta^*)$ has two  eigenvalues in  $[v (\theta^*)-O (\varepsilon),v (\theta^*)+O (\varepsilon)]$. The claim for $\theta\in I^{(2)}_{k,\cup}$ follows from $|I^{(2)}_{k,\cup}|\leq 4\delta_0^{\frac{1}{200}}$ and  the uniform bound $|v'|\leq D$.
		
		To establish the decay of eigenfunctions, we  denote  $\Lambda=B^{(2)}_1\setminus\{o,k\}$. By Proposition \ref{DRpro}, for $x\in \Lambda$, 
		$$  |E_1^{(2)} (\theta)-v (\theta+x\cdot \omega)|\geq  |v (\theta)-v (\theta+x\cdot \omega)|-|E_1^{(2)} (\theta)-v (\theta)|\geq 2\delta_0^{\frac{1}{400}}-O (\varepsilon)>\delta_0^{\frac{1}{400}}.$$
		By Neumann series argument, we get 
		$$\|G_\Lambda (E_1^{(2)})\|\leq \delta_0^{-\frac{1}{400}}, \  |G_\Lambda (E_1^{(2)};x,y)|\leq\delta_0^{-\frac{1}{400}}e^{-\gamma_0\|x-y\|_1},$$ which together with restricting the equation $ (H_{B_1^{(2)}}-E_1^{(2)})\psi_1=0$ to $\Lambda$  yields the exponential decay of $\psi_1$. The proof of exponential decay  of  $\Psi_1$ is analogous.
		The expression \eqref{y} follows from the above estimates and  the orthogonality of $\psi_1$ and $\Psi_1$. 
		
		If there is a third  eigenvalue $|\hat{E}-v (\theta)|<2\delta_0^{\frac{1}{400}}$, the above argument shows that its eigenfunction  can also be expressed as the form \eqref{y}. This   violates the orthogonality relation  of   $\psi_1,\Psi_1$ and leads to  the contradiction. 
		
		Item \textbf{(c)}  follows from \textbf{(a)} immediately.
	\end{proof}
	\begin{rem}
		Thanks  to the perturbation theory  for self-adjoint operators, the two Rellich functions can be  properly parameterized so that  $E_1^{(2)} (\theta),\mathcal{E}_1^{(2)} (\theta)\in C^1 (I^{(2)}_{k,\cup})$, although the level crossing issue may occur. We will  use the   $C^1$-parameterization for $E_1^{(2)}$ and $\mathcal{E}_1^{(2)}$ in the following discussion. 
	\end{rem}
	Verifying the  Morse condition of   $E_1^{(2)} (\theta)$ and $\mathcal{E}_1^{(2)} (\theta)$ in double resonance case needs more efforts. 
	
	For conciseness, we define $$E_{0,-} (\theta)=v (\theta), \ E_{0,+} (\theta)=v (\theta+k\cdot \omega),$$ 
	$$\psi_{0,-}:=\delta (x-o), \  \psi_{0,+}:=\delta (x-k),$$
	and
	$$E_{1,>}^{(2)} (\theta)=\max (E_1^{(2)} (\theta), \mathcal{E}_1^{(2)} (\theta)), \   E_{1,<}^{(2)} (\theta)=\min (E_1^{(2)} (\theta), \mathcal{E}_1^{(2)} (\theta)).$$
	\begin{lem}\label{transe}
		For $\theta\in I_{k,\cup}^{(2)}$, we have
		$$\pm E_{0,\pm}' (\theta)\gtrsim  (l_1^{(1)})^{-2\tau}.$$
	\end{lem}
	\begin{proof}
		Recalling \eqref{160} and  $|v'|\leq D$, we have 
		$$\|\theta_{k,\pm}-\theta_m\|,\|\theta_{k,\pm}-\theta_M\|\gtrsim   (l_1^{(1)})^{-2\tau}.$$
		Since $\theta\in I_{k,\cup}^{(2)}$ and   $|I_{k,\cup}^{(2)}|\leq4\delta_0^{\frac{1}{200}}\ll  (l_1^{(1)})^{-2\tau} $, we have 
		$$\|\theta-\theta_m\|,\|\theta-\theta_M\|,\|\theta+k\cdot \omega -\theta_m\|,\|\theta+k\cdot \omega -\theta_M\|\gtrsim   (l_1^{(1)})^{-2\tau}.$$
		The above estimates together with   Lemma \ref{derv}   yield  the desired lower bound.
	\end{proof}
	
	\begin{thm}[Cauchy Interlacing Theorem]\label{cahuchy}
		{ 
			%		xiugai19
			Let $A$ be a self-adjoint operator on a Hilbert space of dimension $n$, let $m \leqslant n$, and let $P$ be an $m$-dimensional  orthogonal projection. Let $B=P A P|_{\operatorname{Image}P}$ be a compression of $A$,} and denote the  (ordered) eigenvalues of $A$  (resp. $B$) by $\alpha_1 \leqslant \alpha_2 \leqslant \cdots \leqslant \alpha_n$  (resp. $\beta_1 \leqslant \beta_2 \leqslant \cdots \leqslant \beta_m$). Then
		$$
		\alpha_k \leqslant \beta_k \leqslant \alpha_{k+n-m} .
		$$
	\end{thm}
	\begin{proof}
		The interlacing theorem is proven via a standard min-max argument.
	\end{proof}
	
	To settle the level crossing issue, we use Theorem \ref{cahuchy} to construct a pair of auxiliary functions $\lambda_\pm$. For our purpose, let 
	$$P_\pm:=\langle\psi_{0,\pm}|\psi_{0,\pm}\rangle,\ Q_\pm:=\operatorname{Id}_{B_1^{(2)}}-P_\pm,$$
	and consider the  compressed operators 
	$$H_\pm (\theta):=Q_\pm H_{B_1^{(2)}} (\theta) Q_\pm.$$
	\begin{lem}\label{inter1}
		For $\theta\in I^{(2)}_{k,\cup}$, the operators $H_\pm (\theta)$ have unique eigenvalues $\lambda_\mp (\theta)$ interlacing $E_{1,>}^{(2)} (\theta)$ and $E_{1,<}^{(2)} (\theta)$: 
		\begin{equation}\label{inter}
			E_{1,<}^{(2)} (\theta)\leq \lambda_\mp (\theta) \leq E_{1,>}^{(2)} (\theta).
		\end{equation}
		Moreover, $\lambda_\pm (\theta)$ are $C^1$-closed  to $E_{0,\pm} (\theta)$:
		\begin{equation}\label{Epl}
			|\frac{d^s}{d\theta^s} (\lambda_\pm (\theta)-E_{0,\pm} (\theta))|\lesssim\frac{\varepsilon}{\delta_0^s}\ll \delta_0^{10},\  s=0,1.
		\end{equation}
		In particular, we have 
		\begin{equation}\label{tranl}
			\pm \lambda_\pm' (\theta)\gtrsim   (l_1^{(1)})^{-2\tau}\gg \delta_0^{\frac{1}{1000}} .
		\end{equation}
	\end{lem}
	\begin{proof}
		We claim that  $H_\pm (\theta)$ has  a  unique eigenvalue  $\lambda_\mp (\theta)$ in  $[v (\theta)-\delta_0^{\frac{1}{400}},v (\theta)+\delta_0^{\frac{1}{400}}]$ satisfying 
		$$|\lambda_\pm (\theta)-E_{0,\pm} (\theta)|\lesssim\varepsilon.$$
		For   the existence, we notice that 
		$$\| (H_\pm (\theta)-E_{0,\mp} (\theta))\psi_{0,\mp}\|=\|Q_\pm (H (\theta)-E_{0,\mp} (\theta))\psi_{0,\mp}\|\lesssim\varepsilon$$ and employ Lemma \ref{trialcor} to finish the proof.

		For the uniqueness, assuming   $\hat{\lambda}$ is an eigenvalue  of $H_\pm (\theta)$ satisfying $|\hat{\lambda}-v (\theta)|\leq \delta_0^{\frac{1}{400}}$, then its  eigenfunction $\hat{\phi}\in \operatorname{Image}Q_\pm$ satisfies $$Q_\pm  (H_{B_1^{(2)}} (\theta)-\hat{\lambda})\hat{\phi}=0.$$
		Thus $$ (H_{B_1^{(2)}} (\theta)-\hat{\lambda})\hat{\phi}=a\psi_{0,\pm}.$$
		By the same argument as the proof of item \textbf{(b)} of Proposition \ref{k2}, one can prove that $\hat{\phi}=\psi_{0,\mp}+O (\varepsilon\delta_0^{-1})$ since  $\hat{\phi}\in \operatorname{Image}Q_\pm$. Thus such an eigenvalue must be unique or else it will violate the  orthogonality. By Theorem \ref{cahuchy} and item \textbf{(a)} of Proposition \ref{k2}, $\lambda_\mp (\theta)$ must lie between $E_{1,<}^{(2)} (\theta)$ and  $ E_{1,>}^{(2)} (\theta)$.
		Let $\phi_\mp\in \operatorname{Image}Q_\pm $ be the eigenfunctions of $H_\pm$ corresponding to $\lambda_\mp$. Since $\|\phi_\mp-\psi_{0,\mp}\|\lesssim \varepsilon \delta_0^{-1}$, by the Feynman-Hellman formula \textbf{(1)} from  Lemma \ref{daoshu}, we have 
		\begin{align*}
			|\lambda_\mp'-E_{0,\mp}'|&=|\left\langle\phi_\mp, (H_\pm)' \phi_\mp
			\right\rangle-\left\langle\psi_{0,\mp},H' \psi_{0,\mp}\right\rangle|\\
			&=|\left\langle\phi_\mp,V' \phi_\mp
			\right\rangle-\left\langle\psi_{0,\mp},V' \psi_{0,\mp}\right\rangle|\\&\lesssim\frac{\varepsilon}{\delta_0}\ll \delta_0^{10}.
		\end{align*}		
		Finally, \eqref{tranl} follows from the above inequality and Lemma \ref{transe}.
	\end{proof}
	{
		%	xiugai3
		\begin{rem}
			The interlacing argument (c.f. Theorem \ref{cahuchy} and Lemma \ref{inter1}) was first introduced by Forman-VandenBoom \cite{FV21} to study a pair of double resonant Rellich functions. In the one-dimensional case, the interlacing argument together with the classical eigenvalue separation lemma for  Jacobi matrix (c.f. \cite{FV21}, Lemma 1.5) provides  a uniform gap of  the two Rellich graphs in the vertical direction. Since the gap  is much larger than the next resonance strength, energies near the codomain of one resonant Rellich function  avoid  the resonance with the other one, crucially enabling their induction. The present work extends  the interlacing idea  to  the multi-dimensional setting,  where the  level crossing issue of  a pair of double resonant Rellich functions may present. We introduce a new type of Rellich functions (c.f. {\bf Type} \ref{t3}) to handle this issue and enable our induction.
	\end{rem}}
	\begin{lem}\label{cross}
		There is a unique point $\theta_s\in I^{(2)}_{k,\cup}$ satisfying $|\theta_s-\theta_{k,-}|<\delta_0^{10}$, such that 
		$$\lambda_+ (\theta_s)=\lambda_- (\theta_s).$$	
		Moreover, the Rellich children have  quantitative separation away from $\theta_s$: 
		\begin{equation}\label{E1s}
			E_{1,>}^{(2)} (\theta)-E_{1,<}^{(2)} (\theta)\geq |\lambda_+ (\theta)-\lambda_- (\theta)|\geq \delta_0^{\frac{1}{1000}}|\theta-\theta_s|.
		\end{equation}
		As a corollary, if $\theta\in I^{(2)}_{k,\cup}$ satisfies    $E_1^{(2)} (\theta)=\mathcal{E}_1^{(2)} (\theta)$, then $\theta=\theta_s$.
	\end{lem}
	\begin{proof}
		We consider the difference function 
		$$d (\theta):=\lambda_+ (\theta)-\lambda_- (\theta).$$
		Since $E_{0,+} (\theta_{k,-})=E_{0,-} (\theta_{k,-})$, by  \eqref{Epl} (the version $s=0$),  we have 
		$$|d (\theta_{k,-})|\leq |\lambda_+ (\theta_{k,-})-E_{0,+} (\theta_{k,-})|+|\lambda_- (\theta_{k,-})-E_{0,-} (\theta_{k,-})|\lesssim \varepsilon.$$
		Recalling    \eqref{tranl}, we have 
		$$d' (\theta)\geq\delta_0^{\frac{1}{1000}}. $$
		Thus by mean value theorem  there exists  a unique  point $\theta_s$ satisfying $|\theta_s-\theta_{k,-}|\lesssim \varepsilon \delta_0^{-\frac{1}{1000}}<\delta_0^{10}$, such that 
		$d (\theta_s)=0$, that is 	$$\lambda_+ (\theta_s)=\lambda_- (\theta_s).$$	
		Finally, \eqref{E1s} follows from \eqref{inter} and  
		$$	|\lambda_+ (\theta)-\lambda_- (\theta)|=|d (\theta)|=|d (\theta)-d (\theta_s)|\geq \delta_0^{\frac{1}{1000}}|\theta-\theta_s|$$
		using mean value theorem.
	\end{proof}
	\begin{figure}[htp]
		\begin{tikzpicture}[>=latex, scale=0.88]
			\draw[domain=-4:0, samples=100] plot  ({\x},{\x*\x/5});
			\draw[domain=0:2.1, samples=100] plot  ({\x},{\x*\x/1.25});
			\draw[domain=-2.1:0, samples=100] plot  ({\x-0.5},{-1-\x*\x/1.25});
			\draw[domain=0:4, samples=100] plot  ({\x-0.5},{-1-\x*\x/5});
			\draw[domain=-4:-0.25, samples=100,red,dashed,thick] plot  ({\x},{-0.4*(\x+0.25)+0.12*(\x+0.25)*(\x+0.25)-0.55});
			\draw[domain=-0.25:3.5, samples=100,red,dashed,thick] plot  ({\x},{-0.4*(\x+0.25)-0.12*(\x+0.25)*(\x+0.25)-0.55});
			\draw[domain=-0.25:2.1, samples=100,red,dashed,thick] plot  ({\x},{0.3*(\x+0.25)+0.45*(\x+0.25)*(\x+0.25)-0.55});
			\draw[domain=-2.7:-0.25, samples=100,red,dashed,thick] plot  ({\x},{0.3*(\x+0.25)-0.45*(\x+0.25)*(\x+0.25)-0.55});
			\fill  (-0.25,-0.55)circle (2pt); \draw (-0.25,-0.55)[below]node{$\theta_s$};
			\draw[<->]  (3.5,-5)-- node[above]{{$I_{k,\cup}^{(2)}$}} (-4,-5); 
			\draw (1.8,1)node{$\lambda_+$}; 		\draw (1.1,2)node{$E_{1,>}^{(2)}$};  		\draw (1.8,-1.2)node{$\lambda_-$}; \draw (1,-2)node{$E_{1,<}^{(2)}$}; 
			%	\draw[domain=-0.73*pi:0.73*pi, samples=100,red] plot  ({\x+8},{6*sin ( (0.3*\x r) %-1});		
			\draw[domain=-1.7:0, samples=100] plot  ({\x+8},{1.95*\x-\x*\x/5-0.2});
			\draw[domain=1.7:0, samples=100] plot  ({\x+8},{1.95*\x+\x*\x/5-0.2});
			\draw[domain=-2.3:0, samples=100] plot  ({\x+8},{-1.2*\x+\x*\x/5-0.2});
			\draw[domain=2.3:0, samples=100] plot  ({\x+8},{-1.2*\x-\x*\x/5-0.2});
			\draw[domain=-2:0, samples=100,red,dashed,thick ] plot  ({\x+8},{1.5*\x-\x*\x/5-0.2});
			\draw[domain=2:0, samples=100,red,dashed,thick ] plot  ({\x+8},{1.5*\x+\x*\x/5-0.2});
			\draw[domain=-2.3:0, samples=100,red,dashed,thick ] plot  ({\x+8},{-0.9*\x+\x*\x/5-0.2});
			\draw[domain=2.3:0, samples=100,red,dashed,thick ] plot  ({\x+8},{-0.9*\x-\x*\x/5-0.2});
			\fill  (8,-0.2)circle (2pt); \draw (8,-0.3)[below]node{$\theta_s$};
			\draw[<->]  (5.7,-5)-- node[above]{{$I_{k,\cup}^{(2)}$}} (10.3,-5); 
			\draw (9.8,1.4)node{$\lambda_+$}; 		\draw (8.4,2)node{$E_{1,>}^{(2)}$};  		\draw (9.8,-1.7)node{$\lambda_-$}; \draw (8.7,-2)node{$E_{1,<}^{(2)}$}; 
		\end{tikzpicture}
		\caption{A cartoon illustration of  Lemma \ref{inter1}: the graph on the left shows the eigenvalue separation case and the  graph on the  right shows the level crossing case.}
	\end{figure}
	\begin{rem}\label{label}
		In the following, we assume that the Rellich functions are chosen so that $E_1^{(2)}>\mathcal{E}_1^{(2)}$ on the right of $\theta_s$.
	\end{rem}
	Suppose that $|E_{0,+}' (\theta_{k,-})|\geq|E_{0,-}' (\theta_{k,-})|$  (the argument is similar to the opposite case). Recalling Lemma  \ref{transe}, we  define $1\leq r\lesssim  (l_1^{(1)})^{2\tau}$ such that 
	$$    |E_{0,+}' (\theta_{k,-})|=r|E_{0,-}' (\theta_{k,-})|.$$
	\begin{lem}\label{cha}
		For $\theta\in I_{k,\cup}^{(2)}$, we have 
		$$| (E_{0,+}'+rE_{0,-}') (\theta)|\leq \delta_0^{\frac{1}{400}}.$$
	\end{lem}
	\begin{proof}
		Since $E_{0,+}'$ and $E_{0,-}'$ have opposite signs, we have 
		$$E_{0,+}' (\theta_{k,-})+rE_{0,-}' (\theta_{k,-})=0 .  $$
		By mean value  theorem and  $|E_{0,\pm}''|\leq D$, we have 
		$$ | (E_{0,+}'+rE_{0,-}') (\theta)|\leq  (1+r)D |\theta-\theta_{k,-}|\leq  (1+r)D \delta_0^{\frac{1}{200}}\leq\delta_0^{\frac{1}{400}}.$$
	\end{proof}
	\begin{prop}\label{1215}
		For $\theta\in I_{k,\cup}^{(2)}$, we have the following:
		\begin{itemize}
			\item[\textbf{(a).}] If $E_1^{(2)} (\theta)\neq\mathcal{E}_1^{(2)} (\theta)$, then 
			\begin{align}
				(E_1^{(2)})'& = (A^2-rB^2) E_{0,-}'+O (\delta_0^{\frac{1}{400}} ), \label{de}\\
				(\mathcal{E}_1^{(2)})'& = (B^2-rA^2) E_{0,-}'+O (\delta_0^{\frac{1}{400}} ),\label{dee}
			\end{align}
			and 	
			\begin{align}
				( E_1^{(2)})'' & =\frac{2\left\langle\psi_1, V' \Psi_1\right\rangle^2}{E_1^{(2)}-\mathcal{E}_1^{(2)}}+O (\delta_0^{-\frac{1}{400}} ),\label{df} \\
				( \mathcal{E}_1^{(2)})''& =\frac{2\left\langle\psi_1, V' \Psi_1\right\rangle^2}{\mathcal{E}_1^{(2)}-E_1^{(2)}}+O (\delta_0^{-\frac{1}{400}} )\label{ddf}
			\end{align}
			with the notation from Proposition \ref{k2} and $r$ as above.
			\item [\textbf{(b).}] If  $E_1^{(2)} (\theta)\neq \mathcal{E}_1^{(2)} (\theta)$ and  $| (E_1^{(2)})' (\theta)|\leq\delta_0^{\frac{1}{1000}}$, then $| (E_1^{(2)})'' (\theta)|\geq2$. Moreover, the  sign of $ (E_1^{(2)})'' (\theta)$ is the same as that of $ E_1^{ (2)} (\theta)- \mathcal{E}_1^{(2)} (\theta)$. The analogous conclusion holds if we exchange   $E_1^{(2)} (\theta)$ and $\mathcal{E}_1^{ (2)} (\theta)$.
			\item[\textbf{(c).}] If there is a level crossing  (thus $E_1^{(2)} (\theta_s)=\mathcal{E}_1^{(2)} (\theta_s)$ by Lemma \ref{cross}), then for all $\theta\in I_{k,\cup}^{ (2)}$,
			\begin{equation}\label{da}
				(E_1^{(2)})' (\theta)>\delta_0^{\frac{1}{1000}},\  (\mathcal{E}_1^{(2)})' (\theta)<-\delta_0^{\frac{1}{1000}}.
			\end{equation}
		\end{itemize}
	\end{prop}
	\begin{proof}
		We prove \eqref{de}  (the proof of \eqref{dee} is analogous). 	 By \eqref{y} and Lemma \ref{cha}, we  apply Feynman-Hellman formula \textbf{(1)} in  Lemma \ref{daoshu} to obtain 
		\begin{align*}
			( E_1^{(2)})'
			& =\left\langle\psi_1, V' \psi_1\right\rangle =A^2 \left\langle\psi_{0,-}, V' \psi_{0,-}\right\rangle +B^2  \left\langle\psi_{0,+}, V' \psi_{0,+}\right\rangle+O (\delta_0^{10} ) \\
			&=A^2  E_{0,-}'+B^2  E_{0,+}'+O (\delta_0^{10} ) \\
			& = (A^2-rB^2) E_{0,-}'+B^2  (E_{0,+}'+rE_{0,-}')+O (\delta_0^{10} )\\
			& = (A^2-rB^2) E_{0,-}'+O (\delta_0^{\frac{1}{400}} ).
		\end{align*}
		To prove \eqref{df} and \eqref{ddf}, we  use Feynman-Hellman formulas  \textbf{(2)} and \textbf{(3)} to obtain 	
		\begin{equation*}
			(E_1^{(2)})''=\left\langle\psi_1, V'' \psi_1\right\rangle+2 \frac{\left\langle\psi_1, V' \Psi_1\right\rangle^2}{E_1^{(2)}-\mathcal{E}_1^{(2)}}-2\left\langle V' \psi_1,G^{\perp \perp}_1 (E_1^{(2)}) V' \psi_1\right\rangle .
		\end{equation*}
		The first term is bounded by $D$ and the  third  term is bounded by $$2\|G^{\perp \perp}_1 (E_1^{(2)})\|\cdot\|V' \psi_1\|^2\lesssim \delta_0^{-\frac{1}{400}},$$ where we  use the estimate $\|G^{\perp \perp}_1 (E_1^{(2)})\| \leq\delta_0^{-\frac{1}{400}}$ from  item \textbf{(c)} of   Proposition \ref{k2}. Thus we finish the proof of item \textbf{(a)}.
		
		Now we are going to prove \textbf{(b)}. We will show the first term in \eqref{df} is large if  $| (E_1^{(2)})' (\theta)|\leq\delta_0^{\frac{1}{1000}}$. 
		Assume  $| (E_1^{(2)})' (\theta)|\leq\delta_0^{\frac{1}{1000}}$,  by \eqref{de}, we have 
		$$|A^2-rB^2|\cdot |E_{0,-}' (\theta)|\leq| ( E_1^{(2)})' (\theta)|+O (\delta_0^{\frac{1}{400}})\leq 2\delta_0^{\frac{1}{1000}}.$$
		From Lemma \ref{transe}, it follows that 
		$$|A^2-rB^2|\leq 2\delta_0^{\frac{1}{1000}} (l_1^{(1)})^{2\tau}<\frac{1}{100}.$$
		Since $A^2+B^2=1$ and $r\geq1$, we obtain  
		$$ (1+r)B^2\geq A^2+B^2-|A^2-rB^2|\geq \frac{99}{100}$$ and 
		$$ (1+\frac{1}{r})A^2\geq A^2+B^2-\frac{1}{r}|A^2-rB^2|\geq \frac{99}{100}.$$
		Thus $$B^2\geq\frac{1}{4r},\ A^2\geq\frac{1}{4}.$$
		By \eqref{y},  Lemmas \ref{transe}, \ref{cha}  and the  lower bounds on $A,B$,  we have 
		\begin{align*}
			|\langle\psi_1, V' \Psi_1\rangle|&=|AB (\langle\psi_{0,-}, V' \psi_{0,-}\rangle-\langle\psi_{0,+}, V' \psi_{0,+}\rangle)+O (\delta_0^{10} )| \\
			&=|AB   (E_{0,-}' -E_{0,+}')+O (\delta_0^{10})|\\
			&=|AB \left( (1+r) E_{0,-}'- (E_{0,+}'+r E_{0,-}')\right)+O (\delta_0^{10})|\\
			& = \frac{1+r }{4\sqrt{r}} |E_{0,-}'|-O (\delta_0^{\frac{1}{400}} )\\
			&\gtrsim  (l_1^{(1)})^{-2\tau}\gg \delta_0^{\frac{1}{1000}}.
		\end{align*}
		By item \textbf{(a)} of  Proposition \ref{k2}, we obtain  the estimate of the denominator  $|E_1^{(2)}-\mathcal{E}_1^{(2)}|\lesssim \delta_0^{\frac{1}{200}}$. Combining  the previous estimate of numerator, by \eqref{df},  we obtain  $$| (E_1^{(2)})'' (\theta)|\gtrsim     \delta_0^{\frac{1}{500}}   \delta_0^{-\frac{1}{200}}-O (\delta_0^{-\frac{1}{400}})>2,$$ 
		whose sign is the same as  that of $ E_1^{(2)} (\theta)- \mathcal{E}_1^{(2)} (\theta)$.
		
		To prove \textbf{(c)}, we first show \eqref{da} holds for $\theta=\theta_s$. Since $E_1^{(2)} (\theta_s)=\mathcal{E}_1^{(2)} (\theta_s)$ and  \eqref{inter},  we have 
		$$E_1^{(2)} (\theta_s)=\mathcal{E}_1^{(2)} (\theta_s)=\lambda_+ (\theta_s)=\lambda_- (\theta_s),$$ 
		and since $E_1^{(2)}>\mathcal{E}_1^{(2)}$ on the right of $\theta_s$  (c.f. Remark \ref{label}), we have for $\theta>\theta_s$, 
		$$E_1^{(2)} (\theta)\geq\lambda_+ (\theta)>\lambda_- (\theta) \geq \mathcal{E}_1^{(2)} (\theta).$$
		Thus by \eqref{tranl},
		we get $$   (E_1^{(2)})' (\theta_s)\geq \lambda_+' (\theta_s) >\delta_0^{\frac{1}{1000}},\   (\mathcal{E}_1^{(2)})' (\theta_s)\leq \lambda_-' (\theta_s)<-\delta_0^{\frac{1}{1000}}.$$
		We next claim  the inequalities hold for all $\theta\in I_{k,\cup}^{(2)}$. We prove the case that  $ (E_1^{(2)})' (\theta)$ and $\theta>\theta_s$  (the other cases are analogous). If it is not true, then the set $$\{\theta \in I_{k,\cup}^{(2)}:\  \theta>\theta_s, \ \  (E_1^{(2)})' (\theta)\leq  \delta_0^{\frac{1}{1000}} \}\neq \emptyset.$$ Let $\theta^*$ be its infimum. Since $E_1' (\theta)$ is continuous and $  (E_1^{(2)})' (\theta_s) >\delta_0^{\frac{1}{1000}}$, we have $\theta^*>\theta_s$ and  $ (E_1^{(2)})' (\theta)>\delta_0^{\frac{1}{1000}}\geq  (E_1^{(2)})' (\theta^*)$ for $\theta\in [\theta_s,\theta^*)$, which implies $  (E_1^{(2)})'' (\theta^*)\leq 0$. However, by item \textbf{(b)} and $E_1^{(2)} (\theta^*)>\mathcal{E}_1^{ (2)} (\theta^*)$, we get    $ (E_1^{(2)})'' (\theta^*)>2$, a contradiction. Thus we finish the proof of the claim. 
	\end{proof}
	The following proposition shows that 	$E_{1,<}^{(2)} (\theta)$ and $E_{1,>}^{(2)} (\theta)$ are uniformly closed   (of order $\varepsilon$) to the previous scale functions $E_{0,\pm} (\theta)$, providing a more completed  picture for the  structure of  the Rellich functions.
	
	Define 
	$$E_{0,\vee} (\theta):=\max (E_{0,+} (\theta), E_{0,-} (\theta)), \   E_{0,\wedge} (\theta):=\min (E_{0,+} (\theta), E_{0,-} (\theta)).$$
	\begin{prop}
		For  $\theta\in I_{k,\cup}^{(2)}$, we have 
		\begin{equation}\label{apbp}
			|E_{1,>}^{(2)} (\theta)-E_{0,\vee} (\theta)|,| E_{1,<}^{(2)} (\theta)-E_{0,\wedge} (\theta)|\lesssim\varepsilon.
		\end{equation}
	\end{prop}
	\begin{proof}
		Fix  $\theta\in I_{k,\cup}^{(2)}$. We have
		$$\| (H_{B_1^{(2)}}-E_{0,\pm})\psi_{0,\pm}\|\lesssim \varepsilon. $$ 
		By Lemma \ref{trialcor}, we deduce  $H_{B_1^{(2)}}$ must have an eigenvalue in $[E_{0,\vee}-O (\varepsilon), E_{0,\vee}+O (\varepsilon)]$ and in $[E_{0,\wedge}-O (\varepsilon), E_{0,\wedge}+O (\varepsilon)]$. By  item \textbf{(a)} of proposition \ref{k2}, the eigenvalue must be $E_1^{(2)}$ or $\mathcal{E}_1^{(2)}$.  If the two intervals are disjoint, \eqref{apbp} must hold.
		Otherwise, we have \begin{equation}\label{other}
			E_{0,\vee}-E_{0,\wedge}\lesssim \varepsilon.
		\end{equation} Thus we have two trial functions  for  $E_{0,\vee}$: 
		$$\| (H_{B_1^{(2)}}-E_{0,\vee})\psi_{0,\pm}\|\lesssim \varepsilon.$$ 
		By Lemma \ref{trialcor},  $H_{B_1^{(2)}}$ must have two eigenvalues in $[E_{0,\vee}-O (\varepsilon), E_{0,\vee}+O (\varepsilon)].$ They are $E_1^{(2)}$ and  $\mathcal{E}_1^{(2)}$. Combining \eqref{other} shows \eqref{apbp}.
	\end{proof}
	At the end of this part, we  show that the two Rellich children $E_{1,>}^{(2)}$ and $E_{1,<}^{(2)}$ have   the two-monotonicity interval structure.
	\begin{prop}\label{str2n} We have the following:
		\begin{itemize}
			\item [\textbf{(a).}] If there is no level crossing, then $E_1^{(2)}=E_{1,>}^{(2)}$ and $\mathcal{E}_1^{(2)}=E_{1,<}^{(2)}$. Moreover, $E_1^{ (2)}$   (resp. $\mathcal{E}_1^{(2)}$) has the two-monotonicity interval structure with a  critical point $\theta_>$   (resp. $\theta_<$) satisfying $|\theta_>-\theta_{k,-}|<\delta_0^{10}$  (resp. $|\theta_<-\theta_{k,-}|<\delta_0^{10}$).
			\item [\textbf{(b).}] If there is a level crossing  (thus $E_1^{(2)} (\theta_s)=\mathcal{E}_1^{(2)} (\theta_s)$ by Lemma \ref{cross}),  then 
			\begin{align*}
				E_1^{(2)} (\theta)=	\left\{\begin{aligned}
					&E_{1,<}^{(2)} (\theta) &\text{if $\theta<\theta_s$,}\\
					&E_{1,>}^{(2)} (\theta) &\text{if $\theta>\theta_s$,}
				\end{aligned}\right. 
			\end{align*}
			and 
			\begin{align*}
				\mathcal{E}_1^{(2)} (\theta)=	\left\{\begin{aligned}
					&E_{1,>}^{(2)} (\theta) &\text{if $\theta<\theta_s$,}\\
					&E_{1,<}^{(2)} (\theta) &\text{if $\theta>\theta_s$.}
				\end{aligned}\right.	\end{align*}
			Moreover, $E_{1,>}^{(2)}$ and $E_{1,<}^{(2)}$ are piecewise $C^1$ functions  (except at the point $\theta_s$) with two-monotonicity interval structure.
		\end{itemize}
	\end{prop}
	\begin{proof}
		If there is no level crossing, then $E_{1,>}^{(2)} (\theta)>E_{1,<}^{(2)} (\theta)$ for all $\theta \in I_{k,\cup}^{(2)}$. Since $E_1^{(2)}>\mathcal{E}_1^{(2)}$ on the right of $\theta_s$, by continuity $E_1^{ (2)} (\theta)>\mathcal{E}_1^{(2)} (\theta)$ for all $\theta \in I_{k,\cup}^{ (2)}$. Thus $E_1^{(2)}=E_{1,>}^{(2)}$ and $\mathcal{E}_1^{ (2)}=E_{1,<}^{(2)}$.
		
		Now we show the existence of the critical point $\theta_>$ for $E_1^{(2)}$  (the proof of $\theta_<$ for $\mathcal{E}_1^{(2)}$ is analogous).  Let $\theta_-:=\theta_{k,-}-\delta_0^{10}$ and $\theta_+:=\theta_{k,-}+\delta_0^{10}$.  Recalling  Lemma \ref{transe} and \eqref{apbp}, we have 
		$$E_1^{ (2)} (\theta_+)\geq  E_{0,+} (\theta_+)-O (\varepsilon)\geq   E_{0,+} (\theta_{k,-})+\delta_0^{11}-O (\varepsilon) 
		> E_1^{ (2)} (\theta_{k,-}) $$ and 
		$$E_1^{ (2)} (\theta_-)\geq  E_{0,-} (\theta_-)-O (\varepsilon)\geq   E_{0,-} (\theta_{k,-})+\delta_0^{11}-O (\varepsilon) 
		> E_1^{ (2)} (\theta_{k,-}). $$
		Thus $E_1^{ (2)}$ must have a critical point $\theta_>\in [\theta_-,\theta_+]$ by mean value theorem.
		The two-monotonicity interval structure follows from  item \textbf{(b)} of Proposition \ref{1215} and Lemma \ref{C2}. Thus we finish the proof of \textbf{(a)}.
		
		If $E_1^{ (2)} (\theta_s)=\mathcal{E}_1^{ (2)} (\theta_s)$, by item \textbf{(c)} of Proposition \ref{1215}, $E_1^{ (2)}$ is strictly increasing and $\mathcal{E}_1^{ (2)}$ is strictly decreasing. Thus $E_1^{ (2)}>\mathcal{E}_1^{ (2)}$ on the right of $\theta_s$ and $E_1^{ (2)}<\mathcal{E}_1^{ (2)}$ on the left  of $\theta_s$. From the above observation,  the proof of \textbf{(b)} is obvious.
	\end{proof}
	{
		%	xiugai4
		\begin{rem}
			While in the previous proposition, we resolve a  pair of double resonant Rellich functions  as two  piecewise $C^1$ Rellich curves  satisfying a Morse condition, we still crucially utilize the $C^2$-smoothness of the potential $v$ to ensure the $C^2$ smoothness of  Rellich functions in order to apply  Feynman-Hellman formulas (c.f. {\rm (2)} of Lemma \ref{daoshu}). It will  be interesting that   the $C^2$-cosine type assumption on $v$ could  be weakened to the  piecewise $C
			^1$  one  with two critical points satisfying  some  Morse condition.
	\end{rem}}
	\subsubsection{Domain adjustment}\label{adjm}
	In the above two parts, for every phase domain  $I_{i}^{ (1)}$ or $I_{k,\cup}^{ (2)}$, we have found a  block $B_1^{ (1)}$ or $B_1^{ (2)}$  such that any resonant Rellich child is Morse with two-monotonicity interval structure  (c.f. Proposition \ref{str1} and \ref{str2n}). To finish the construction of  the first generation Rellich functions $\mathcal{C}_1$, we must modify the domain  so  that the   Rellich child has  the same image on  each of its   monotone  interval to prepare for  the  construction of  the next scale.  {
		%		xiugai3
		 This technical argument was also first introduced in \cite{FV21}.}
	\begin{prop}\label{adj}We have the following:
		\begin{itemize}
			\item[\textbf{(a).}] In the  simple resonance case, for the Rellich function   $E_1^{ (1)}$ defined on  $I^{ (1)}_i$, we can find  $I (E_1^{ (1)})\subset I^{ (1)}_i$ with 
			\begin{equation}\label{error}
				|I^{ (1)}_i\setminus I (E_1^{ (1)})|\leq \delta_0^{5},
			\end{equation}
			such that  $E_1^{ (1)}$ maintains the two-monotonicity interval structure  on  $I (E_1^{ (1)})$ and has the same image $J (E_1^{ (1)}):=E_1^{ (1)} (I (E_1^{ (1)}))$ on each monotone  component of $I (E_1^{ (1)})$.
			\item[\textbf{(b)}.] In the  double resonance case, for the Rellich functions  $E_{1,>}^{(2)},E_{1,<}^{(2)}$ defined on  $I^{ (2)}_{k,\cup}$, we can find $I (E_{1,>}^{(2)})\subset [\operatorname{inf} I^{(2)}_{k,-},\operatorname{sup}I^{(2)}_{k,+}-k\cdot\omega ],I (E_{1,<}^{(2)})\subset [\operatorname{inf} I^{(2)}_{k,+}-k\cdot \omega ,\operatorname{sup}I^{ (2)}_{k,-}]$ with 
			\begin{equation}\label{inside}
				[\theta_{k,-}-\delta_0, \theta_{k,-}+\delta_0]\subset  I (E_{1,>}^{(2)})\cap I (E_{1,<}^{(2)}),
			\end{equation}
			$$|[\operatorname{inf} I^{(2)}_{k,-},\operatorname{sup}I^{(2)}_{k,+}-k\cdot\omega ]\setminus I (E_{1,>}^{(2)}) | \leq \delta_0^{5},$$
			$$|[\operatorname{inf} I^{(2)}_{k,+}-k\cdot \omega ,\operatorname{sup}I^{(2)}_{k,-}]\setminus I (E_{1,<}^{(2)}) | \leq \delta_0^{5},$$
			such that the Rellich function $E_{1,>}^{(2)}$  (resp. $E_{1,<}^{(2)}$)  maintains the two-monotonicity interval structure on  $ I (E_{1,>}^{(2)})$  (resp. $I (E_{1,<}^{(2)})$) and  has the same image $J (E_{1,>}^{(2)}):=E_{1,>}^{(2)} (I (E_{1,>}^{(2)}))$  (resp. $J (E_{1,<}^{(2)}):=E_{1,<}^{(2)} (I (E_{1,<}^{(2)}))$)  on each monotone  component of $ I (E_{1,>}^{(2)})$  (resp. $I (E_{1,<}^{(2)})$).
		\end{itemize}
		We abuse the notation and henceforth denote by $E_1^{(1)}, E_{1,>}^{(2)},E_{1,<}^{(2)}$ the restrictions  of the Rellich functions  to their modified domains $I (E_1^{(1)}), I (E_{1,>}^{(2)}),  I (E_{1,<}^{(2)})$.
	\end{prop}
	\begin{proof}
		For the simple resonance case, recalling the two-monotonicity interval structure  (c.f. Proposition \ref{str1}),  we denote by $ (E_1^{(1)})_+$  (resp. $ (E_1^{(1)})_-$) the restriction of $E_1^{(1)}$ to the increasing  (resp. decreasing) interval on  $I_{i}^{(1)}$ and define 
		\begin{align*}
			J (E_1^{(1)})&:=\operatorname{Image} (E_1^{(1)})_+ \cap\operatorname{Image} (E_1^{(1)})_-,\\
			I (E_1^{(1)})_+ &:=  ( (E_1^{(1)})_+)^{-1} (J (E_1^{(1)}) ), \  I (E_1^{(1)})_-: =  ( (E_1^{(1)})_-)^{-1} (J (E_1^{(1)}) ),\\
			I (E_1^{(1)})&:=I (E_1^{(1)})_+\cup  I (E_1^{(1)})_-.
		\end{align*}
		We only need to verify \eqref{error}. By \eqref{C0}, we deduce 
		$$|J_i^{(1)}\Delta \operatorname{Image} (E_1^{(1)})_+|,\  |J_i^{(1)}\Delta \operatorname{Image} (E_1^{(1)})_-|\lesssim\varepsilon,$$
		where $X\Delta Y:= (X\setminus Y)\cup (Y\setminus X)$ for sets $X,Y$. 
		Thus $$|\operatorname{Image} (E_1^{(1)})_+\Delta \operatorname{Image} (E_1^{(1)})_-|\lesssim\varepsilon,$$
		from which we obtain 
		$$   |\operatorname{Image} (E_1^{(1)})_+\setminus J (E_1^{(1)})|, \ |\operatorname{Image} (E_1^{(1)})_-\setminus J (E_1^{(1)})| \lesssim\varepsilon.$$
		Thus by the Morse condition of $ (E_1^{(1)})_{\pm}$  (c.f. Proposition \ref{31}) and Lemma \ref{C2}, we get upper bound for    the contraction of   the preimage
		\begin{align*}
			&|I^{(1)}_i\setminus I (E_1^{(1)})|\\
			=&| ( (E_1^{(1)})_+)^{-1}  (\operatorname{Image} (E_1^{(1)})_+\setminus J (E_1^{(1)}))|+ | ( (E_1^{(1)})_-)^{-1}  (\operatorname{Image} (E_1^{(1)})_-\setminus J (E_1^{(1)}))|\\
			\lesssim&\sqrt{\varepsilon}< \delta_0^{5}.
		\end{align*}
		Thus we finish the proof of \textbf{(a)}.  
		
		The proof of \textbf{(b)} is analogous.  Recalling the two-monotonicity interval structure  (c.f. Proposition \ref{str2}),  we denote by $ (E_{1,\bullet}^{(2)})_+$  (resp. $ (E_{1,\bullet}^{(2)})_-$) the restriction of $E_{1,\bullet}^{ (2)}$ to the increasing  (resp. decreasing) interval on  $I_{k,\cup}^{ (2)}$  and  define 
		\begin{align*}
			J (E_{1,\bullet}^{(2)})&:=\operatorname{Image} (E_{1,\bullet }^{(2)})_+ \cap\operatorname{Image} (E_{1,\bullet}^{(2)})_-,\\
			I (E_{1,\bullet}^{(2)})_+ &:=  ( (E_{1,\bullet}^{(2)})_+)^{-1} (J (E_{1,\bullet}^{(2)})), \ I (E_{1,\bullet}^{ (2)})_- :=  ( (E_{1,\bullet}^{ (2)})_-)^{-1} (J (E_{1,\bullet}^{ (2)})), \\
			I (E_{1,\bullet}^{ (2)})&:=I (E_{1,\bullet}^{ (2)})_+\cup I (E_{1,\bullet}^{ (2)})_-, 
		\end{align*}
		where $\bullet\in \{>,<\}.$
		As the above argument,  we can employ  approximation \eqref{apbp} and the Morse condition of  $E_{1,>}^{(2)},E_{1,<}^{(2)}$ to prove
		$$|[\operatorname{inf} I^{(2)}_{k,-},\operatorname{sup}I^{(2)}_{k,+}-k\cdot\omega ]\setminus I (E_{1,>}^{(2)}) | \lesssim\sqrt{\varepsilon}< \delta_0^{5},$$
		$$|[\operatorname{inf} I^{(2)}_{k,+}-k\cdot \omega ,\operatorname{sup}I^{(2)}_{k,-}]\setminus I (E_{1,<}^{(2)}) | \lesssim\sqrt{\varepsilon}< \delta_0^{5}.$$
		Finally, 	\eqref{inside} follows from the uniform bound $D$ of the derivative of the Rellich functions and $|J^{(2)}_k|\sim\delta_0^{\frac{1}{100}}\gg\delta_0.$
	\end{proof}
	\begin{figure}[htp]
		\begin{tikzpicture}[>=latex, scale=1]
			\draw  (-6,0.6)--  (3,0.6);
			\draw[domain=-5.3:-2.7, samples=100] plot  ({\x},{\x*\x/6});
			\draw[domain=1.5:2.5, samples=100] plot  ({\x},{\x*\x/1.5});
			\draw[dashed,]  (-5,4.16)--  (-5,0.6); \draw[dashed]  (-3,1.5)--  (-3,0.6);
			\draw[dashed,]  (-5,4.16)--  (2.5,4.16); \draw[dashed]  (-3,1.5)--  (1.5,1.5);
			\draw[dashed]  (1.5,1.5)--  (1.5,0.6);\draw[dashed]  (2.5,4.16)--  (2.5,0.6);
			\draw[<->,red]  (-5,0.6)-- node[above]{{\tiny $I (E_{1}^{(1)})_-$}} (-3,0.6);
			\draw[<->,red]  (1.5,0.6)-- node[above]{{\tiny $I (E_{1}^{(1)})_+$}} (2.5,0.6); 
			\draw[<->]  (-1,1.5)-- node[fill=white]{$J (E_{1}^{(1)})$} (-1,4.16); 
			\draw (-3.1,2.5)node{$ (E_1^{(1)})_-$}; 	\draw (1.3,2.5)node{$ (E_1^{(1)})_+$};
			\draw[dashed,]  (-5.3,4.7)--  (-5.3,0.6);	\draw[dashed,]  (-2.7,1.22)--  (-2.7,0.6);
			\draw  (-5.3,0.3)--  (-5.3,0.6); 	\draw  (-2.7,0.3)--  (-2.7,0.6); 	
			\draw [<->] (-2.7,0.45)--node[below]{{\tiny $I_{i,-}^{(1)}$}}  (-5.3,0.45);
			\draw  (1.5,0.3)--  (1.5,0.6); 	\draw  (2.5,0.3)--  (2.5,0.6); 	
			\draw [<->] (1.5,0.45)--node[below]{{\tiny $I_{i,+}^{(1)}$}}  (2.5,0.45);
			\draw[green,line width=1.5pt]  (-2.7,0.6)--  (-3,0.6); 	\draw[<-]  (-2.85,0.65)--  (-2.4,0.9); 
			\draw  (-2.5,0.9)node[right]{{\tiny contraction$\leq \delta_0^5$}}; 
		\end{tikzpicture}
		\caption{A carton illustration of the domain adjustment of the simple resonant case with no  critical point: After a domain  contraction smaller than $\delta_0^5$, the   Rellich child $E_1^{(1)}$  has  the same image on  each of its   monotone   interval.}
	\end{figure}
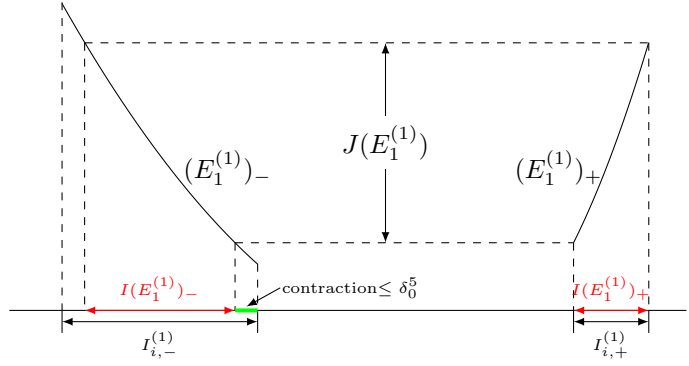
	%Since the contraction of the modified domains  (on order less than $\delta_0^5$) are much smaller than the   overlaps of the adjacent initial domains  (on order $\delta$ by  Proposition \ref{coverpr} and the uniform bound of derivative), the  modified domains and their relevant translates  also form a cover of $\T$.
	
	\subsubsection{Green's function estimates}
	In this part, we relate the first scale resonance to the Rellich function collection $\mathcal{C}_1$ we have constructed and establish  the Green's function estimates for $1$-nonresonant sets with $1$-regularity condition.

	In order to handle technicalities that arise near the boundaries of the Rellich functions, we  introduce the following definition for  ``modified codomains''.  With the notation from Proposition \ref{adj}, for $E_1^{(1)}\in \mathcal{C}_1^{(1)}$, we define 
	\begin{align*}
		\tilde{J} (E_1^{(1)}):=\left\{\begin{aligned}
			&[\inf J (E_1^{(1)})+\frac{9}{8}\delta_0,\sup J (E_1^{(1)})-\frac{9}{8}\delta_0] \text{ \  if $E_1^{(1)}$ has no critical point in $ I (E_1^{(1)})$,}\\
			& ( -\infty,\sup J (E_1^{(1)})-\frac{9}{8}\delta_0] \text{ \  if $E_1^{(1)}$ attains  minimum at a  critical point in $ I (E_1^{(1)})$,}\\
			&[\inf J (E_1^{(1)})+\frac{9}{8}\delta_0,+\infty)\text{ \  if $E_1^{(1)}$ attains  maximum at a  critical point in $ I (E_1^{(1)})$,}\\
		\end{aligned}\right. 
	\end{align*}
	and for $E_{1,\bullet}^{(2)}\in \mathcal{C}_1^{(2)} (\bullet\in \{>,<\})$,  we define  
	$$ \tilde{J} (E_{1,\bullet}^{(2)}):=[\inf J (E_{1,<}^{(2)})+\frac{9}{8}\delta_0,\sup J (E_{1,>}^{(2)})-\frac{9}{8}\delta_0].$$
	Fixing $\theta^*, E^*$, let $E_1^{(j)}\in \mathcal{C}_1$ be the  Rellich function such that $E^*\in \tilde{J} (E_1^{(j)}) \  (j\in \{1,2\}).$
	Define the set of $1$-resonant points   (relative to $ (\theta^*,E^*)$ and $E_1^{(j)}$)  as:\\
	\textbf{(1)}. $j=1$,
	$$S_1 (\theta^*,E^*):=\{x\in \Z^d:\ \theta^*+x\cdot\omega \in  I (E_1^{(1)}) , \ |E_1^{(1)} (\theta^*+x\cdot \omega)-E^*|<\delta_1^{(1)}\},$$
	\textbf{(2)}.  $j=2$,
	$$S_1 (\theta^*,E^*):=\{x\in \Z^d:\ \theta^*+x\cdot\omega \in  I (E_{1,>}^{(2)})\cup I (E_{1,<}^{(2)}) , \ \min_{\bullet\in \{>,<\}}|E_{1,\bullet}^{(2)} (\theta^*+x\cdot \omega)-E^*|<\delta_1^{(2)}\},$$
	where  $\delta_1^{(j)}:=e^{- (l_1^{(j)})^{\frac{2}{3}}}$.
	
	We say that a set $\Lambda\subset \Z^d$ is $1$-nonresonant if $\Lambda\cap   S_1 (\theta^*,E^*)=\emptyset$. Moreover, we say that a set $\Lambda\subset \Z^d$ is $1$-regular if $(\Lambda_{2l_1^{(j)}}+x)\subset\Lambda$ for any $x\in S_0 (\theta^*,E^*)\cap\Lambda.$ %with $B_1^{(j)}+x\cap\Lambda\neq \emptyset$.% 
	Finally, we say that a set is $1$-good if it is both $1$-nonresonant and $1$-regular. With the above definitions, we can now prove 
	\begin{thm}\label{1g}  Fix $\theta^*, E^*, E_1^{(j)}\in \mathcal{C}_1$  such that $E^*\in \tilde{J} (E_1^{(j)})$  and a finite set $\Lambda\subset \Z^d$.   If $\Lambda$ is $1$-good, then  for $|\theta-\theta^*|<\delta_1^{(j)}/ (10D), |E-E^*|<\delta_1^{(j)}/5$, 
		\begin{align*}
			\|G_\Lambda (\theta,E)\|&\leq10 (\delta_1^{(j)})^{-1},\\
			\ |G_\Lambda (\theta,E;x,y)|&\leq e^{-\gamma_1\|x-y\|_1}, \  \|x-y\|_1\geq  (l_1^{(j)})^{\frac{5}{6}},
		\end{align*}
		where $\gamma_1= (1-O ( (l_1^{(j)})^{-\frac{1}{30}}))\gamma_0$.
	\end{thm}
	\begin{proof}
		Recalling  the set 	$S_0 (\theta^*,E^*)$  (c.f. \eqref{S0}), we  define $S_0^{\Lambda}:= S_0 (\theta^*,E^*)\cap \Lambda $. If $S_0^{\Lambda}=\emptyset$, then $\Lambda$ is $0$-nonresonant and the theorem follows from Theorem \ref{0ge}. We thus assume $S_0^{\Lambda}$ is nonempty. In order to apply resolvent identity, we wish to  find a finite family $\{B_1^{(j)}+p\}_{p\in \bar{S}_0^{\Lambda} }$ of translations  of $B_1^{(j)}$ such that each $x\in S_0^{\Lambda}$ is near the center of $B_1^{(j)}+p$ for some $p\in \bar{S}_0^{\Lambda}$ and $\|G_{B_1^{(j)}+p} (\theta^*,E^*)\|\leq  (\delta_1^{ (j)})^{-1}$
		for all $p\in \bar{S}_0^{\Lambda}$. We do so by case analysis:\\
		\textbf{(1).} $j=1$: Assume that $E^*\in \tilde{J} (E_1^{(1)})$. 	We first prove  that  for any  $x\in S_0^{\Lambda}$,  $\theta^*+x\cdot \omega \in  I (E_1^{(1)})$. The proof proceeds also by case analysis:
		\begin{itemize}
			\item [\textbf{(a)}.] $E_1^{(1)}$ has no critical point in $ I (E_1^{(1)})$: In this case, $ I (E_1^{(1)})$ is a union of two disjoint intervals. Since $E^*\in\tilde{J} (E_1^{(1)})$ and $x\in S_0^{\Lambda}$, by the definition of $ \tilde{J} (E_1^{(1)})$, we obtain 
			$$v (\theta^*+x\cdot \omega )\in [\inf J (E_1^{(1)})+\frac{1}{8}\delta_0,\sup J (E_1^{(1)})-\frac{1}{8}\delta_0].$$
			Since $j=1$, by approximation \eqref{C0}, we have 
			$$[\inf J (E_1^{(1)})+\frac{1}{8}\delta_0,\sup J (E_1^{(1)})-\frac{1}{8}\delta_0]\subset v_\pm (I (E_1^{(1)})_\pm ).$$ 
			Thus $v (\theta^*+x\cdot \omega )\in v_\pm (I (E_1^{(1)})_\pm )$ and so $\theta^*+x\cdot \omega  \in  I (E_1^{(1)}) $.
			\item [\textbf{(b)}.] $E_1^{(1)}$ has a  critical point in $ I (E_1^{(1)})$: In this case, $ I (E_1^{(1)})$ is a single interval. Suppose that this critical point is a minimum  (the other case is analogous), then $v$ likewise achieves its minimum at a critical point in $ I (E_1^{(1)})$. As above,  we have 
			$$v (\theta^*+x\cdot \omega ) \in  [\inf J^{(1)}_{i},\sup J (E_1^{(1)})-\frac{1}{8}\delta_0].$$
			and 
			$$ [\inf J^{(1)}_{i},\sup J (E_1^{(1)})-\frac{1}{8}\delta_0]   \subset v_\pm (I (E_1^{(1)})_\pm).$$
			Thus  $\theta^*+x\cdot \omega\in I (E_1^{(1)})$.   
		\end{itemize}
		For $j=1$, we define $\bar{S}_0^{\Lambda}=S_0^{\Lambda}$. \\
		Since $|v (\theta^*+x\cdot\omega)-E^*|\leq \delta_0$ and $\Lambda$ is $1$-nonresonant, we have 
		$$ \delta_1^{(1)}\leq |E_1^{(1)} (\theta^*+x\cdot\omega)-E^*|\leq \frac{9}{8}\delta_0.$$
		By the eigenvalue separation estimates from Proposition \ref{k1}, it follows that $|\hat{E}-E^*|\geq \delta_1^{(1)}$ for any eigenvalue $\hat{E}$ of $H_{B_1^{(1)}+x} (\theta^*)$, and so 
		$$\|G_{B_1^{(1)}+x} (\theta^*,E^*)\|\leq  (\delta_1^{(1)})^{-1}.$$
		\textbf{(2).} $j=2$: Assume that $E^*\in \tilde{J} (E_1^{(2)})$ and $E_{1,\bullet}^{(2)}$ $  (\bullet\in\{>,<\})$ are the Rellich children generated from the double resonance interval  $J_k^{(2)}$.  We first prove  that for any  $x\in S_0^{\Lambda}$, $\theta^*+p\cdot \omega \in   I (E_{1,>}^{(2)})\cup I (E_{1,<}^{(2)})$ for some $p\in \{x,x-k\}$.\\ Since $E^*\in\tilde{J} (E_1^{(2)})$ and $x\in S_0^{\Lambda}$, by the definition of $ \tilde{J} (E_1^{(2)})$, we obtain 
		$$v (\theta^*+x\cdot \omega )\in [\inf J (E_{1,<}^{(2)})+\frac{1}{8}\delta_0,\sup J (E_{1,>}^{(2)})-\frac{1}{8}\delta_0].$$
		Since $j=2$, by approximation \eqref{apbp}, we have 
		$$[\inf J (E_{1,<}^{(2)})+\frac{1}{8}\delta_0,\sup J (E_{1,>}^{(2)})-\frac{1}{8}\delta_0]\subset v_-\left ([\inf I (E_{1,>}^{(2)}), \sup I (E_{1,<}^{(2)})]\right) $$
		and 
		$$[\inf J (E_{1,<}^{(2)})+\frac{1}{8}\delta_0,\sup J (E_{1,>}^{(2)})-\frac{1}{8}\delta_0]\subset v_+\left ([\inf I (E_{1,<}^{(2)}), \sup I (E_{1,>}^{(2)})]+k\cdot \omega \right).$$
		Thus $\theta^*+x\cdot \omega\in [\inf I (E_{1,>}^{(2)}), \sup I (E_{1,<}^{(2)})]\cup [\inf I (E_{1,<}^{(2)}), \sup I (E_{1,>}^{(2)})]+k\cdot \omega$. It follows that $\theta^*+x\cdot \omega \in 
		[\inf I (E_{1,>}^{(2)}), \sup I (E_{1,<}^{(2)})]$ or $\theta^*+ (x-k)\cdot \omega \in 
		[\inf I (E_{1,<}^{(2)}), \sup I (E_{1,>}^{(2)})]$. Thus  $\theta^*+p\cdot \omega \in   I (E_{1,>}^{(2)})\cup I (E_{1,<}^{(2)})$ for some $p\in \{x,x-k\}$.\\
		For $j=2$, we define $\bar{S}_0^{\Lambda}$ to be the set of all point $p$ coming from  $x\in S_0^{\Lambda}$ as above. 
		Since $\Lambda$ is $1$-regular and $\|k\|_1\leq 10l_1^{(1)}<2l_1^{(2)}$, it follows that $\bar{S}_0^{\Lambda}\subset \Lambda$ and  since $\Lambda$ is $1$-nonresonant,   by the eigenvalue separation estimates from Proposition \ref{k2}, it follows that  $|\hat{E}-E^*|\geq \delta_1^{(2)}$ for any eigenvalue $\hat{E}$ of $H_{B_1^{(2)}+p} (\theta^*)$, and so 
		$$\|G_{B_1^{(2)}+p} (\theta^*,E^*)\|\leq  (\delta_1^{(2)})^{-1}.$$
		For both $j=1,2$, by the $1$-regularity of $\Lambda$,  $(B_1^{(j)}+p)\subset\Lambda$ for all $p\in \bar{S}_0^{\Lambda}$. Moreover, by Diophantine condition, the separation of $p\in \bar{S}_0^{\Lambda}$ is larger than the size of $B_1^{(j)}$, thus the blocks $\{B_1^{(j)}+p\}_{p\in \bar{S}_0^{\Lambda} }$ are non-overlapping. The Green's function estimates follow  from a standard application of resolvent identity as we will show below. Since $j$ has little effect in the proof, we  omit the superscript ``$ (j)$''.  
		
		First,  we prove the case  that  $\Lambda=B_1+p$ is a  single $1$-nonresonant block. By the previous discussion, we have  
		\begin{equation*}
			\|G_{\Lambda} (\theta^*,E^*)\|\leq\delta_1^{-1}.
		\end{equation*}
		It follows from Neumann series argument that  for $|\theta-\theta^*|<\delta_1/ (10D)$ and  $|E-E^*|<\frac{2}{5}\delta_1$, 
		\begin{equation*}\label{L2}
			\|G_{\Lambda } (\theta,E)\|\leq2\delta_1^{-1}.
		\end{equation*}
		Let $x,y\in \Lambda $ satisfy  $\|x-y\|_1\geq l_1^\frac{4}{5}$. Since $G_{\Lambda }$ is self-adjoint, we may  assume $\|x-p\|_1\geq l_1^\frac{3}{4}$. Let $\Omega_{1,p}$ be an $l_1^\frac{2}{3}$-size cube centered at $p$. Then $\Lambda \setminus \Omega_{1,p}$ is $0$-nonresonant since $S_0^\Lambda\subset \Omega_{1,p}$.  We omit the dependence on $E,\theta$ of Green's functions.  By the  resolvent identity, we obtain 
		\begin{align*}
			|G_{\Lambda } (\theta,E;x,y)|&=|G_{\Lambda \setminus \Omega_{1,p}} (x,y)\chi (y)+\sum_{z,z'}G_{\Lambda \setminus \Omega_{1,p}} (x,z)\Gamma_{z,z'}G_{\Lambda } (z',y)|\\
			&\leq e^{-\gamma_0\|x-y\|_1}+C (d)\sup_{z,z'}e^{-\gamma_0\|x-z\|_1}|G_{\Lambda } (z',y)|\\
			&\leq e^{-\gamma_0\|x-y\|_1}+C (d)\sup_{z,z'}e^{-\gamma_0\|x-z\|_1}e^{-\gamma_0 (\|z'-y\|_1-l_1^\frac{3}{4})}\delta_1^{-1}\\
			&\leq e^{-\gamma'_0\|x-y\|_1}
		\end{align*}
		with  $\gamma'_0= (1-O (l_1^{-\frac{1}{30}}))\gamma_0$, where we used if $\|z'-y\|\leq l_{1}^\frac{3}{4}$, then 
		\begin{equation*}
			|G_{\Lambda} (z',y)|\leq 	\|G_{\Lambda}\|\leq2\delta_{1}^{-1}\leq 2e^{-\gamma_0 (\|z'-y\|_1-l_{1}^\frac{3}{4})}\delta_{1}^{-1}, 
		\end{equation*}
		if $\|z'-y\|\geq l_{1}^\frac{3}{4}$, then 
		\begin{align*}
			|G_{\Lambda} (z',y)|=	|G_{\Lambda} (y,z')|&\leq 	\sum_{w,w'}|G_{\Lambda \setminus \Omega_{1,p}} (y,w)\Gamma_{w,w'}G_{\Lambda} (w',z')|\\
			&\leq C (d)e^{-\gamma_0\|y-w\|_1}\|G_{\Lambda}\| \\
			&\leq C (d)	e^{-\gamma_0 (\|y-z'\|_1-l_{1}^\frac{3}{4})}\delta_{1}^{-1}
		\end{align*} and  $\delta_1^{-1}=e^{l_1^\frac{2}{3}}\ll e^{\gamma_0\|x-y\|_1}$ to bound the second term. Thus  we finish  the case that  $\Lambda$ is a single $1$-nonresonant block. 
		
		Now  assume $\Lambda$ is an arbitrary $1$-good set. We must show that $G_\Lambda$ does exist. By Schur's test, it suffices to prove 
		\begin{equation}\label{Schur}
			\sup_x\sum_{y}|G_\Lambda (\theta,E+io;x,y)|<C<\infty.
		\end{equation}
		Define  $$\Lambda':=\Lambda\setminus\bigcup_{p\in \bar{S}_0^{\Lambda}} \Omega_{1,p}.$$ Then $\Lambda'$ is $0$-nonresonant  since $S_0^{\Lambda}\subset\cup_{p\in \bar{S}_0^{\Lambda}} \Omega_{1,p}.$ Let $\tilde{\Omega}_{1,p}$ be  a  $2l_1^\frac{2}{3}$-size cube centered at $p$.  For $x\in \Lambda\setminus\cup_{p\in \bar{S}_0^{\Lambda}} \tilde{\Omega}_{1,p}$, by resolvent identity, we have
		\begin{align*}
			\sum_y|G_\Lambda (x,y)|&\leq \sum_y|G_{\Lambda'} (x,y)|+\sum_{z,z',y}|G_{\Lambda'} (x,z)\Gamma_{z,z'}G_{\Lambda} (z',y)|\\
			&\leq	C (d)\delta_0^{-1}+	C (d)e^{-l_1^\frac{2}{3}}\sup_{z'}\sum_y|G_{\Lambda} (z',y)|.
		\end{align*}
		For $x\in \tilde{\Omega}_{1,p}$, by resolvent identity, we have
		\begin{align*}
			\sum_y|G_\Lambda (x,y)|&\leq \sum_y|G_{B_1+p} (x,y)|+\sum_{z,z',y}|G_{B_1+p} (x,z)\Gamma_{z,z'}G_{\Lambda} (z',y)|\\
			&\leq \delta_1^{-2}+C (d)e^{-\frac{1}{2}l_1}\sup_{z'}\sum_y|G_{\Lambda} (z',y)|.
		\end{align*}
		Taking supremum for $x$ on the left hand side of the above two inequalities, we get $$\sup_x\sum_y|G_\Lambda (x,y)|\leq \delta_1^{-2}+\frac{1}{2}\sup_x\sum_y|G_\Lambda (x,y)|,$$ thus 
		$$\sup_x\sum_y|G_\Lambda (x,y)|\leq 2\delta_1^{-2},$$
		which gives \eqref{Schur}.
		Note that  for $|\theta-\theta^*|<\delta_1/ (10D)$ and $|E-E^*|<\frac{2}{5}\delta_1$, $G_\Lambda (\theta,E)$ does exist.  Then  we obtain   $\operatorname{dist} (\sigma (H_\Lambda (\theta)),E^*)\geq \frac{2}{5}\delta_1$ and  $\operatorname{dist} (\sigma (H_\Lambda (\theta)),E)\geq \frac{1}{5}\delta_1$ for all $|E-E^*|<\frac{1}{5}\delta_1$. Thus we get the desired  operator norm 
		$$\|G_\Lambda (\theta,E)\|=\|\left ( H_\Lambda (\theta)-E\right) ^{-1}\|=\frac{1}{\operatorname{dist} (\sigma (H_\Lambda (\theta)),E)}\leq10\delta_1^{-1}.$$
		It remains to  prove the off-diagonal decay of $G_\Lambda$.
		Let $x,y\in \Lambda$ be such that $\|x-y\|_1\geq l_1^\frac{5}{6}$. We define 
		\[B_x:=\left\{\begin{aligned}
			&\Lambda_{\l_1^\frac{1}{2}} (x)\cap\Lambda  \quad \text{if }   x\in \Lambda\setminus\bigcup_{p\in \bar{S}_0^{\Lambda}} \tilde{\Omega}_{1,p},\text{ \  (Choice A)}  \\
			&B_1+p\   \quad\quad \  \text{if }  x \in\tilde{\Omega}_{1,p}. \text{\  (Choice B)} \ 
		\end{aligned}\right. \]
		Then  $B_x$ has the following two properties: \textbf{(1).} $B_x$ is either   a   $0$-nonresonant  set (Choice A)  or a $1$-nonresonant  block (Choice B); \textbf{(2).} $x$ is far away from the  relative boundary $\partial_\Lambda B_x$. We can iterate the resolvent identity  to obtain 
		\begin{align}
			|G_\Lambda (x,y)|&\leq\prod_{s=0}^{L-1}  (C (d) l_{1}^d e^{-\gamma_0'\|x_{s}-x_{s+1}\|_1})|G_\Lambda (x_L,y)|\nonumber \\
			&\leq e^{-\gamma_0''\|x-x_L\|_1}|G_\Lambda (x_L,y)|, \label{728}
		\end{align}
		where $x_0:=x$ and  $x_{s+1}\in \partial B_{x_{s}}$ with  $\|x_{s+1}-x_s\|_1\geq l_1^\frac{1}{2}$ for  Choice A and  $\|x_{s+1}-x_s\|_1\geq \frac{1}{2}l_1\gg l_1^{\frac{5}{6}}$ for  Choice B. We stop the iteration until  $y\in B_{x_L}$. 
		By  resolvent identity, we get 
		\begin{align}
			|G_\Lambda (x_L,y)|&\leq|G_{B_{x_L}} (x_L,y)|+\sum_{z,z'}|G_{B_{x_L}} (x_L,z)\Gamma_{z,z'}G_{\Lambda} (z',y)|\nonumber\\
			&\leq C (d) e^{-\gamma_0' (\|x_L-y\|_1-l_1^\frac{4}{5})}\delta_1^{-1},\label{728.}
		\end{align}
		where we have used the off-diagonal exponential decay of $G_{B_{x_L}}$ and the operator norm estimate  $\|G_{\Lambda}\|\leq10\delta_1^{-1}$. Since $\|x-y\|_1\geq l_1^\frac{5}{6}$, \eqref{728} together with  \eqref{728.} gives the desired off-diagonal estimate
		$$|G_\Lambda (x,y)|\leq e^{-\gamma_1\|x-y\|_1}$$
		with $\gamma_1= (1-O (l_1^{-\frac{1}{30}}))\gamma_0$. Thus we finish the proof.
	\end{proof}
	\subsection{Induction  definitions and propositions}\label{H} 
	In this part, we will list the most important properties of the Rellich functions $\mathcal{C}_n$ in our induction hypothesis. The induction will proceed as follows: we suppose that we have  inductively
	constructed   Rellich functions $E_n:I (E_n)\rightarrow J (E_n)$ satisfying Morse condition and two-monotonicity interval structure and established   Green’s function estimates on $n$-good sets  (relative to  $E_n$) for   energies $E^*\in \tilde{J} (E_n)$, where  $\tilde{J} (E_n)$  is an energy interval  near
	the codomain of $E_n$.  Under these assumptions, we first classify  energy regions $J^{(j)}_i\subset J (E_n)$ as being simple or double resonant  (indicated by $j\in \{1,2\}$)  and prevent recurrence to those energy regions for long times. For each energy region, the resonant sites for $E_n$  will be so well-separated that we can find a block $B_{n+1}^{(j)}$ of $l_{n+1}^{(j)}$-size having a large annulus  with Green’s function estimates  by our induction  assumptions. Then we employ the eigenvalue perturbation theory to  construct  the  Rellich function $E_{n+1}^{(j)}$  of $H_{B_{n+1}^{(j)}}$ maintaining the  induction  structure near $E_n$. Finally, we prove  Green's function estimates on $(n+1)$-good sets.

	We say that  $E_{n+1}^{(j)}$ is a child of $E_n$ and  $E_n$ is the parent  of $E_{n+1}^{(j)}$, denoted by $E_{n+1}^{(j)}\in \mathcal{C}^{(j)} (E_n)$.  
	We introduce some notation to refer to different parts of our Rellich tree. Recall that $\mathcal{C} (E_n)$ denotes the collection of all immediate children of the Rellich curve $E_n$, and that $\mathcal{C}_n$ denotes the collection of all  $n$-scale Rellich functions, i.e., the $n^{\text {th }}$ ``generation'' of the tree. For a Rellich function $E_m \in \mathcal{C}_m, m \leqslant n$, we denote by $\mathcal{C}_n (E_m)$ the collection of all the   $n$-scale descendants of $E_m$; i.e., for a $m$-scale Rellich function $E_m$, $\mathcal{C}_m (E_m)=\left\{E_m\right\}$, and for all $s \geqslant m$, if $E_s \in \mathcal{C}_s\left (E_m\right)$ and $E_{s+1} \in \mathcal{C}\left(E_s\right)$, then $E_{s+1} \in \mathcal{C}_{s+1}\left(E_m\right)$.   For $0\leq m\leq n$, we say that  $E_n$ is a descendant of $E_m$ and  $E_m$ is an ancestor   of $E_n$ if $E_n\in \mathcal{C}_n (E_m)$.
	
	Now we define the induction  parameters: 
	$$l_{n}^{(j)}:= (|\log\varepsilon_0|^{4^n})^j, \ n\geq 1, \ j \in \{1,2\} \quad \text{ (the $n$-scale length)}   $$
	and 
	$$\delta_{n}^{(j)}:=e^{- (l_{n}^{(j)})^\frac{2}{3}}, \ n\geq 1, \ j \in \{1,2\} \quad \text{ (the $n$-scale resonance)},$$ 
	where the superscript ``$ (j)$'' may be  omitted if not emphasized.
	
	Subject to these definitions, we suppose  that the following induction  proposition holds.
	\begin{prop}[Induction hypotheses of scale $n$]\label{induc}
		\begin{hp}\label{h1}
			For $1\leq m\leq n$, we have collections
			\begin{align*}
				\mathcal{C}_m^{(j)}&=\bigcup_{E_{m-1}\in \mathcal{C}_{m-1}}\mathcal{C}^{(j)} (E_{m-1}), \ \ j\in \{1,2\},\\
				\mathcal{C}_m&=\mathcal{C}_m^{(1)}\cup\mathcal{C}_m^{(2)}
			\end{align*}
			of Rellich functions $$E_m^{(j)}:I (E_m^{(j)})\rightarrow J (E_m^{(j)})$$ with $\delta_{m-2}^{10}\lesssim |I (E_m^{(j)})|\lesssim \delta_{m-2}^5, |J (E_m^{(j)})|\sim\delta_{m-2}^{10}$ for $m\geq 2$ and $\delta_0^{\frac{1}{100}}\lesssim|I (E_1^{(j)})| \lesssim \delta_0^{\frac{1}{200}}, |J (E_1^{(j)})|\sim \delta_0^{\frac{1}{100}}$ for $m=1$  of  $H_{B_m^{(j)}}$, where $B_m^{(j)}$ is a $l_{m}^{(j)}$-size block centered at origin.
		\end{hp}
		\begin{hp}\label{h2}
			Each $E_m\in \mathcal{C}_m$ has a  two-monotonicity interval structure with the same image on each monotonicity component, that is, there are two closed intervals $I (E_m)_\pm$ with disjoint interiors, such that 
			$$  I (E_m)=I (E_m)_+\cup I (E_m)_-,$$
			$$\pm E_m'|_{I (E_m)_\pm}\geq 0,  $$
			$$E_m (I (E_m)_+)=E_m (I (E_m)_-)=J (E_m).$$
		\end{hp}
		\begin{hp}\label{h3}
			Each $E_m\in \mathcal{C}_m$  satisfies the Morse condition:\\ 
			If $$|E_m' (\theta)|\leq \delta_{m-2}^{2}\  (\delta_{-1}:=\delta_0^\frac{1}{400}),$$ then  $$|E_m'' (\theta)|\geq 3-\sum_{s=0}^{n-1}\delta_s^{10}\geq 2$$
			with a unique sign for all such $\theta$. 
		\end{hp}
		\begin{hp}\label{h4}
			The  Rellich functions can be classified into the following three types:
			\begin{tp}\label{t1}
				$E_m$ has no critical point  and $I (E_m)_+$ and $I (E_m)_-$ are two disjoint intervals, and the derivative of $E_m$ on each interval satisfies 
				$$\pm E_m'|_{I (E_m)_\pm}\gtrsim  \delta_{m-2}^{10}.$$
				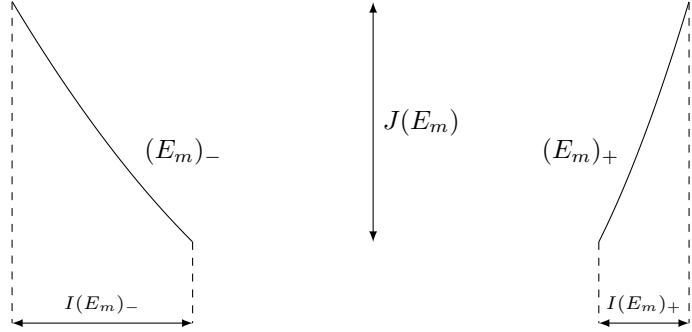
\begin{figure}[htp]
					\begin{tikzpicture}[>=latex, scale=1.2]
						\draw[domain=-5:-3, samples=100] plot  ({\x},{\x*\x/6});
						\draw[domain=1.5:2.5, samples=100] plot  ({\x},{\x*\x/1.5});
						\draw[dashed,]  (-5,4.16)--  (-5,0.6); \draw[dashed]  (-3,1.5)--  (-3,0.6);
						\draw[dashed]  (1.5,1.5)--  (1.5,0.6);\draw[dashed]  (2.5,4.16)--  (2.5,0.6);
						\draw[<->]  (-5,0.6)-- node[above]{{\scriptsize $I (E_m)_-$}} (-3,0.6);
						\draw[<->]  (1.5,0.6)-- node[above]{{\scriptsize $I (E_m)_+$}} (2.5,0.6); 
						\draw[<->]  (-1,1.5)-- node[right]{$J (E_m)$} (-1,4.16); 
						\draw (-3.1,2.5)node{$ (E_m)_-$}; 	\draw (1.3,2.5)node{$ (E_m)_+$};
						
					\end{tikzpicture}
					\caption{A cartoon illustration of $E_m$ belonging to \textbf{Type 1}.}
				\end{figure}
			\end{tp} 
			\begin{tp}\label{t2}
				$E_m$ has a critical point and $I (E_m)$ is a single interval. 
				
				%Moreover, the derivatives of $E_m$ on the edges of the interval  satisfy 
				%  $$\left|E_m'|_{\partial I (E_m)}\right|\gtrsim  \delta_{m-2}^{10}$$ with different signs.\\
				
				\begin{figure}[htp]
					\begin{tikzpicture}[>=latex, scale=1]
						\draw[domain=-4:0, samples=100] plot  ({\x},{\x*\x/5});
						\draw[domain=0:2, samples=100] plot  ({\x},{\x*\x/1.25});
						\draw[dashed,]  (0,0)--  (0,-1);
						\draw[dashed]  (-4,3.2)--  (-4,-1);	\draw[dashed]  (2,3.2)--  (2,-1);
						\draw[<->]  (0,-1)-- node[above]{$I (E_m)_-$} (-4,-1);
						\draw[<->]  (0,-1)-- node[above]{$I (E_m)_+$} (2,-1); 
						\draw[<->]  (0,0)-- node[right]{$J (E_m)$} (0,3.2); 
						\draw (-2,1.5)node{$E_m$};
					\end{tikzpicture}
					\caption{A cartoon illustration of $E_m$ belonging to \textbf{Type 2} (the minimum case).}
				\end{figure}
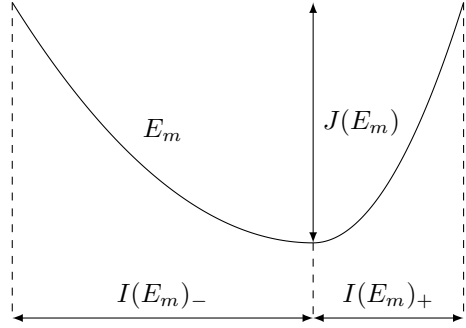
				
			\end{tp}
			\begin{tp}\label{t3}
				$E_m$ has a Rellich  brother and we denote the larger  one by $E_{m,>}$ and the smaller one by $E_{m,<}$.   Both of their domains  are single interval with a critical point (minimum for $E_{m,>}$ and maximum for $E_{m,<}$) in their common domain $I (  E_{m,>}) \cap I (E_{m,<})\neq \emptyset$.\\ The two Rellich functions have resonance near their critical points:   $$0  \leq  \inf J (  E_{m,>}) -\sup J (E_{m,<})\leq 3\delta_m, $$ 
				and they   are separated by  two interlacing functions $\lambda_{m,\pm}$:  
				$$	E_{m,<} (\theta)\leq \lambda_{m,\pm} (\theta) \leq E_{m,>} (\theta),$$
				$$ \pm \lambda_{m,\pm}' (\theta) \geq \delta_{m-2}.$$ 
				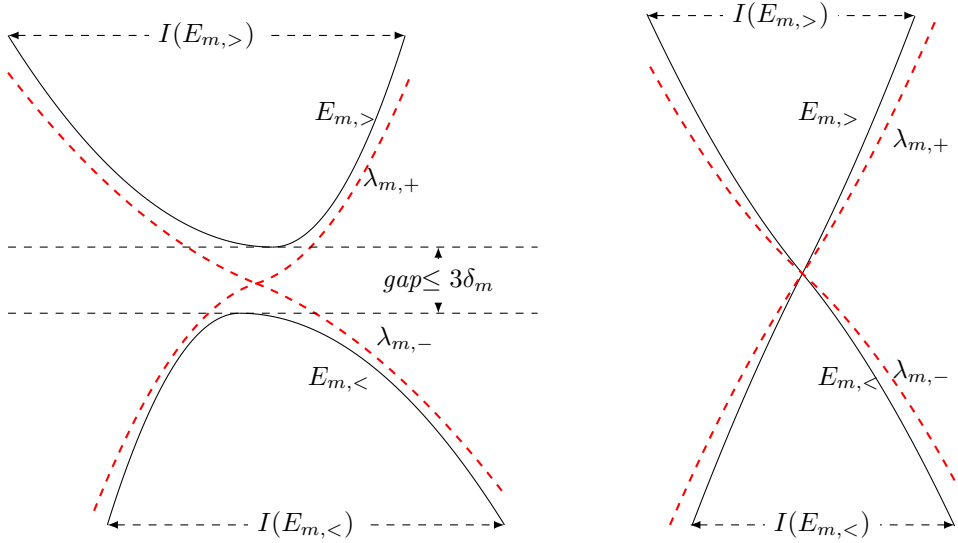
\begin{figure}[htp]
					\begin{tikzpicture}[>=latex, scale=0.88]
						\draw[domain=-4:0, samples=100] plot  ({\x},{\x*\x/5});
						\draw[domain=0:2, samples=100] plot  ({\x},{\x*\x/1.25});
						\draw[<->,dashed]  (-4,3.2)-- node[fill=white]{$I (E_{m,>})$} (2,3.2); 
						\draw[domain=-2:0, samples=100] plot  ({\x-0.5},{-1-\x*\x/1.25});
						\draw[domain=0:4, samples=100] plot  ({\x-0.5},{-1-\x*\x/5});
						\draw[<->,dashed]  (-2.5,-4.2)-- node[fill=white]{$I (E_{m,<})$} (3.5,-4.2); 
						\draw[domain=-4:-0.25, samples=100,red,dashed,thick] plot  ({\x},{-0.4*(\x+0.25)+0.12*(\x+0.25)*(\x+0.25)-0.55});
						\draw[domain=-0.25:3.5, samples=100,red,dashed,thick] plot  ({\x},{-0.4*(\x+0.25)-0.12*(\x+0.25)*(\x+0.25)-0.55});
						\draw[domain=-0.25:2.1, samples=100,red,dashed,thick] plot  ({\x},{0.3*(\x+0.25)+0.45*(\x+0.25)*(\x+0.25)-0.55});
						\draw[domain=-2.7:-0.25, samples=100,red,dashed,thick] plot  ({\x},{0.3*(\x+0.25)-0.45*(\x+0.25)*(\x+0.25)-0.55});
						\draw (1.8,1)node{$\lambda_{m,+}$}; 		\draw (1.1,2)node{$E_{m,>}$};  		\draw (2,-1.4)node{$\lambda_{m,-}$}; \draw (1,-2)node{$E_{m,<}$}; 
						\draw[dashed]  (-4,0)--  (4,0); 	\draw[dashed]  (-4,-1)--  (4,-1); 
						\draw[dashed,<->]  (2.5,0)--node[fill=white]{gap$\leq 3\delta_m$}  (2.5,-1); 
						%	\draw[domain=-0.73*pi:0.73*pi, samples=100,red] plot  ({\x+8},{6*sin ( (0.3*\x r) %-1});		
						\draw[domain=-1.67:0, samples=100] plot  ({\x+8},{1.95*\x-\x*\x/5-0.4});
						\draw[domain=1.7:0, samples=100] plot  ({\x+8},{1.95*\x+\x*\x/5-0.4});
						\draw[domain=-2.35:0, samples=100] plot  ({\x+8},{-1.2*\x+\x*\x/5-0.4});
						\draw[domain=2.3:0, samples=100] plot  ({\x+8},{-1.2*\x-\x*\x/5-0.4});
						\draw[<->,dashed]  (5.65,3.5)-- node[fill=white]{$I (E_{m,>})$} (9.7,3.5);
						\draw[<->,dashed]  (6.3,-4.2)-- node[fill=white]{$I (E_{m,<})$} (10.3,-4.2);  
						\draw[domain=-2:0, samples=100,red,dashed,thick ] plot  ({\x+8},{1.5*\x-\x*\x/5-0.4});
						\draw[domain=2:0, samples=100,red,dashed,thick ] plot  ({\x+8},{1.5*\x+\x*\x/5-0.4});
						\draw[domain=-2.3:0, samples=100,red,dashed,thick ] plot  ({\x+8},{-0.9*\x+\x*\x/5-0.4});
						\draw[domain=2.3:0, samples=100,red,dashed,thick ] plot  ({\x+8},{-0.9*\x-\x*\x/5-0.4});
						
						\draw (9.8,1.6)node{$\lambda_{m,+}$}; 		\draw (8.4,2)node{$E_{m,>}$};  		\draw (9.8,-1.9)node{$\lambda_{m,-}$}; \draw (8.7,-2)node{$E_{m,<}$}; 
					\end{tikzpicture}
					\caption{A cartoon illustration of   of $E_m$ belonging to \textbf{Type 3}: the graph on the left shows the eigenvalue separation case and the  graph on the  right shows the level crossing case.}
				\end{figure}
			\end{tp} 
			Moreover, if $E_m (\theta)$ belongs to  \textbf{Type} \ref{t1} or \ref{t2}, then $E_m (\theta)$ is the unique eigenvalue of $H_{B_m^{(1)}} (\theta)$  in the interval $[E_m (\theta)-3\delta_m,E_m (\theta)+3\delta_m]$, and if $E_m (\theta)$ belongs to   \textbf{Type} \ref{t3}, then  $E_{m,<} (\theta)$ and $E_{m,>} (\theta)$ are the only two eigenvalues of $H_{B_m^{(2)}} (\theta)$  in the interval  $ [\inf J  (E_{m,<})-\frac{1}{2}\delta_{m-2},\sup J ( E_{m,>})+\frac{1}{2}\delta_{m-2}]$. In all cases, the corresponding eigenfunction $\psi_m$ decays exponentially: 
			\begin{equation}\label{Hdecay}
				|\psi_m (x)|\leq e^{-\frac{1}{4}\gamma_0\|x\|}, \  \|x\|_1\geq l_m^\frac{6}{7}.
			\end{equation}
		\end{hp}
		\begin{hp}\label{h5}
			We introduce the following  definitions of ``modified domain'' $\tilde{J} (E_m)$ and  $\tilde{\tilde{J}} (E_m)$:\\
			With the notation from {\rm Hypothesis \ref{h4}}, if $E_m$ belongs to  \textbf{Type} \ref{t1} or \ref{t2},  we define 
			\begin{align*}
				\tilde{J} (E_m):=	\left\{\begin{aligned}
					&[\inf J (E_m)+\frac{9}{8}\delta_{m-1},\sup J (E_m)-\frac{9}{8}\delta_{m-1}] \text{ \  if $E_m$ has no critical point in $ I (E_m)$,}\\
					&[\inf J (E_m)-2\delta_m,\sup J (E_m)-\frac{9}{8}\delta_{m-1}] \text{ \  if $E_m$ has  minimum at a  critical point in $ I (E_m)$,}\\
					&[\inf J (E_m)+\frac{9}{8}\delta_{m-1}, \sup J (E_m)+2\delta_m]\text{ \  if $E_m$ has  maximum at a  critical point in $ I (E_m)$,}\\
				\end{aligned}\right. 
			\end{align*}
			and  if $E_m$ belongs to \textbf{Type} \ref{t3},  we define  
			$$ \tilde{J} (E_m):=[\inf J (E_{m,<})+\frac{9}{8}\delta_{m-1},\sup J (E_{m,>})-\frac{9}{8}\delta_{m-1}].$$
			$\tilde{\tilde{J}} (E_m)$ is defined similarly, with $\frac{5}{4}$ replacing $\frac{9}{8}$.\\
			The codomain of $E_{m-1}$ is almost covered by the codomains  of its children in the following sense:\\
			If $E_{m-1}$ does not attain its maximum at a critical point, then
			$$
			\sup \tilde{\tilde{J}} (E_{m-1}) \leq \sup \bigcup_{E_m \in \mathcal{C} (E_{m-1})} \tilde{\tilde{J}} (E_m) .
			$$
			Similarly, if $E_{m-1}$ does not attain its minimum at a critical point, then
			$$
			\inf \tilde{\tilde{J}}\left(E_{m-1}\right) \geqslant \inf \bigcup_{E_m \in \mathcal{C}\left(E_{m-1}\right)} \tilde{\tilde{J}} (E_m) .
			$$
			If $E_m$ does not attain its supremum $\sup J (E_m)$ at a critical point and $\sup J (E_m) \neq$ $\sup _{E_m \in \mathcal{C}\left (E_{m-1}\right)} \sup J (E_m)$, then there is some $\tilde{E}_m \in \mathcal{C} (E_{m-1})$ such that 
			$$
			B_{\frac{11}{4}\delta_{m-1}}\left(\sup J (E_m)\right) \subset J (\tilde{E}_m).
			$$
			The same statement holds with inf in place of sup.
		\end{hp}
		\begin{hp}\label{h6}
			It will be convenient to fix language describing ``nonresonant'' and ``regular'' sets at each scale. Fix $\theta^*,E^*$. Let  $E_m\in \mathcal{C}_m$ be the Rellich function such that $E^*\in \tilde{J} (E_m)$ (if exists). The set of $m$-resonant points is defined as: \\ If  $E_{m}$ belongs  \textbf{Type} \ref{t1} or \ref{t2}, then  
			$$	S_{m} (\theta^*,E^*):=	\{x\in \Z^d:\ \theta^*+x\cdot\omega \in I (E_{m}) , \ |E_{m} (\theta^*+x\cdot \omega)-E^*|<\delta_{m}\},$$
			and if  $E_{m}$ belongs to  \textbf{Type} \ref{t3}, then 
			$$S_{m} (\theta^*,E^*):=\{x\in \Z^d:\ \theta^*+x\cdot\omega \in  I (E_{m,<}) \cup I (E_{m,>}) , \ \min_{\bullet\in\{>,<\}} |E_{m,\bullet} (\theta^*+x\cdot \omega)-E^*|<\delta_{m}\}.$$
			%  \begin{align*}
				% 	S_m (\theta^*,E^*):=	\left\{\begin{aligned}
					% 		&\{x\in \Z^d:\ \theta^*+x\cdot\omega \in I (E_m) , \ |E_m (\theta^*+x\cdot \omega)-E^*|<\delta_m\}\  \text{ \textbf{Type} \ref{t1}, \ref{t2},}\\
					% 		&\{x\in \Z^d:\ \theta^*+x\cdot\omega \in  I (E_{m,<}) \cup I (E_{m,>}) , \ \min|E_{m,\bullet} (\theta^*+x\cdot \omega)-E^*|<\delta_m\}\  \text{  \textbf{Type} \ref{t3}, }\\
					% 	\end{aligned}\right. 
				% \end{align*}
			% where $ \bullet\in \{>,<\}.$
			We say that a set $\Lambda\subset \Z^d$ is $m$-nonresonant  (relative to $ (\theta^*,E^*)$ and $E_m$) if $\Lambda\cap S_m (\theta^*,E^*) =\emptyset$ and is   $m$-regular  if  $ (\Lambda_{2l_{i+1}}+x)\subset\Lambda$ for any $x\in S_{i} (\theta^*,E^*)\cap\Lambda$   relative to  each ancestor $E_i \ (0\leq i\leq m-1)$ of $E_m$. Moreover,  we say that a set is   $m$-strongly regular if  $(\Lambda_{2l_{i+1}}+x)\subset\Lambda$ for any $x\in S_{i} (\theta^*,E^*)$ with  $ (\Lambda_{2l_{i+1}}+x)\cap\Lambda\neq \emptyset$   relative to  each ancestor $E_i \ (0\leq i\leq m-1)$ of $E_m$.
			\begin{rem}
				We refer to Appendix \ref{chouti} that  we can always make an at most the present scale length-deformation   on our blocks  so that they are regular when required. Thus we can  assume that all blocks are chosen to be regular. In fact, for a certain set $\Lambda$, the regularity condition implies  that the whole blocks  covering  the previous-scale resonant sites in $\Lambda$  are all contained in $\Lambda$, which ensures the  resolvent identity iteration proceeds, and the nonresonant condition provides the resolvent bound for the blocks contained in $\Lambda$.
			\end{rem}    Finally, we say that a set is $m$-good if it is both $m$-nonresonant and $m$-regular. With the above definition, we have Green's function estimates for $m$-good sets:\\
			Fix $\theta^*, E^*$ and a finite set $\Lambda\subset \Z^d$. Let  $E_m\in \mathcal{C}_m$ be such that $E^*\in \tilde{J} (E_m)$ (if exists).    If $\Lambda$ is $m$-good, then  for $|\theta-\theta^*|<\delta_m/ (10D), |E-E^*|<\delta_m/5$, 
			\begin{align*}
				\|G_\Lambda (\theta,E)\|&\leq10\delta_m^{-1},\\
				|G_\Lambda (\theta,E;x,y)|&\leq e^{-\gamma_m\|x-y\|_1}, \  \|x-y\|_1\geq l_m^{\frac{5}{6}},
			\end{align*}
			where $\gamma_m= (1-O (l_m^{-\frac{1}{30}}))\gamma_{m-1}\geq \frac{1}{2}\gamma_0$.
			The above estimates also hold if $\Lambda$ is $m$-regular and $E^*\notin  \tilde{J} (E_m)$ for any  $E_m\in \mathcal{C}_m$.
		\end{hp}
	\end{prop} 
	
	\subsection{Verification of induction hypotheses for $n=1$}
	In Section \ref{n=1}, we have verified Proposition \ref{induc} for $n=1$:
	
	Hypothesis \ref{h1} by the construction.
	
	Hypothesis \ref{h2} by Proposition \ref{adj}.  
	
	Hypothesis \ref{h3} by  \ref{31}  ($j=1$) and Proposition \ref{1215}  ($j=2$).
	
	Hypothesis \ref{h4} by case analysis: \\ If  $E_1^{(1)}\in \mathcal{C}_1^{(1)}$, then it  belongs to \textbf{Type} \ref{t1} or \ref{t2} (c.f.   Proposition \ref{k1}, \ref{str1}). If the two brothers $E_{1,>}^{(2)}, E_{1,<}^{(2)}\in \mathcal{C}_1^{(2)}$ satisfy  $$ \inf   E_{1,>}^{(2)} -\sup E_{1,<} ^{(2)}\leq 3\delta_1^{ (2)}, $$ then they belong to \textbf{Type} \ref{t3}. Otherwise, each of them belongs to \textbf{Type} \ref{t2} (c.f.  Lemma \ref{inter},  Propositions \ref{k2}, \ref{str2n}).
	
	Hypothesis \ref{h5} holds true  because  the overlap of the adjacent codomains ($\sim 3\delta_0$)  is  much larger  than the codomain contraction in our adjustment (c.f. Propositions \ref{coverpr}, \ref{adj}).
	
	Finally, Hypothesis \ref{h6} by Theorem \ref{1g}. 
	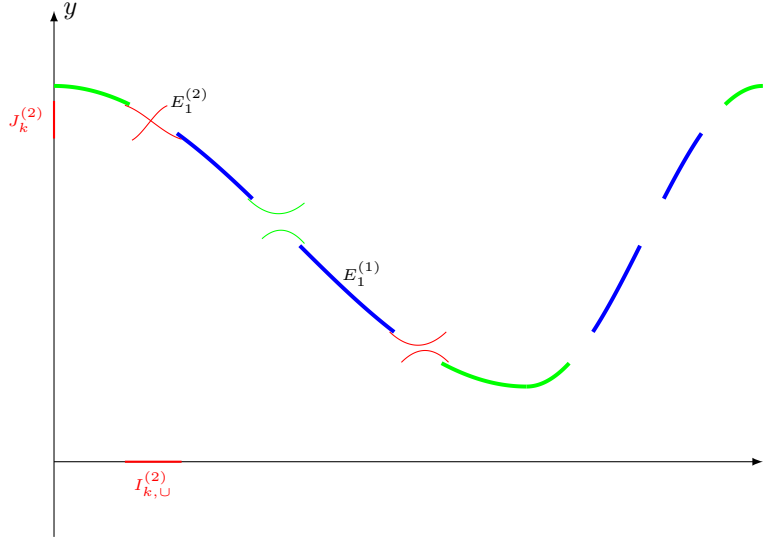
\begin{figure}[htp]
		\begin{tikzpicture}[>=latex, scale=1]
			\draw[->] (0,0)-- (3*pi,0);
			\draw[->]  (0,-1)-- (0,6)node[right]{$y$};
			\draw[thick,red,] (0,4.3)-- (0,4.8); 	\draw (0,4.55)node[left]{\color{red}{\tiny $J_k^{(2)}$}}; 
			\draw[thick,red,] (0.3*pi,0)-- (0.54*pi,0); 	\draw (0.42*pi,0)node[below]{\color{red}{\tiny $I_{k,\cup}^{(2)}$}}; 
			\draw  (0,0); \draw (0.45*pi,4.8)node[right]{{\tiny $E_{1}^{(2)}$}}; 
			\draw  (0,0); \draw (1.18*pi,2.5)node[right]{{\tiny $E_{1}^{(1)}$}}; 
			\draw  (0,0);
			\draw[domain=0:0.32*pi, samples=100,line width=1.5pt,green] plot  ({\x},{2*cos (\x/2 r) +3});
			\draw[domain=2.84*pi:3*pi, samples=100,line width=1.5pt,green] plot  ({\x},{3-2*cos ((\x-2*pi) r) });
			\draw[domain=0.52*pi:0.84*pi, samples=100,line width=1.5pt,blue] plot  ({\x},{2*cos (\x/2 r) +3});
			\draw[domain=2.74*pi:2.58*pi, samples=100,line width=1.5pt,blue] plot  ({\x},{3-2*cos ((\x-2*pi) r) });
			\draw[domain=1.04*pi:1.44*pi, samples=100,line width=1.5pt,blue] plot  ({\x},{2*cos (\x/2 r) +3});
			\draw[domain=2.48*pi:2.28*pi, samples=100,line width=1.5pt,blue] plot  ({\x},{3-2*cos ((\x-2*pi) r) });
			\draw[domain=1.64*pi:2*pi, samples=100,line width=1.5pt,green] plot  ({\x},{2*cos (\x/2 r) +3});
			\draw[domain=2*pi:2.18*pi, samples=100,line width=1.5pt,green] plot  ({\x},{3-2*cos ((\x-2*pi) r) });
			\draw[domain=0.82*pi:1.06*pi, samples=100,green] plot  ({\x},{1.2* (\x-0.95*pi)* (\x-0.95*pi)+3.3});
			\draw[domain=0.88*pi:1.06*pi, samples=100,green] plot  ({\x},{-1.8* (\x-0.96*pi)* (\x-0.96*pi)+3.08});
			\draw[domain=1.42*pi:1.66*pi, samples=100,red] plot  ({\x},{1.25* (\x-1.54*pi)* (\x-1.54*pi)+1.55});
			\draw[domain=1.47*pi:1.67*pi, samples=100,red] plot  ({\x},{-1.6* (\x-1.57*pi)* (\x-1.57*pi)+1.48});	
			\draw[domain=0.33*pi:0.48*pi, samples=100,red] plot  ({\x},{0.25*sin ( (5*\x-2*pi) r) +4.5});	
			\draw[domain=0.3*pi:0.54*pi, samples=100,red] plot  ({\x},{-0.25*sin ( (3*\x-1.25*pi) r) +4.52});	
		\end{tikzpicture}
		\caption{A cartoon output of the first induction step: A collection $\mathcal{C}_1$ of local Rellich functions of  $H_{B_1^{(j)}}$. The thick curves  come from simple resonances, and the thin curves come from double resonances.
			The adjacent domains of the curves have overlap $\gtrsim\delta_0$. The curves in blue (resp. green and red) belong  to \textbf{Type} \textbf{1} (resp. 
			\textbf{2} and \textbf{3}).  	 Note that different curves need not agree on the overlap of their domains.}
	\end{figure}
	
	\subsection{Construction of the next length scale}
	In this section, we suppose that  Proposition \ref{induc} holds for $1\leq m\leq n$, and prove the Proposition for $m=n+1$. The inductive argument proceeds on each constructed function $E_n\in \mathcal{C}_n$. For this purpose, we fix a Rellich function $E_n\in \mathcal{C}_n$ and construct the collection of its children  $\mathcal{C} (E_n)$. By our hypotheses, $E_n$ has a two-monotonicity interval structure and satisfies the  Morse condition. To classify the type of resonance, we need the following lemma: 
	\begin{lem}\label{sepn}
		There exists at most one $0\neq \|{k_n}\|_1\leq 10l_{n+1}^{(1)}$, such that there is some $\theta_{{k_n},-}\in  I (E_n)_-$ satisfying 
		$$\theta_{{k_n},-}+{k_n}\cdot \omega  \in  I (E_n)_+, \ e_{k_n}:= E_n (\theta_{{k_n},-})=E_n (\theta_{{k_n},-}+{k_n}\cdot \omega). $$
		Moreover, if such a  ${k_n}$ exists, then $I (E_n)_+$ and $I (E_n)_-$ are two disjoint intervals and $E_n$ belongs to \textbf{Type} \ref{t1}.
	\end{lem}
	\begin{proof}
		The proof is analogous to the  proof of  Lemma \ref{sep}. Since  $ |I (E_n)_\pm|\lesssim \delta_{n-2}^5$   (c.f. Hypothesis \ref{h1}) and $E_n$ satisfies the Morse condition  (c.f. Hypothesis \ref{h3}), we can employ  Lemma \ref{C2}  to show  if there is another ${k'_n}$ and $\theta_{{k'_n},-}$, then 
		$$ (l_{n+1}^{(1)})^{-2\tau}\lesssim |e_{k_n}-e_{{k'_n}}|\leq |J (E_n)|\lesssim\delta_{n-2}^{10},$$ a contradiction.

		If $I (E_n)$ is a single interval, then 
		$$  (l_{n+1}^{(1)})^{-2\tau}\lesssim\| (\theta_{{k_n},-}+{k_n}\cdot \omega )-\theta_{{k_n},-}\|\leq|I (E_n)| \lesssim \delta_{n-2}^5$$ since $\theta_{{k_n},-}$ and $ \theta_{{k_n},-}+{k_n}\cdot \omega \in I (E_n)$, a contradiction.
	\end{proof}
	Define the double resonant interval  (if exists)
	$$J^{(2)}_{k_n}:=\bar{B}_{\delta_{n-1}^{10}} (e_{k_n})$$
	and  
	$$	J^{SR}_n:=\tilde{J} (E_n)\setminus \bar{B}_{\delta_{n-1}^{10}-3\delta_n} (e_{k_n}).$$
	We now divide it into simple resonant intervals of comparable size to the double-resonant
	interval:
	\begin{prop}\label{covern}
		$\tilde{\tilde{J}} (E_n)\cap J (E_n)$ can be covered by closed intervals $J_i^{(1)}$ and  $J^{(2)}_{k_n}$ (if needed), such that 
		\begin{itemize}
			\item[\textbf{(1)}.]  $\delta_{n-1}^{10}\leq |J_i^{(1)}|\leq  2\delta_{n-1}^{10}$. 
			\item[\textbf{(2)}.] $J_i^{(1)}\subset J^{SR}_n.$
			\item[\textbf{(3)}.] $J_i^{(1)}\cap J_{i'}^{(1)}\neq\emptyset\Rightarrow |J_i^{(1)}\cap J_{i'}^{(1)}|=3\delta_n.$
			\item[\textbf{(4)}.] $ J_i^{(1)}\cap J^{(2)}_{k_n}\neq\emptyset\Rightarrow |J_i^{(1)}\cap J^{ (2)}_{k_n}|=3\delta_n. $
			\item[\textbf{(5)}.] The total number of such intervals does not exceed $2\delta_{n-1}^{-10}$.
			\item[\textbf{(6)}.]For any $E\in J_i^{(1)}$, $\bar{B}_{2\delta_n} (E)\subset \tilde{J} (E_n)$.
			\item[\textbf{(7)}.]  $\bar{B}_{5D\delta_{n-1}^{5}} (e_{k_n})\subset\tilde{J} (E_n)$.
			\item[\textbf{(8)}.] Moreover, if $E_n$ belongs to \textbf{Type} \ref{t3}, the closed intervals can be constructed such that  \begin{equation}\label{chongdie}
				[\sup J (E_{n,<})-\frac{1}{2}\delta_{n-1}^{10},\inf J (  E_{n,>})+\frac{1}{2}\delta_{n-1}^{10}]\subset J_{i_n}^{(1)}
			\end{equation} for some $i_n$, and  $$ [\sup J (E_{n,<})-\frac{1}{2}\delta_{n-1}^{10},\inf J (  E_{n,>})+\frac{1}{2}\delta_{n-1}^{10}]\cap J_{i}^{(1)}=\emptyset$$ for any $i\neq i_n$.
		\end{itemize}
	\end{prop}
	\begin{proof}
		Since $|\tilde{J} (E_n)|\sim \delta_{n-2}^{10}\gg \delta_{n-1}^{10}$, we can construct a collection of closed intervals of size between $\delta_{n-1}^{10}$ and $2\delta_{n-1}^{10}$ overlapping with only their nearest neighbors by exactly $3\delta_n$, together with the interval  (if $\bar{B}_{5D\delta_{n-1}^{5}} (e_{k_n})\subset\tilde{J} (E_n)$)
		$J^{(2)}_{k_n}$  to cover  $\tilde{J} (E_n)\cap J (E_n)$, with the possible exception of intervals of length at most $2\delta_{n-1}^{10}$ at the boundaries  of  $\tilde{J} (E_n)\cap J (E_n)$ where $E_n$ does not attain a critical point. Since $ 2\delta_{n-1}^{10}<\frac{1}{8}\delta_{n-1}$, they thus cover $\tilde{\tilde{J}} (E_n)\cap J (E_n)$. If $E_n$ belongs to \textbf{Type} \ref{t3}, by the assumption, the gap between $\inf J (  E_{m,>}) $ and $\sup J (E_{m,<})$ does not exceed $3\delta_n$. We can define 
		\begin{equation*}
			J_{i_n}^{(1)}=   [\sup J (E_{n,<})-\frac{1}{2}\delta_{n-1}^{10}-3\delta_n,\inf J (  E_{n,>})+\frac{1}{2}\delta_{n-1}^{10}+3\delta_n] 
		\end{equation*}with  $\delta_{n-1}^{10}\leq |J_{i_n}^{(1)}|\leq  2\delta_{n-1}^{10}$ and construct the other intervals.
	\end{proof}
	The  preimages of $J_i^{(j)}$  are 
	$$I^{(1)}_{i,\pm}:= (E_n)_{\pm}^{-1} (J^{(1)}_i), \ I^{(1)}_{i}:=I^{(1)}_{i,-}\cup I^{ (1)}_{i,+}  $$  and $$I^{ (2)}_{k_n,\pm}:= (E_n)_{\pm}^{-1} (J^{(2)}_{k_n}), \ I^{(2)}_{k_n}:=I^{(2)}_{k_n,-}\cup I^{ (2)}_{k_n,+}.$$
	Since $E_n$ satisfies the  Morse condition, by Lemma \ref{C2}, we have 
	$$2|I^{(j)}_{\bullet,\pm}|^2\leq |E_n (I^{(j)}_{\bullet,\pm})|=|J^{(j)}_{\bullet}|\leq D|I^{(j)}_{\bullet, \pm}|.$$ 
	Thus \begin{equation}\label{domainesti}
		\delta_{n-1}^{10}\lesssim |I^{(j)}_{\bullet, \pm}|\leq2\delta_{n-1}^{5},
	\end{equation}
	where $\bullet\in \{i,k_n\}$.
	\subsubsection{The simple resonance case} In this section, we fix an interval  $J_i^{(1)}$ and its preimage $I_i^{(1)}$ to  construct the   Rellich child(ren) locally defined on $I_i^{(1)}$ of   $E_n$. 
	\begin{prop}\label{SRnpro}
		If $\theta\in I^{(1)}_i$, then for any $0<\|x\|_1\leq 10 l_{n+1}^{(1)}$ such that $\theta+ x\cdot \omega \in I (E_n)$, we have $$|E_n (\theta)-E_n (\theta+x\cdot \omega)|\geq 3\delta_n.$$
	\end{prop}
	\begin{proof}
		Suppose $\theta\in I_{i,-}^{(1)}$ (the case $\theta\in I_{i,+}^{(1)}$ is completely analogous). If $\theta+x\cdot \omega\in I (E_n)_-$, since   $0<\|x\|_1\leq  10l_{n+1}^{(1)}$, by the Morse condition of $E_n$, we   get $$|E_n (\theta)-E_n (\theta+x\cdot \omega)|\geq \|x\cdot \omega\|^2\gtrsim  (l_{n+1}^{(1)})^{-2\tau}\gg 3\delta_n.$$
		We thus suppose $\theta+x\cdot \omega\in I (E_n)_+$. Since $E_n$ satisfies  the two-monotonicity interval  structure and  has the same image on each monotonicity component, it must be $x=k_n$ and there exists  some $\theta_{k_n,-}\in  I (E_n)_-$ satisfying 
		$$\theta_{k_n,-}+k_n\cdot \omega  \in  I (E_n)_+, \ e_{k_n}:= E_n (\theta_{k_n,-})=E_n (\theta_{k_n,-}+k_n\cdot \omega). $$  Assume $\theta\leq \theta_{k_n,-}$. Then  $\theta +k_n \cdot \omega \leq \theta_{k_n,+}$. By the monotonicity of $E_n$ on $I (E_n)_{\pm}$, $E_n (\theta)\geq e_{k_n} $ and $E_n (\theta+k_n\cdot \omega )\leq e_{k_n} $. Since $E_n (I_i^{(1)})=J_i^{(1)}\subset J^{SR}_n$,
		$$E_n (\theta)-E_n (\theta+k_n\cdot \omega)\geq E_n (\theta)-e_{k_n}\geq \delta_{n-1}^{10}-3\delta_n  >3\delta_n.$$
		If $\theta\geq \theta_{k_n,-}$, then similarly $E_n (\theta)\leq e_{k_n}, E_n (\theta+k_n\cdot \omega )\geq e_{k_n}$, and 	$$|E_n (\theta)-E_n (\theta+k_n\cdot \omega)|=E_n (\theta+k_n\cdot \omega)-E_n (\theta)\geq e_{k_n} -E_n (\theta)  >3\delta_n.$$
	\end{proof}
	We choose a block  $B_{n+1}^{(1)}$ satisfying 
	$$\Lambda_{l_{n+1}^{(1)}}\subset B_{n+1}^{(1)} \subset \Lambda_{l_{n+1}^{(1)}+50l_n}$$
	such that for all $\theta\in I_i^{(1)}$,  $B_{n+1}^{(1)}$ is $n$-strongly regular relative to $ (\theta,E)$ for all $E\in B_{2\delta_n} (E_n (\theta))$. 
	
	The simple resonance case will be divided into two subcases.
	% depending on whether $J_i^{(1)}$ is the interval satisfying \eqref{chongdie}
	\\
	
	\underline{\textbf{Subcase A.}}  $E_n$ does not belong  to \textbf{Type} \ref{t3} or  $J_i^{(1)}$ does not satisfy \eqref{chongdie}. \\
	
	In this subcase, by \textbf{Hypothesis} \ref{h4}, $E_n (\theta)$ is the unique eigenvalue of $H_{B_n} (\theta)$  in the interval $[E_n (\theta)-3\delta_n,E_n (\theta)+3\delta_n]$ for any $\theta\in I_i^{(1)}$. Thus we have \begin{equation}\label{Suba}
		\|G_n^\perp (E)\|\lesssim \delta_n^{-1}
	\end{equation}
	for any $|E-E_n (\theta)|\leq 2\delta_n$, where $G_n^\perp$ is the Green's function for $B_n$ on the orthogonal complement of the eigenspace corresponding to $E_n (\theta)$. We will show that  $E_n$ has a  unique Rellich child $E_{n+1}^{(1)}$ of the Dirichlet restriction $H_{B_{n+1}^{(1)}}$, which is $C^2$-closed  to $E_n$ and belongs to \textbf{Type} \ref{t1} or \ref{t2}.
	\begin{prop}\label{kn}
		For $\theta\in I^{(1)}_i$, we have 
		\begin{itemize}
			\item[\textbf{(a).}]  $H_{B_{n+1}^{(1)}} (\theta)$ has a unique eigenvalue $E_{n+1}^{(1)} (\theta)$ such that $|E_{n+1}^{(1)} (\theta)-E_n (\theta)|\lesssim e^{-\frac{\gamma_0}{4}l_n}\ll\delta_n^{10}$. Moreover, any other $\hat{E}\in\sigma (H_{B_{n+1}^{(1)}} (\theta)) $ must obey $|\hat{E}-E_n (\theta)|>2\delta_{n}$.
			\item[\textbf{(b).}]  The corresponding eigenfunction $\psi_{n+1}$   to $E_{n+1}^{(1)}$   satisfies  
			\begin{equation}\label{dec}
				|\psi_{n+1} (x)|\leq e^{-\frac{\gamma_0}{4}\|x\|_1}, \  
				\|x\|_1\geq l_n^{\frac{6}{7}}. 
			\end{equation}
			\item[\textbf{(c).}] Let $\psi_{n}$ be  the eigenfunction corresponding to  $E_n (\theta)$ of  $H_{B_{n}} (\theta)$. Then up to a choice of sign, 
			\begin{equation}\label{closed}
				\|\psi_{n+1}-\psi_{n}\|\lesssim e^{-\frac{\gamma_0}{4}l_n}\delta_n^{-1}.
			\end{equation}
			\item[\textbf{(d).}] $\|G_{n+1}^\perp (E_{n+1}^{(1)})\|\leq\delta_n^{-1}$, where $G_{n+1}^\perp$ denotes the Green's function for $B_{n+1}^{(1)}$ on the orthogonal complement of $\psi_{n+1}$.
		\end{itemize} 
	\end{prop}
	\begin{proof}
		By  the decay  \eqref{Hdecay} of eigenfunction in the  inductive hypothesis, we have  $$\| (H_{B_{n+1}^{(1)}} (\theta)-E_n (\theta))\psi_n\|=\|\Gamma_{n}\psi_n\|\lesssim e^{-\frac{\gamma_0}{4}l_n}.$$ The above estimate together with  Lemma \ref{trialcor}   shows the existence of $E_{n+1}^{(1)} (\theta)$.
		
		Let $E\in B_{2\delta_n} (E_n (\theta))\subset\tilde{J} (E_n)$ and $\hat{B}_n$  be  a $O (l_n^{2/3})$-size block  centered at origin   so that $\Lambda:=B_{n+1}^{(1)}\setminus\hat{B}_n$ is $n$-regular relative to $ (\theta,E)$. By Proposition \ref{SRnpro}, 
		we have  $\Lambda\cap 
		S_{n} (\theta,E)=\emptyset$.	Thus $\Lambda$ is $n$-good relative to $ (\theta,E)$. By the \textbf{Hypothesis} \ref{h6}, we obtain   $$ \ |G_\Lambda (\theta,E;x,y)|\leq e^{-\gamma_n\|x-y\|_1}, \  \|x-y\|_1\geq l_n^{\frac{5}{6}}.$$  
		Let $E\in \sigma (H_{B_{n+1}^{(1)}} (\theta))$ such that $|E-E_n (\theta)|\leq2 \delta_n$.  We determine the value of $\psi_1 (x)$ for $\|x\|\geq l_n^{\frac{6}{7}}$  by 
		$$\psi_{n+1} (x)= \sum_{z,z'} G_{\Lambda} (E,\theta; x, z) \Gamma_{z,z'} \psi_{n+1} (z').$$
		Since $\operatorname{dist} (x,\partial\hat{B}_n)\geq \|x\|_1-O (l_n^{2/3})\geq\frac{2}{3} \|x\|_1>l_n^{5/6}$, by the off-diagonal exponential decay of  $G_\Lambda$,  we get 
		\begin{align*}
			|\psi_{n+1} (x)|&\leq C (d)\sum_{z'\in \partial^+ \hat{B}_n}e^{-\frac{1}{3}\gamma_0 \|x\|_1}   |\psi_{n+1} (z')|\\
			&\leq e^{-\frac{1}{4}\gamma_0 \|x\|_1}.
		\end{align*} Thus we finish the proof of  \textbf{(b)}.

		To establish \textbf{(c)}, we must show that $\psi_{n+1}$ is close to  $\psi_{n}$ inside $B_n$. To see this, we restrict the equation  $H_{B_{n+1}^{(1)}} (\theta)\psi_{n+1}=E_{n+1}^{(1)} (\theta)\psi_{n+1}$ on $B_n$ to get 
		$$\left(H_{B_{n}}-E_{n+1}^{(1)}\right)\psi_{n+1}=\Gamma_{B_{n}}\psi_{n+1}.$$
		By \eqref{Suba},  \eqref{dec} and the above equation, we get  $$\|P_n^\perp\psi_{n+1}\|=\|G_n^\perp (E_{n+1}^{(1)})P_n^\perp\Gamma_{B_{n}}\psi_{n+1}\|\lesssim e^{-\frac{1}{4}\gamma_0l_n}\delta_n^{-1},$$
		where $P_n^\perp$ is the projection on  the orthogonal complement of $\psi_n$. Since $\psi_{n+1}$ is normalized, we obtain $\|\psi_{n+1}-\psi_{n}\|\lesssim e^{-\frac{1}{4}\gamma_0l_n}\delta_n^{-1}$ up to a choice of sign.
		If there is another $\hat{E}\in \sigma (H_{B_{n+1}^{(1)}} (\theta))$ satisfying $|\hat{E}-E_n (\theta)|\leq2\delta_n$, as above we can  show that the corresponding eigenfunction $\hat{\psi}$  must  also satisfy \eqref{closed} replacing $\psi_{n+1}$  by $\hat{\psi}$,  which violates the  orthogonality. Thus we prove the uniqueness part of \textbf{(a)}.
		
		Finally, \textbf{(d)} follows from any other $\hat{E}\in \sigma (H_{B_{n+1}^{(1)}} (\theta))$ must obey $$|\hat{E}-E_{n+1}^{(1)} (\theta)|\geq |\hat{E}-E_{n} (\theta)|-|E_{n} (\theta)-E_{n+1}^{(1)} (\theta)|\geq2\delta_n-\delta_n\geq\delta_n.$$
	\end{proof}
	%Next proposition shows that $E_{n+1}^{(1)} (\theta)$ is a perturbation of $E_0 (\theta)$ for $|\theta-\theta^\bullet|<\delta_0/ (10M_1)$.
	\begin{prop}\label{apn} We have the following:  
		$$|\frac{d^s}{d\theta^s} (E_{n+1}^{(1)} (\theta)-E_n (\theta))|\lesssim  e^{-\frac{1}{16}\gamma_0l_n}\delta_n^{-s} \ll\delta_n^{10} \ {\rm for\ }s=0,1,2.$$
	\end{prop}
	\begin{proof}
		By the previous proposition, we have 
		\begin{equation}\label{C0n}
			|E_{n+1}^{(1)} (\theta)-E_n (\theta)|\lesssim e^{-\frac{1}{4}\gamma_0l_n}.
		\end{equation}
		Since \eqref{closed}, by Feynman-Hellman formula,  we obtain 
		$$|  (E_{n+1}^{(1)})' (\theta)- E_n' (\theta)|=|\left\langle\psi_{n+1},V' \psi_{n+1}\right\rangle-\left\langle\psi_n,V' \psi_n\right\rangle|\lesssim e^{-\frac{\gamma_0}{4}l_n}\delta_n^{-1}.$$
		For $s=2$, by Feynman-Hellman formula, we have 
		$$ E_{n}'' (\theta)=\left\langle\psi_{n}, V'' \psi_{n}\right\rangle-2\left\langle\psi_{n}, V' G_{n}^{\perp} (E_{n}) V' \psi_{n}\right\rangle,$$
		$$  (E_{n+1}^{(1)})'' (\theta)=\left\langle\psi_{n+1}, V'' \psi_{n+1}\right\rangle-2\left\langle\psi_{n+1}, V' G_{n+1}^{\perp} (E_{n+1}^{(1)}) V' \psi_{n+1}\right\rangle. $$
		Thus it suffices to estimate 
		\begin{align*}
			&|\left\langle\psi_{n+1}, V' G_{n+1}^{\perp} (E_{n+1}^{(1)}) V' \psi_{n+1}\right\rangle-\left\langle\psi_{n}, V' G_{n}^{\perp} (E_{n}) V' \psi_{n}\right\rangle|\\
			\leq&|\left\langle\psi_{n+1}, V' G_{n+1}^{\perp} (E_{n+1}^{(1)}) V' \psi_{n+1}\right\rangle-\left\langle\psi_{n}, V' G_{n+1}^{\perp} (E_{n+1}^{(1)}) V' \psi_{n}\right\rangle|\\+&|\left\langle\psi_{n}, V' G_{n+1}^{\perp} (E_{n+1}^{(1)}) V' \psi_{n}\right\rangle-\left\langle\psi_{n}, V' G_{n}^{\perp} (E_{n}) V' \psi_{n}\right\rangle|\\
			\leq&O (e^{-\frac{\gamma_0}{4}l_n}\delta_n^{-2})+|\left\langle V'\psi_{n},  \left (G_{n+1}^{\perp} (E_{n+1}^{(1)})-G_{n}^{\perp} (E_{n})\right) V' \psi_{n}\right\rangle|,
		\end{align*}
		where we have used \eqref{closed} and $\|G_{n+1}^{\perp} (E_{n+1}^{(1)})\|\leq \delta_n^{-1}$ to bound the term on the second line.

		It remains to  estimate the operator $G_{n+1}^{\perp} (E_{n+1}^{(1)})-G_{n}^{\perp} (E_{n})$ restricted to $B_n$.  By  resolvent identity, we have 
		\begin{equation}\label{545}
			\begin{aligned}
				& \ \  \ \ G_{n+1}^{\perp} (E_{n+1}^{(1)})-G_{n}^{\perp} (E_{n})\\&
				=G_{n+1}^{\perp} (E_{n+1}^{(1)})P_n^\perp-P_{n+1}^\perp G_{n}^{\perp} (E_{n})+G_{n+1}^{\perp} (E_{n+1}^{(1)})P_n-P_{n+1}G_{n}^{\perp} (E_{n})\\
				&=G_{n+1}^{\perp} (E_{n+1}^{(1)})\left (-\Gamma_n+ (E_{n+1}^{(1)}-E_n)\right)G_{n}^{\perp} (E_{n})\\ &+G_{n+1}^{\perp} (E_{n+1}^{(1)})P_{n+1}^\perp P_n-P_{n+1}P_n^\perp G_{n}^{\perp} (E_{n}).
			\end{aligned}
		\end{equation}
		We have used the orthogonal projection $P_n$ and $P_{n+1}$ onto $\psi_{n}$ and $\psi_{n+1}$ respectively and the relation $P_n+P_n^\perp=\operatorname{Id}_{B_n}$, $P_{n+1}+P_{n+1}^\perp=\operatorname{Id}_{B_{n+1}^{(1)}}$. The last two terms on the right hand side  of \eqref{545} are bounded by $O (e^{-\frac{\gamma_0}{4}l_n}\delta_n^{-2})$  since $$\|P_{n+1}^\perp P_n\|=\|P_{n+1}^\perp\psi_{n}\|\leq2\|\psi_{n+1}-\psi_{n}\|=O (e^{-\frac{\gamma_0}{4}l_n}\delta_n^{-1}),$$
		$$\|P_{n+1}P_n^\perp\|=\|P_n^\perp P_{n+1}\|=\|P_{n}^\perp\psi_{n+1}\|\leq2\|\psi_{n+1}-\psi_{n}\|=O (e^{-\frac{\gamma_0}{4}l_n}\delta_n^{-1})$$
		and $$\|G_{n}^{\perp} (E_{n})\|  , \ \|G_{n+1}^{\perp} (E_{n+1}^{(1)})\|=O (\delta_n^{-1}).$$
		The second term on the right hand side of \eqref{545} is bounded by $O (e^{-\frac{\gamma_0}{4}l_n}\delta_n^{-2})$ since $|E_{n+1}^{(1)}-E_n|=O (e^{-\frac{\gamma_0}{4}l_n})$. Therefore the case $s=2$ follows if we prove 
		\begin{equation}\label{928}
			\|\Gamma_n G_n^\perp (E_n)V'\psi_{n}\|=O (e^{-\frac{\gamma_0}{4}l_n}\delta_n^{-1}).
		\end{equation}
		Let $\chi_n$ be the restriction operator  of the block $\Lambda_{\frac{l_n}{4}}\subset B_n$.  By the estimate \eqref{dec} we have 
		$$\| (1-\chi_n)V'\psi_{n}\|=O (e^{-\frac{\gamma_0}{16}l_n}).$$
		Thus in order to prove \eqref{928}, it suffices to show $\|\Gamma_nG_n^\perp \chi_n\|=O (e^{-\frac{\gamma_0}{16}l_n}\delta_n^{-1})$. To do this, we choose a $O (l_n^{2/3})$-size block $\hat{B}$ centered at origin so that $A:=B_n\setminus \hat{B}$ is $(n-1)$-good relative to $ (\theta,E_n (\theta))$. By resolvent identity, we get 
		\begin{align*}
			\Gamma_nG_n^\perp\chi_n&=\Gamma_nG_A\chi_n+\Gamma_n (\chi_AG_n^\perp-G_AP_n^\perp)\chi_n-\Gamma_nG_AP_n\chi_n\\
			&=\Gamma_nG_A\chi_n+\Gamma_nG_A\Gamma_AG_n^\perp\chi_n-\Gamma_nG_AP_n\chi_n.
		\end{align*}
		Since $A$ is $ (n-1)$-good, we have  $\|G_A (E_n)\|\leq10\delta_{n-1}^{-1}$ and   $G_A (\theta,E_n;x,y)$ decays exponentially  for $\|x-y\|_1\geq l_{n-1}^\frac{5}{6}$. Thus $\|\Gamma_nG_A\chi_n\|,\|\Gamma_nG_A\Gamma_A\|=O (e^{-\frac{\gamma_0}{16}l_n})$. To estimate the final term, we use $\|P_n\chi_n-\chi_nP_n\|\leq \|P_n\chi_n-P_n\|+\|P_n-\chi_nP_n\|\leq2\| (1-\chi_n)P_n\|= 2\| (1-\chi_n)\psi_{n}\|=O (e^{-\frac{\gamma_0}{16}l_n})$ to obtain 
		$$\|\Gamma_nG_AP_n\chi_n\|\leq\|\Gamma_nG_A\chi_nP_n\|+\|\Gamma_nG_A\chi_n (P_n\chi_n-\chi_nP_n)\|=O (e^{-\frac{\gamma_0}{16}l_n}\delta_n^{-1}).$$
	\end{proof}
	The next proposition verifies the Morse condition of $E_{n+1}^{(1)} (\theta)$.
	\begin{prop}\label{31n}
		Assume $\theta\in I_{i}^{(1)}$ such that $| (E_{n+1}^{(1)})' (\theta)|\leq\delta_{n-1}^2$. Then $$| (E_{n+1}^{(1)})'' (\theta)|\geq3-\sum_{s=0}^{n}\delta_s^{10}>2$$ with a unique sign for all such $\theta$. 
	\end{prop}
	\begin{proof}
		If $| (E_{n+1}^{(1)})' (\theta)|\leq\delta_{n-1}^2$, by Proposition \ref{apn}, we have 
		$$|E_n' (\theta)|\leq| (E_{n+1}^{(1)})' (\theta)|+\delta_0^{10}\leq 2\delta_{n-1}<\delta_{n-2}^2.$$
		By  Hypothesis \ref{h3}, 
		$$|E_n'' (\theta)|\geq 3-\sum_{s=0}^{n-1}\delta_s^{10}$$ with a unique sign for all such $\theta$.  We employ  Proposition \ref{apn} (the version $s=1$) to the above inequality  to  get  $$| (E_{n+1}^{(1)})'' (\theta)|\geq|E_n'' (\theta)|-\delta_n^{10}\geq3-\sum_{s=0}^{n}\delta_s^{10}>2$$ with a unique sign for all such $\theta$.
	\end{proof}
	In the following proposition,  we show each Rellich child  $E_{n+1}^{(1)}$ defined on $I^{(1)}_i$ inherits  the two-monotonicity interval structure and  its derivative satisfies the estimates  in Hypothesis \ref{h4} for $m=n+1$.
	\begin{prop}\label{strn1} We have the following:
		\begin{itemize}
			\item [\textbf{(a).}] Suppose   $E_n$ has a critical point in $I^{(1)}_i$.  (Thus $I^{(1)}_i$ is a closed interval.) Then $E_{n+1}^{(1)}$ also has a critical point in $I^{(1)}_i$. Moreover, $I^{(1)}_i$ can be divided into two closed  intervals $\tilde{I}^{ (1)}_{i,+}$  and  $\tilde{I}^{ (1)}_{i,-}$ with disjoint interiors, such that 
			$$\pm  (E_{n+1}^{(1)})'|_{\tilde{I}^{(1)}_{i,\pm}}\geq0, \ \tilde{I}^{(1)}_{i,+}\cup\tilde{I}^{(1)}_{i,-}=I^{(1)}_i.$$
			In this case, $E_{n+1}^{(1)}$ belongs to \textbf{Type} \ref{t2}.
			\item [\textbf{(b).}] Suppose   $E_n$ has no  critical point in $I^{(1)}_i$  (Thus  $I^{(1)}_i=I^{(1)}_{i,+}\cup I^{ (1)}_{i,-}$ is a union of two closed disjoint intervals).  Then 
			$$\pm  (E_{n+1}^{(1)})'|_{I^{(1)}_{i,\pm}}\gtrsim \delta_{n-1}^{10}.$$
			In this case, $E_{n+1}^{(1)}$ belongs to \textbf{Type} \ref{t1}.
		\end{itemize}
	\end{prop}
	\begin{proof}
		To prove \textbf{(a)}, assume  $\theta_c$ is the  critical point of $E_n$ in $I^{(1)}_i$. By \eqref{domainesti}, we have  $\operatorname{dist} (\theta_c,\partial I^{(1)}_i)\gtrsim\delta_{n-1}^{10}$. By  Hypothesis \ref{h3}, we can employ Lemma \ref{C2} to obtain $| (E_{n})'|_{\partial I^{(1)}_i}|\gtrsim\delta_{n-1}^{10}$ with different signs. By   Proposition \ref{apn}, we get
		$$| (E_{n+1}^{(1)})'|_{\partial I^{(1)}_i}|\geq | (E_{n})'|_{\partial I^{(1)}_i}|-\delta_n^{10} \gtrsim\delta_{n-1}^{10}$$
		with different sign on the edges of $I^{(1)}_i$. Since $ (E_{n+1}^{(1)})'$ is continuous and  $I^{(1)}_i$ is connected, there is point $\theta_s\in I^{(1)}_i$ such that $ (E_{n+1}^{(1)})' (\theta_s)=0$.  The two monotonicity structure  follows from Proposition \ref{31n} and Lemma \ref{C2} immediately.

		We prove \textbf{(b)} by case analysis: 
		\begin{itemize}
			\item[ (1).]  $E_n$ belongs to \textbf{Type} \ref{t1}: In this case, by Hypothesis \ref{h4}, we have 
			$$\pm  (E_{n})'|_{I^{(1)}_{i,\pm}}\gtrsim \delta_{n-2}^{10}.$$
			Thus by Proposition \ref{apn}, 
			$$\pm  (E_{n+1}^{(1)})'|_{I^{(1)}_{i,\pm}}\geq  \pm  (E_{n})'|_{I^{(1)}_{i,\pm}}- \delta_n^{10} \gtrsim \delta_{n-1}^{10}.$$
			\item[ (2).] $E_n$ belongs to \textbf{Type} \ref{t2} or \ref{t3}: In this case, we  notice that from  the construction (c.f. Proposition \ref{covern}), if $J_i^{(1)}$ does not contain a critical value of $E_n$, then $J_i^{(1)}$ has a  distance  on order $\delta_{n-1}^{10}$ away  from the critical value. Thus its preimage $I_i^{(1)}$  has a distance on order $\delta_{n-1}^{10}$ away  from the critical point.  
			Since $E_n$ satisfies  the Morse condition, by Lemma \ref{C2}, we get $$\pm  (E_n)'|_{I^{(1)}_{i,\pm}}\gtrsim \delta_{n-1}^{10}.$$   The desired estimate for  $\pm  (E_{n+1}^{(1)})'|_{I^{(1)}_{i,\pm}}$ follows from the above estimate and Proposition \ref{apn} immediately.
		\end{itemize}   
	\end{proof}
	
	\underline{\textbf{Subcase B.}}  {The negation of  \textbf{Subcase A}:  $E_n$ belongs  to \textbf{Type} \ref{t3}  and  $J_{i_n}^{(1)}$ is the interval  satisfying \eqref{chongdie}}.\\

	Before analyzing \textbf{Subcase B}, we  are going to deal with the double resonance case, where we can extract some helpful information.
	
	\subsubsection{The double resonance case} In this case,    there exists  some $k_n$ with  $0\neq \|k_n\|_1\leq 10l_{n+1}^{(1)}$ and     $\theta_{k_n,-}\in  I^{(2)}_{k_n,-}$ such that  
	\begin{equation}\label{huiyi}
		\theta_{k_n,-}+k_n\cdot \omega  \in  I^{(2)}_{k_n,+}, \ e_{k_n}:= E_n (\theta_{k_n,-})=E_n (\theta_{k_n,-}+k_n\cdot \omega),
	\end{equation} and $E_n$ belongs to \textbf{Type} \ref{t1} by Lemma \ref{sepn}.
	Recall   that  $$J^{(2)}_{k_n}:=\bar{B}_{\delta_{n-1}^{10}} (e_{k_n})$$  and $$I^{(2)}_{k_n,\pm}:= (E_n)_{\pm}^{-1} (J^{(2)}_{k_n}), \ I^{(2)}_{k_n}:=I^{(2)}_{k_n,-}\cup I^{ (2)}_{k_n,+}. $$ The domain of the Rellich children is defined as   $$ I^{ (2)}_{k_n,\cup}:= I^{(2)}_{k_n,-}\cup  (I^{ (2)}_{k_n,+}-k_n\cdot \omega),$$ which is a single interval satisfying $|I^{(2)}_{k_n,\cup}|\leq 4\delta_{n-1}^5$ by \eqref{domainesti}.  By the construction (c.f. \textbf{(7)} of Proposition \ref{covern}) and  $|v'|\leq D$, we have    $I^{ (2)}_{k_n,\cup},I^{ (2)}_{k_n,\cup}+k_n\cdot \omega \subset I (E_n)$. In the following, we denote by $E_{n-1}$ the parent Rellich function of $E_n$. 
	We note that double resonance of $E_n$  can happen only at length greater than $10l_n$:
	\begin{lem}\label{ds}
		If such $k_n$ exists, then  	$E_n$  is   generated from    a simple resonant  interval belonging to \textbf{Subcase A} of its parent $E_{n-1}$.
	\end{lem}
	\begin{proof}
		Suppose that  $E_n$ is  generated from a double  resonant interval  or the simple resonant  interval belonging to \textbf{Subcase B} of its parent $E_{n-1}$.   By the  construction, $I (E_n)$ is a single interval of length  of   order at most $\delta_{n-2}^5\ll \|k_n\cdot \omega \|$ by Diophantine condition, thus $ \theta_{k_n,-}$ and $ \theta_{k_n,-}+k_n\cdot \omega $ can not  both belong to  it,  a contradiction. 
	\end{proof}
	\begin{lem}
		The above $k_n$ satisfies $10l_n<\|k_n\|\leq 10l_{n+1}^{(1)}$.
	\end{lem}
	\begin{proof}
		Suppose  that $\|k_n\|_1\leq 10l_n$. By the previous lemma and Proposition \ref{apn}, $E_{n}$ is $C^2$-closed  to $E_{n-1}$ of  order at most $\delta_{n-1}^{10}$. Since $E_n$ belongs to \textbf{Type} \ref{t1}, we have by Hypothesis \ref{h4},
		$$\pm E_n'|_{I (E_n)_\pm}\gtrsim  \delta_{n-2}^{10}.$$
		Thus $$\pm E_{n-1}'|_{I (E_n)_\pm}\gtrsim  \delta_{n-2}^{10}.$$
		
		By the two-monotonicity interval structure and mean value theorem, there exists some point $\theta_{k_{n-1},-}$ with 
		\begin{equation}\label{?}
			|\theta_{k_{n-1},-}-{\theta}_{k_n,-}|\lesssim \delta_{n-1}^{10}\delta_{n-2}^{-10}<\delta_{n-1}^9\end{equation} 
		such that    $$E_{n-1} (\theta_{k_{n-1},-})=E_{n-1} (\theta_{k_{n-1},-}+k_{n-1}\cdot \omega). $$
		
		Since $\|k_n\|_1\leq 10l_n$, by the construction, $\theta_{k_{n-1},-}$ has a distance larger than $\delta_{n-1}^9$ from  any simple resonant domain of $E_{n-1}$. Thus by \eqref{?},  ${\theta}_{k,-}$ does not belong to any simple resonant domain of $E_{n-1}$, a contradiction with Lemma \ref{ds}.
	\end{proof}
	
	By the previous lemma, we have $B_n\cap (B_n+k_n)=\emptyset$. Since $H_{B_n+k_n} (\theta)=H_{B_n} (\theta+k_n\cdot \omega)$,  $H_{B_n+k_n} (\theta)$ has an eigenvalue $E_n (\theta+k_n\cdot \omega)$. For convenience, we denote 
	\begin{equation}\label{not}
		B_{n,-}=B_n,\ E_{n,-} (\theta)=E_n (\theta),
	\end{equation}  and  the  eigenfunction of $E_{n,-}$ by $\psi_{n,-}$. 
	Similarly, we denote   
	\begin{equation}\label{not1}
		B_{n,+}=B_n+k_n,\ E_{n,+} (\theta)=E_n (\theta+k_n\cdot \omega ),
	\end{equation} 
	and  the  eigenfunction of $E_{n,+}$ by $\psi_{n,+}$.
	\begin{prop}\label{DRpron}
		If $\theta\in I^{(2)}_{k_n,\cup}$, then for any $\|x\|_1\leq  (l_{n+1}^{(1)})^{4}$, $x\notin\{o,k_n\}$,  we have 
		$$\theta+x\cdot \omega\notin I (E_{n-1}), \ n\geq 2,$$ 
		$$|v (\theta)-v (\theta+x\cdot \omega)|\geq 2\delta_0^{\frac{1}{400}}, \ n=1. $$
	\end{prop}
	\begin{proof}
		For $n\geq 2$, $I (E_{n-1})$ has at most two connected component with  length at most $2\delta_{n-3}^5\ll  (l_{n+1}^{(1)})^{-4\tau}$. Each of the component contains at least one point of $\{o,k_n\}$. The proof follows from  Diophantine condition immediately.
		The  case $n=1$ can be proved   similarly  as  Proposition \ref{DRpro}.
	\end{proof}
	We choose a block   $B_{n+1}^{(2)}$ satisfying 
	$$\Lambda_{l_{n+1}^{(2)}}\subset B_{n+1}^{(2)} \subset \Lambda_{l_{n+1}^{(2)}+50l_n}$$
	such that for all $\theta\in I_{k_n,\cup}^{(2)}$,  $B_{n+1}^{(2)}$ is $(n-1)$-strongly regular relative to $ (\theta,E)$ for all $E\in B_{2\delta_{n-1}} (E_{n,-} (\theta))$. 
	\begin{prop}\label{k2n}
		The following holds  for $\theta\in I^{(2)}_{k_n,\cup}$:
		\begin{itemize}
			\item[\textbf{(a).}]  $H_{B_{n+1}^{(2)}} (\theta)$ has exactly two  eigenvalues $E_{n+1}^{(2)} (\theta)$ and $\mathcal{E}_{n+1}^{(2)} (\theta)$ in  $[E_{n,-} (\theta)-10D\delta_{n-1}^{5},E_{n,-} (\theta)+10D\delta_{n-1}^{5}]$. Moreover, any other $\hat{E}\in\sigma (H_{B_{n+1}^{(2)}} (\theta)) $ must obey $|\hat{E}-E_{n,-} (\theta)|\geq \delta_{n-1}$.
			\item [\textbf{(b).}] The corresponding   eigenfunction of $E_{n+1}^{(2)}${\rm\  (resp. $\mathcal{E}_{n+1}^{(2)}$)}, $\psi_{n+1}${\rm\  (resp. $\Psi_{n+1}$)} decays exponentially fast away from $o$ and $k_n$, i.e.,
			$$|\psi_{n+1} (x)|\leq e^{-\frac{\gamma_0}{4}\|x\|_1}+e^{-\frac{\gamma_0}{4}\|x-k_n\|_1},$$
			$$|\Psi_{n+1} (x)|\leq e^{-\frac{\gamma_0}{4}\|x\|_1}+e^{-\frac{\gamma_0}{4}\|x-k_n\|_1}$$
			for $\operatorname{dist} (x,\{o,k_n\})\geq l_n^{6/7}.$	
			Moreover, the two eigenfunctions can be expressed as 
			\begin{align}\label{yn}
				\begin{split}
					\psi_{n+1}=A\psi_{n,-}+B\psi_{n,+}+O (\delta_n^{10}),\\
					\Psi_{n+1}=B\psi_{n,-}-A\psi_{n,+}+O (\delta_n^{10}),
				\end{split}	
			\end{align}
			where $A,B$ depend on $\theta$ satisfying $A^2+B^2=1$.
			\item [\textbf{(c).}] $\|G_{n+1}^{\perp\perp} (E_{n+1}^{(2)})\|,\|G_{n+1}^{\perp\perp} (\mathcal{E}_{n+1}^{(2)})\|\leq2\delta_{n-1}^{-1}$, where $G_{n+1}^{\perp\perp}$ denotes the Green's function for $B_{n+1}^{(2)}$ on the orthogonal complement of the space spanned by $\psi_{n+1}$ and $\Psi_{n+1}$.
		\end{itemize}
	\end{prop}
	\begin{proof}
		We first consider the case $\theta^*=\theta_{k_n,-}$.  In this case,  since $E_{n,-} (\theta^*)=E_{n,
			+} (\theta^*)$,  by the decay of the eigenfunctions (c.f. \eqref{Hdecay}), we have
		$$\| (H_{B_{n+1}^{(2)}} (\theta^*)-E_{n,-} (\theta^*))\psi_{n,-}\|=\|\Gamma_{B_{n,-}}\psi_{n,-}\|\lesssim e^{-\frac{\gamma_0}{4}l_n}\ll\delta_n^{10}, $$
		$$\| (H_{B_{n+1}^{(2)}} (\theta^*)-E_{n,+} (\theta^*))\psi_{n,+}\|=\|\Gamma_{B_{n,+}}\psi_{n,+}\|\lesssim e^{-\frac{\gamma_0}{4}l_n}\ll\delta_n^{10}.$$
		Since $B_{n,-}$ and   $B_{n,+}$ are disjoint, $\psi_{n,-}$ and $\psi_{n,+}$ serve as two trial wave functions  for the operator $H_{B_{n+1}^{(2)}} (\theta^*)-E_{n,-} (\theta^*)$. It follows that by  Lemma \ref{trialcor}, $H_{B_{n+1}^{(2)}} (\theta^*)$ has two  eigenvalues in  $[E_{n,-} (\theta^*)-O (\delta_n^{10}),E_{n,-} (\theta^*)+O (\delta_n^{10})]$. The claim for $\theta\in I^{(2)}_{k_n,\cup}$ follows from $|I^{(2)}_{k_n,\cup}|\leq 4\delta_{n-1}^{5}$ and   $|v'|\leq D$.
		
		To establish the decay of eigenfunctions, assume $E\in\sigma (H_{B_{n+1}}^{(2)} (\theta))$ with $|E-E_{n,-} (\theta)|\leq \delta_{n-1}$. We  denote  $\Lambda=B^{(2)}_{n+1}\setminus (B_{n-1}\cup (B_{n-1}+k_n))$, which is an $(n-1)$-regular set relative to $ (\theta,E)$ and $E_{n-1}$   by the construction. Moreover,  by Proposition \ref{DRpron}, $\Lambda$ is $(n-1)$-nonresonant, thus $G_\Lambda (E)$ decays exponentially,  which  together with restricting the eigenfunction  equation  to $\Lambda$   yields the exponential decay of $\psi_{n+1}$.
		
		To prove \eqref{yn}, it suffices to show $\psi_{n+1}$ and $\Psi_{n+1}$ are close to a linear combination of $\psi_{n,-}$ and $\psi_{n,+}$ inside $B_{n,-}\cup B_{n,+}$.
		We restrict the equation $H_{B_{n+1}^{(2)}} (\theta)\psi_{n+1}=E_{n+1}^{(2)} (\theta)\psi_{n+1}$ on $B_{n,-}$ to get 
		$$\left(H_{B_{n,-}}-E_{n+1}^{(2)}\right)\psi_{n+1}=\Gamma_{B_{n,-}}\psi_{n+1}.$$
		Thus  $$\|P_{n,-}^\perp\psi_{n+1}\|=\|G_{B_{n,-}}^\perp (E_{n+1}^{(2)})P_{n,-}^\perp\Gamma_{B_{n,-}}\psi_{n+1}\|=O (\delta_{n-1}^{-1}e^{-\frac{1}{4}\gamma_0l_n})\leq\delta_n^{10},$$
		where $P_{n,-}^\perp=I-\langle\psi_{n,-}|\psi_{n,-}\rangle$ is the projection onto the orthogonal complement of $\psi_{n,-}$ and $G_{B_{n,-}}^\perp (E_{n+1}^{(2)})$ is the Green's function of  $B_{n,-}$ on $\operatorname{Image}P_{n,-}^\perp$ with upper bound $O (\delta_{n-1}^{-1})$ by $|E_{n,-} (\theta)-E_{n+1}^{(2)} (\theta)|\lesssim\delta_{n-1}^5$ together with  the eigenvalue separation result of simple resonance case (c.f. Proposition \ref{kn} replacing $n$ by $n-1$). 
		Therefore inside $B_{n,-}$, we have 
		$$P_{n,-}^\perp\psi_{n+1}=O (\delta_n^{10}).$$ Thus  
		$$\psi_{n+1}\chi_{B_{n,-}}=a\psi_{n,-}+O (\delta_n^{10}).$$
		where $a=\langle \psi_{n+1},\psi_{n,-}\rangle$.\\
		Similarly, 
		$$\psi_{n+1}\chi_{B_{n,+}}=b\psi_{n,+}+O (\delta_n^{10})$$
		with $b=\langle \psi_{n+1},\psi_{n,+}\rangle$.\\ Since the eigenfunction decays exponentially, we have  $$\|\psi_{n+1}\chi_{B^{(2)}_{n+1}\setminus  (B_{n,-}\cup B_{n,+})}\|\leq \delta_n^{10}.$$ Thus we can write 
		$$\psi_{n+1}=a\psi_{n,-}+b\psi_{n,+}+O (\delta_n^{10}).$$
		Taking norm of the above equation, we obtain $k:=a^2+b^2=1-O (\delta_n^{10})$. We set $A=a/k$ and $B=b/k$. Hence $A^2+B^2=1$ and $|A-a|,|B-b|=O (\delta_n^{10})$, which gives the desired expression of $\psi_{n+1}$. A similar argument gives $\Psi_{n+1}=C\psi_{n,-}+D\psi_{n,+}+O (\delta_n^{10})$ with $C^2+D^2=1$. For convenience, we write $A=\cos\alpha, B=\sin\alpha, C=\sin\beta, D=-\cos\beta$. Since $\langle\psi_{n+1} ,\Psi_{n+1}\rangle=0$, we get  $|\sin (\beta-\alpha)|=O (\delta_n^{10})$. We can choose $\beta$ satisfying $|\beta-\alpha|=O (\delta_n^{10})$, thus $|B-C|=|\sin\alpha-\sin\beta|=O (\delta_n^{10})$ and $|A+D|=|\cos\alpha-\cos\beta|=O (\delta_n^{10})$, which gives  the desired expression $\Psi_{n+1}=B\psi_{n,-}-A\psi_{n,+}+O (\delta_n^{10})$.  Assume $\hat{E}\in\sigma (H_{B_{n+1}}^{(2)} (\theta))$ is a third eigenvalue in the interval $|\hat{E}-E_{n,-} (\theta)|\leq \delta_{n-1}$. The above argument  still holds if we replace $E_{n+1}^{(2)}$ by $\hat{E}$. Thus the  eigenfunction $\hat{\psi}$ corresponding to  $\hat{E}$ can also be expressed as 
		$$\hat{\psi}=\hat{A}\psi_{n,-}+\hat{B}\psi_{n,+}+O (\delta_n^{10})$$ 
		with $\hat{A}^2+\hat{B}^2=1.$ By orthogonality, we have $A\hat{A}+B\hat{B}=O (\delta_n^{10})$ and $B\hat{A}-A\hat{B}=O (\delta_n^{10})$, a contradiction since  $ (A\hat{A}+B\hat{B})^2+ (B\hat{A}-A\hat{B})^2=1$.
		Hence any other third eigenvalue must obey $|\hat{E}-E_{n,-} (\theta)|\geq\delta_{n-1}$. Finally,  \textbf{(c)} follows from \textbf{(a)} immediately.
	\end{proof}
	With the notation of the previous proposition, we denote  
	$$E_{n+1,>}^{(2)} (\theta)=\max (E_{n+1}^{(2)} (\theta), \mathcal{E}_{n+1}^{(2)} (\theta)), \   E_{n+1,<}^{(2)} (\theta)=\min (E_{n+1}^{(2)} (\theta), \mathcal{E}_{n+1}^{(2)} (\theta)).$$
	Define 
	$$P_{n,\pm}:=\langle\psi_{n,\pm}|\psi_{n,\pm}\rangle,\ Q_\pm:=\operatorname{Id}_{B_{n+1}^{(2)}}-P_{n,\pm},$$
	and consider the  compressed operators 
	$$H_\pm (\theta):=Q_\pm H_{B_{n+1}^{(2)}} (\theta) Q_\pm.$$
	\begin{lem}\label{it}
		For $\theta\in I^{(2)}_{k_n,\cup}$, the operators $H_\pm (\theta)$ have unique eigenvalues $\lambda_\mp (\theta)$ interlacing $E_{n+1,>}^{(2)} (\theta)$ and $E_{n+1,<}^{(2)} (\theta)$: 
		\begin{equation}\label{intern}
			E_{n+1,<}^{(2)} (\theta)\leq \lambda_\mp (\theta) \leq E_{n+1,>}^{(2)} (\theta).
		\end{equation}
		Moreover, $\lambda_\pm (\theta)$ are $C^1$-close to $E_{n,\pm} (\theta)$:
		\begin{equation}\label{Epln}
			|\frac{d^s}{d\theta^s} (\lambda_\pm (\theta)-E_{n,\pm} (\theta))|\lesssim e^{-\frac{1}{5}\gamma_0l_n}\delta_{n-1}^{-s}\ll \delta_n^{10},\  s=0,1.
		\end{equation}
		In particular, we have 
		\begin{equation}\label{tranln}
			\pm \lambda_\pm' (\theta)\gtrsim  \delta_{n-2}^{10}\gg \delta_{n-1} .
		\end{equation}
	\end{lem}
	\begin{proof}
		We claim  $H_\pm (\theta)$ have  unique eigenvalues $\lambda_\mp (\theta)$ in  $[E_{n,-} (\theta)-\delta_{n-1},E_{n,-} (\theta)+\delta_{n-1}]$ satisfying 
		$$|\lambda_\pm (\theta)-E_{n,\pm} (\theta)|\lesssim e^{-\frac15 \gamma_0l_n}.$$
		Since $Q_\pm\psi_{n,\mp}=\psi_{n,\mp}$, 
		$$\| (H_\pm (\theta)-E_{n,\mp} (\theta))\psi_{n,\mp}\|=\|Q_\pm\Gamma_{B_{n,\mp}}\psi_{n,\mp}\|\lesssim e^{-\frac15 \gamma_0l_n}.$$ The above estimate together with  Lemma \ref{trialcor}. yields the existence of $\lambda_\pm$.
		
		If $\hat{\lambda}$ is an eigenvalue  of $H_\pm (\theta)$ satisfying $|\hat{\lambda}-E_{n,-} (\theta)|\leq \delta_{n-1}$, then its  eigenfunction $\hat{\phi}\in \operatorname{Image}Q_\pm$ satisfies $$Q_\pm  (H_{B_{n+1}^{(2)}} (\theta)-\hat{\lambda})\hat{\phi}=0.$$
		Thus $$ (H_{B_{n+1}^{(2)}} (\theta)-\hat{\lambda})\hat{\phi}=a\psi_{n,\pm}.$$
		Employing  the same argument as the proof of item \textbf{(b)} of Proposition \ref{k2n}  together with the exponential decay of $\psi_{n,\pm}$, one can prove that 
		$$|\hat{\phi} (x)|\leq e^{-\frac{\gamma_0}{5}\|x\|_1}+e^{-\frac{\gamma_0}{5}\|x-k_n\|_1}$$ 
		for $\operatorname{dist} (x,\{o,k_n\})\geq l_n^{7/8}$ and 	 $\hat{\phi}=\psi_{n,\mp}+O ( e^{-\frac15 \gamma_0l_n}\delta_{n-1}^{-1})$ since  $\hat{\phi}\in \operatorname{Image}Q_\pm$.
		Thus such an eigenvalue must be unique or else it will violate the  orthogonality. By Theorem \ref{cahuchy} and item \textbf{(a)} of Proposition \ref{k2n}, $\lambda_\mp (\theta)$ must lie between $E_{n+1,<}^{(2)} (\theta)$ and  $ E_{n+1,>}^{(2)} (\theta)$.
		Let $\phi_\mp\in \operatorname{Image}Q_\pm$ be the eigenfunctions of $H_\pm$ corresponding to $\lambda_\mp$. Since $\|\phi_\mp-\psi_{n,\mp}\|\lesssim  e^{-\frac15 \gamma_0l_n} \delta_{n-1}^{-1}$, by the Feynman-Hellman formula \textbf{(1)} from Lemma \ref{daoshu}, we have 
		\begin{align*}
			|\lambda_\mp'-E_{n,\mp}'|&=|\left\langle\phi_\mp, (H_\pm)' \phi_\mp
			\right\rangle-\left\langle\psi_{n,\mp},H' \psi_{n,\mp}\right\rangle|\\
			&=|\left\langle\phi_\mp,V' \phi_\mp
			\right\rangle-
			\left\langle\psi_{n,\mp},V' \psi_{n,\mp}\right\rangle -2\langle\phi_\mp, (P_{n,\pm})'  H_{B_{n+1}^{(2)}} \phi_\mp
			\rangle|\\
			&\lesssim e^{-\frac15 \gamma_0l_n}\delta_{n-1}^{-1} + |\langle\phi_\mp,Q_{n,\pm} (P_{n,\pm})'  H_{B_{n+1}^{(2)}} Q_{n,\pm}\phi_\mp\rangle| .
		\end{align*}
		Thus it suffices to estimate $Q_{n,\pm}P_{n,\pm}'  H_{B_{n+1}^{(2)}} Q_{n,\pm}$. 
		By differentiating the relation $$ (H_{B_{n,\pm}}-E_{n,\pm})\psi_{n,\pm}=0,$$
		we get 
		$$ (H_{B_{n,\pm}}-E_{n,\pm})\psi_{n,\pm}'=- (V'-E_{n,\pm}')\psi_{n,\pm}.$$
		Since $\langle \psi_{n,\pm}',\psi_{n,\pm}\rangle=0$, we have   $$\psi_{n,\pm}'=-G_{B_{n,\pm}}^\perp (E_{n,\pm})V'\psi_{n,\pm}.$$
		Thus\begin{align*}
			P_{n,\pm}'&=\langle\psi_{n,\pm} |\psi_{n,\pm}'\rangle+\langle\psi_{n,\pm}'|\psi_{n,\pm}\rangle \\
			&=-G_{B_{n,\pm}}^\perp (E_{n,\pm})V'P_{n,\pm}-P_{n,\pm}V'G_{B_{n,\pm}}^\perp (E_{n,\pm}).
		\end{align*}
		Since $ P_{n,\pm}Q_{n,\pm}=Q_{n,\pm}P_{n,\pm}=0$, it follows that 
		\begin{equation}\label{69}
			Q_{n,\pm}P_{n,\pm}'  H_{B_{n+1}^{(2)}} Q_{n,\pm}=-Q_{n,\pm}G_{B_{n,\pm}}^\perp (E_{n,\pm})V'P_{n,\pm}  H_{B_{n+1}^{(2)}} Q_{n,\pm}.
		\end{equation}
		We can employ the  estimates 
		\begin{align*}
			\|	P_{n,\pm}  H_{B_{n+1}^{(2)}} Q_{n,\pm}	\|&=\|P_{n,\pm} (  H_{B_{n+1}^{(2)}}-E_{n,\pm} )Q_{n,\pm}\|	\\
			& \leq \|P_{n,\pm} (  H_{B_{n+1}^{(2)}}-E_{n,\pm})\|\\
			&\leq \|\Gamma_{B_{n,\pm}}\psi_{n,\pm}\|\lesssim e^{-\frac15 \gamma_0l_n}
		\end{align*} 
		and $\|G_{B_{n,\pm}}^\perp (E_{n,\pm})\|=O (\delta_{n-1}^{-1})$ to \eqref{69} and complete the proof of \eqref{Epln}.  
		Since $E_n$ belongs to \textbf{Type} \ref{t1}, by Hypothesis \ref{h4},  we have $$\pm E_{n,\pm}'\gtrsim \delta_{n-2}^{10}. $$
		
		Finally, \eqref{tranln} follows from the above estimate and \eqref{Epln} immediately.
	\end{proof}
	
	\begin{lem}\label{crossn}
		There is a unique point $\theta_s\in I^{(2)}_{k_n,\cup}$ satisfying $|\theta_s-\theta_{k_n,-}|<\delta_n^{10}$, such that 
		$$\lambda_+ (\theta_s)=\lambda_- (\theta_s).$$	
		Moreover, the Rellich children have  quantitative separation away from $\theta_s$:
		\begin{equation}\label{E1sn}
			E_{n+1,>}^{(2)} (\theta)-E_{n+1,<}^{(2)} (\theta)\geq |\lambda_+ (\theta)-\lambda_- (\theta)|\geq \delta_{n-1}|\theta-\theta_s|.
		\end{equation}
		As a corollary, if  $E_{n+1}^{(2)} (\theta)=\mathcal{E}_{n+1}^{(2)} (\theta)$, then $\theta=\theta_s$.
	\end{lem}
	\begin{proof}
		We consider the difference function 
		$$d (\theta):=\lambda_+ (\theta)-\lambda_- (\theta).$$
		Since $E_{n,+} (\theta_{k_n,-})=E_{n,-} (\theta_{k_n,-})$, by  \eqref{Epln} (the version $s=0$),  we have 
		$$|d (\theta_{k_n,-})|\leq |\lambda_+ (\theta_{k_n,-})-E_{n,+} (\theta_{k_n,-})|+|\lambda_- (\theta_{k_n,-})-E_{n,-} (\theta_{k_n,-})|\lesssim e^{-\frac{1}{5}\gamma_0l_n}.$$
		By   \eqref{tranln}, we have 
		$$d' (\theta)\geq\delta_{n-1}. $$
		Thus there exists  a unique  point $\theta_s$ satisfying $|\theta_s-\theta_{k_n,-}|\lesssim e^{-\frac{1}{5}\gamma_0l_n} \delta_{n-1}^{-1}<\delta_n^{10}$, such that 
		$d (\theta_s)=0$, that is 	$$\lambda_+ (\theta_s)=\lambda_- (\theta_s).$$	
		Then \eqref{E1sn} follows from \eqref{intern} and 
		$$	|\lambda_+ (\theta)-\lambda_- (\theta)|=|d (\theta)|=|d (\theta)-d (\theta_s)|\geq \delta_{n-1}|\theta-\theta_s|$$
		by mean value theorem.
	\end{proof}
	In the following, we assume that the Rellich functions are chosen so that $E_{n+1}^{(2)}>\mathcal{E}_{n+1}^{(2)}$ on the right of $\theta_s$.
	
	Suppose that $|E_{n,+}' (\theta_{k_n,-})\geq|E_{n,-}' (\theta_{k_n,-})|$  (the argument is similar for the opposite case). Recalling the derivative estimate for $E_n$ belonging  to \textbf{Type} \ref{t1} from   Hypotheses \ref{h4}, we define    $1\leq r\leq  \delta_{n-1}^{-1}$ such that 
	$$    |E_{n,+}' (\theta_{k_n,-})|=r|E_{n,-}' (\theta_{k_n,-})|.$$
	\begin{lem}\label{chan}
		For $\theta\in I_{k_n,\cup}^{(2)}$, we have 
		$$| (E_{n,+}'+rE_{n,-}') (\theta)|\leq \delta_{n-1}^2$$
	\end{lem}
	\begin{proof}
		Since $E_{n,+}'$ and $E_{n,-}'$ have opposite signs, we have 
		$$E_{n,+}' (\theta_{k_n,-})+rE_{n,-}' (\theta_{k_n,-})=0 . $$
		By the Feynman-Hellman formula \textbf{(2)} from  Theorem \ref{daoshu}, we have the bound 
		$$|E_{n,\pm}''|\lesssim D+\|G_{B_{n,\pm}} (E_{n,\pm})\|\lesssim \delta_{n-1}^{-1},$$
		where we employ the eigenvalue separation result of simple resonance case (c.f. Proposition \ref{kn} replacing $n$ by $n-1$) to bound the term $\|G_{B_{n,\pm}} (E_{n,\pm})\|$ in the above inequality.   
		Thus by mean value  theorem, we have 
		\begin{align*}
			| (E_{n,+}'+rE_{n,-}') (\theta)|&\leq  (1+r)\delta_{n-1}^{-1} |\theta-\theta_{k_n,-}| \\ 
			&\lesssim r\delta_{n-1}^{-1} |I_{k_n,\cup}^{(2)}| \\
			&\leq\delta_{n-1}^2.
		\end{align*}
	\end{proof}
	\begin{prop}\label{1215n}
		For $\theta\in I_{k_n,\cup}^{(2)}$, we have the following:
		\begin{itemize}
			\item[\textbf{(a).}] If $E_{n+1}^{(2)} (\theta)\neq\mathcal{E}_{n+1}^{(2)} (\theta)$, then 
			\begin{align}
				(E_{n+1}^{(2)})'& = (A^2-rB^2) E_{n,-}'+O (\delta_{n-1}^2 ), \label{den}\\
				(\mathcal{E}_{n+1}^{(2)})'& = (B^2-rA^2) E_{n,-}'+O (\delta_{n-1}^2 ),\label{deen}
			\end{align}
			and 	
			\begin{align}
				( E_{n+1}^{(2)})'' & =\frac{2\left\langle\psi_{n+1}, V' \Psi_{n+1}\right\rangle^2}{E_{n+1}^{(2)}-\mathcal{E}_{n+1}^{(2)}}+O (\delta_{n-1}^{-1} ),\label{dfn} \\
				( \mathcal{E}_{n+1}^{(2)})''& =\frac{2\left\langle\psi_{n+1}, V' \Psi_{n+1}\right\rangle^2}{\mathcal{E}_{n+1}^{(2)}-E_{n+1}^{(2)}}+O (\delta_{n-1}^{-1} )\label{ddfn}
			\end{align}
			with the notation from Proposition \ref{k2n} and $r$ as above.
			\item [\textbf{(b).}] If  $E_{n+1}^{(2)} (\theta)\neq \mathcal{E}_{n+1}^{(2)} (\theta)$ and  $| (E_{n+1}^{(2)})' (\theta)|\leq\delta_{n-1}^2$, then $| (E_{n+1}^{(2)})'' (\theta)|\geq2$. Moreover, the  sign of $ (E_{n+1}^{(2)})'' (\theta)$ is the same as that of $ E_{n+1}^{(2)} (\theta)- \mathcal{E}_{n+1}^{(2)} (\theta)$. The analogous conclusion holds if we exchange   $E_{n+1}^{(2)} (\theta)$ and $\mathcal{E}_{n+1}^{(2)} (\theta)$.
			\item[\textbf{(c).}] If there is a level crossing  (thus $E_{n+1}^{(2)} (\theta_s)=\mathcal{E}_{n+1}^{(2)} (\theta_s)$ by Lemma \ref{crossn}), then for all $\theta\in I_{k_n,\cup}^{ (2)}$,
			\begin{equation}\label{dan}
				(E_{n+1}^{(2)})' (\theta)>\delta_{n-1}^2,\  (\mathcal{E}_{n+1}^{(2)})' (\theta)<-\delta_{n-1}^2.
			\end{equation}
		\end{itemize}
	\end{prop}
	\begin{proof}
		We prove \eqref{den}  (the proof of \eqref{deen} is analogous). 	 By \eqref{yn} and Lemma \ref{chan}, we apply Feynman-Hellman formula \textbf{(1)} from Lemma \ref{daoshu} to obtain 
		\begin{align*}
			( E_{n+1}^{(2)})'
			& =\left\langle\psi_{n+1}, V' \psi_{n+1}\right\rangle =A^2 \left\langle\psi_{n,-}, V' \psi_{n,-}\right\rangle +B^2  \left\langle\psi_{n,+}, V' \psi_{n,+}\right\rangle+O (\delta_n^{10} ) \\
			&=A^2  E_{n,-}'+B^2  E_{n,+}'+O (\delta_n^{10} ) \\
			& = (A^2-rB^2) E_{n,-}'+B^2  (E_{n,+}'+rE_{n,-}')+O (\delta_n^{10} )\\
			& = (A^2-rB^2) E_{n,-}'+O (\delta_{n-1}^2 ).
		\end{align*}
		To prove \eqref{dfn} and \eqref{ddfn}, we  use Feynman-Hellman formulas  \textbf{(2)} and \textbf{(3)} to obtain 	
		\begin{equation*}
			(E_{n+1}^{(2)})''=\left\langle\psi_{n+1}, V'' \psi_{n+1}\right\rangle+2 \frac{\left\langle\psi_{n+1}, V' \Psi_{n+1}\right\rangle^2}{E_{n+1}^{(2)}-\mathcal{E}_{n+1}^{(2)}}-2\left\langle V' \psi_{n+1},G^{\perp \perp}_{n+1} (E_{n+1}^{(2)}) V' \psi_{n+1}\right\rangle .
		\end{equation*}
		The first term is bounded by $D$ and the  third  term is bounded by $$2\|G^{\perp \perp}_{n+1} (E_{n+1}^{(2)})\|\cdot\|V' \psi_{n+1}\|^2\lesssim \delta_{n-1}^{-1},$$ where we  use the estimate $\|G^{\perp \perp}_{n+1} (E_{n+1}^{(2)})\| \leq2\delta_{n-1}^{-1}$ from  item \textbf{(c)} of   Proposition \ref{k2n}. Thus we finish the proof of item \textbf{(a)}.
		
		Now we are going to prove \textbf{(b)}. We will show the first term in \eqref{dfn} is large if  $| (E_{n+1}^{(2)})' (\theta)|\leq\delta_{n-1}^2$.  
		Assuming  $| (E_{n+1}^{(2)})' (\theta)|\leq\delta_{n-1}^2$, by \eqref{den}, we have 
		$$|A^2-rB^2| |E_{n,-}' (\theta)|\leq| ( E_{n+1}^{(2)})' (\theta)|+O (\delta_{n-1}^2)=O (\delta_{n-1}^2).$$
		Since $$|E_{n,-}' (\theta)|\gtrsim \delta_{n-2}^{10}>\delta_{n-1} ,$$ it follows that 
		$$|A^2-rB^2|\lesssim \delta_{n-1}^2 \delta_{n-1}^{-1}<\frac{1}{100}.$$
		Since $A^2+B^2=1$ and $r\geq1$, we obtain  
		$$ (1+r)B^2\geq A^2+B^2-|A^2-rB^2|\geq \frac{99}{100}$$ and 
		$$ (1+\frac{1}{r})A^2\geq A^2+B^2-\frac{1}{r}|A^2-rB^2|\geq \frac{99}{100}.$$
		Thus $$B^2\geq\frac{1}{4r},\ A^2\geq\frac{1}{4}.$$
		By \eqref{yn}, Lemma  \ref{chan}  and the previous lower bound for $A,B$,  we have 
		\begin{align*}
			|\langle\psi_{n+1}, V' \Psi_{n+1}\rangle|&=|AB (\langle\psi_{n,-}, V' \psi_{n,-}\rangle-\langle\psi_{n,+}, V' \psi_{n,+}\rangle)+O (\delta_n^{10} )| \\
			&=|AB   (E_{n,-}' -E_{n,+}')+O (\delta_n^{10})|\\
			&=|AB \left( (1+r) E_{n,-}'- (E_{n,+}'+r E_{n,-}')\right)+O (\delta_n^{10})|\\
			& \geq \frac{1+r }{4\sqrt{r}}  (|E_{n,-}'|-O (\delta_{n-1}^2 ))\\
			&\gtrsim \delta_{n-2}^{10}\gg \delta_{n-1}.
		\end{align*}
		By item \textbf{(a)} of  Proposition \ref{k2n}, we obtain  the estimate of the denominator  $|E_{n+1}^{(2)}-\mathcal{E}_{n+1}^{(2)}|\lesssim \delta_{n-1}^5$. Combining  the previous estimate of numerator, by \eqref{dfn},  we obtain  $$| (E_{n+1}^{(2)})'' (\theta)|\gtrsim     \delta_{n-1}^{2}  \delta_{n-1}^{-5}-O (\delta_{n-1}^{-1})>2,$$ 
		whose sign is the same as  that of $ E_{n+1}^{(2)} (\theta)- \mathcal{E}_{n+1}^{(2)} (\theta)$.
		
		To prove \textbf{(c)}, we first show \eqref{dan} holds for $\theta=\theta_s$. Since $E_{n+1}^{(2)} (\theta_s)=\mathcal{E}_{n+1}^{(2)} (\theta_s)$, by \eqref{intern} we have 
		$$E_{n+1}^{(2)} (\theta_s)=\mathcal{E}_{n+1}^{(2)} (\theta_s)=\lambda_+ (\theta_s)=\lambda_- (\theta_s),$$ 
		and since $E_{n+1}^{(2)}>\mathcal{E}_{n+1}^{(2)}$ on the right of $\theta_s$, we have for $\theta>\theta_s$, 
		$$E_{n+1}^{(2)} (\theta)\geq\lambda_+ (\theta)>\lambda_- (\theta) \geq \mathcal{E}_{n+1}^{(2)} (\theta).$$
		Thus by \eqref{tranln},
		we get $$   (E_{n+1}^{(2)})' (\theta_s)\geq \lambda_+' (\theta_s) >\delta_{n-1},\   (\mathcal{E}_{n+1}^{(2)})' (\theta_s)\leq \lambda_-' (\theta_s)<-\delta_{n-1}.$$
		We next claim  the inequalities hold for all $\theta\in I_{k_n,\cup}^{(2)}$. We prove the case  that  $ (E_{n+1}^{(2)})' (\theta)$ and $\theta>\theta_s$  (the other cases are analogous). If it  is not true, then the set 
		$$\{\theta \in I_{k_n,\cup}^{(2)}: \ \theta>\theta_s, \   (E_{n+1}^{(2)})' (\theta)\leq  \delta_{n-1}^2 \}\neq \emptyset.$$ 
		Let $\theta^*$ be its infimum. Since $ (E_{n+1}^{(2)})' (\theta)$ is continuous and $  (E_{n+1}^{(2)})' (\theta_s) >\delta_{n-1}^2$, we have $\theta^*>\theta_s$ and  $ (E_{n+1}^{(2)})' (\theta)>\delta_{n-1}^2\geq  (E_{n+1}^{(2)})' (\theta^*)$ for $\theta\in [\theta_s,\theta^*)$, which implies $  (E_{n+1}^{(2)})'' (\theta^*)\leq 0$. However, by item \textbf{(b)} and $E_{n+1}^{(2)} (\theta^*)>\mathcal{E}_{n+1}^{(2)} (\theta^*)$, we get   $ (E_{n+1}^{(2)})'' (\theta^*)>2$, a contradiction. Thus we  prove our claim. 
	\end{proof}
	
	Define 
	$$E_{n,\vee} (\theta)=\max (E_{n,+} (\theta), E_{n,-} (\theta)), \   E_{n,\wedge} (\theta)=\min (E_{n,+} (\theta), E_{n,-} (\theta)).$$
	\begin{prop}
		For  $\theta\in I_{k_n,\cup}^{(2)}$, we have 
		\begin{equation}\label{apbpn}
			|E_{n+1,>}^{(2)} (\theta)-E_{n,\vee} (\theta)|,| E_{n+1,<}^{(2)} (\theta)-E_{n,\wedge} (\theta)|\lesssim e^{-\frac14 \gamma_0l_n}.
		\end{equation}
	\end{prop}
	\begin{proof}
		Fix  $\theta\in I_{k_n,\cup}^{(2)}$. We have
		$$\| (H_{B_{n+1}^{(2)}}-E_{n,\pm})\psi_{n,\pm}\|=\|\Gamma_{B_{n,\pm}}\psi_{n,\pm}\|\lesssim e^{-\frac14 \gamma_0l_n}. $$ 
		By Lemma \ref{trialcor}, we deduce  $H_{B_{n+1}^{(2)}}$ must have an eigenvalue in $[E_{n,\vee}-O (e^{-\frac14 \gamma_0l_n}), E_{n,\vee}+O (e^{-\frac14 \gamma_0l_n})]$ and in $[E_{n,\wedge}-O (e^{-\frac14 \gamma_0l_n}), E_{n,\wedge}+O (e^{-\frac14 \gamma_0l_n})]$. By  item \textbf{(a)} of Proposition \ref{k2n}, the eigenvalue must be $E_{n+1}^{(2)}$ or $\mathcal{E}_{n+1}^{(2)}$.  If the two intervals are disjoint, \eqref{apbpn} must hold.
		Otherwise, we have \begin{equation}\label{othern}
			E_{n,\vee}-E_{n,\wedge}\lesssim e^{-\frac14 \gamma_0l_n}.
		\end{equation} Thus we have two trial functions  for  $E_{n,\vee}$: 
		$$\| (H_{B_{n+1}^{(2)}}-E_{n,\vee})\psi_{n,\pm}\|\lesssim e^{-\frac14 \gamma_0l_n}.$$ 
		By Lemma \ref{trialcor},  $H_{B_{n+1}^{(2)}}$ must have two eigenvalues in $[E_{n,\vee}-O (e^{-\frac14 \gamma_0l_n}), E_{n,\vee}+O (e^{-\frac14 \gamma_0l_n})].$ They are $E_{n+1}^{(2)}$ and  $\mathcal{E}_{n+1}^{(2)}$. Combining \eqref{othern}, we also obtain \eqref{apbpn}.
	\end{proof}
	
	The next proposition shows  the two-monotonicity interval structure of $E_{n+1,>}^{(2)}$ and $E_{n+1,<}^{(2)}$.
	\begin{prop}\label{str2} We have the following:
		\begin{itemize}
			\item [\textbf{(a).}] If there is no level crossing, then $E_{n+1}^{(2)}=E_{n+1,>}^{(2)}$ and $\mathcal{E}_{n+1}^{(2)}=E_{n+1,<}^{(2)}$. Moreover, $E_{n+1}^{(2)}$   (resp. $\mathcal{E}_{n+1}^{(2)}$) has the two-monotonicity interval structure with a  critical point $\theta_>$   (resp. $\theta_<$) satisfying $|\theta_>-\theta_{k_n,-}|<\delta_n^{10}$  (resp. $|\theta_<-\theta_{k_n,-}|<\delta_n^{10}$).
			\item [\textbf{(b).}] If there is a level crossing  (thus $E_{n+1}^{(2)} (\theta_s)=\mathcal{E}_{n+1}^{(2)} (\theta_s)$),  then 
			\begin{align*}
				E_{n+1}^{(2)} (\theta)=	\left\{\begin{aligned}
					&E_{n+1,<}^{(2)} (\theta) &\text{if $\theta<\theta_s$,}\\
					&E_{n+1,>}^{(2)} (\theta) &\text{if $\theta>\theta_s$,}
				\end{aligned}\right. 
			\end{align*}
			and 
			\begin{align*}
				\mathcal{E}_{n+1}^{(2)} (\theta)=	\left\{\begin{aligned}
					&E_{n+1,>}^{(2)} (\theta) &\text{if $\theta<\theta_s$,}\\
					&E_{n+1,<}^{(2)} (\theta) &\text{if $\theta>\theta_s$.}
				\end{aligned}\right.	\end{align*}
			Moreover, $E_{n+1,>}^{(2)}$ and $E_{n+1,<}^{(2)}$ are piecewise $C^1$ functions  (except at the point $\theta_s$) with two-monotonicity interval structure.
		\end{itemize}
		Combining the previous propositions of  double resonance case,  we deduce that if the gap between the critical values of $E_{n+1,>}^{(2)}$ and $E_{n+1,<}^{(2)}$ is larger than $3\delta_{n+1}$, then each of them belongs to \textbf{Type} \ref{t2}, otherwise, $E_{n+1,>}^{(2)}$ along with  $E_{n+1,<}^{(2)}$ belongs to \textbf{Type} \ref{t3}.
	\end{prop}
	\begin{proof}
		If there is no level crossing, then $E_{n+1,>}^{(2)} (\theta)>E_{n+1,<}^{(2)} (\theta)$ for all $\theta \in I_{k_n,\cup}^{(2)}$. Since $E_{n+1}^{(2)}>\mathcal{E}_{n+1}^{(2)}$ on the right of $\theta_s$ and by continuity,   $E_{n+1}^{(2)} (\theta)>\mathcal{E}_{n+1}^{(2)} (\theta)$ for all $\theta \in I_{k_n,\cup}^{(2)}$. Thus $E_{n+1}^{(2)}=E_{n+1,>}^{(2)}$ and $\mathcal{E}_{n+1}^{(2)}=E_{n+1,<}^{(2)}$.
		
		Now we show the existence of the critical point $\theta_>$ for $E_{n+1}^{(2)}$  (the proof of $\theta_<$ and $\mathcal{E}_{n+1}^{(2)}$ is analogous).  Let $\theta_-:=\theta_{k_n,-}-\delta_n^{10}$ and $\theta_+:=\theta_{k_n,-}+\delta_n^{10}$.  Recalling  $\pm E_{n,\pm}'\gtrsim \delta_{n-2}^{10}>\delta_n$ and \eqref{apbpn}, we have by mean value theorem 
		$$E_{n+1}^{(2)} (\theta_+)\geq  E_{n,+} (\theta_+)-O ( e^{-\frac14 \gamma_0l_n})\geq   E_{n,+} (\theta_{k_n,-})+\delta_n^{11}-O ( e^{-\frac14 \gamma_0l_n}) 
		> E_{n+1}^{(2)} (\theta_{k_n,-}) $$ and 
		$$E_{n+1}^{(2)} (\theta_-)\geq  E_{n,-} (\theta_-)-O ( e^{-\frac14 \gamma_0l_n})\geq   E_{n,-} (\theta_{k_n,-})+\delta_n^{11}-O ( e^{-\frac14 \gamma_0l_n}) 
		> E_{n+1}^{(2)} (\theta_{k_n,-}). $$
		Thus $E_{n+1}^{(2)}$ must have a critical point $\theta_>\in [\theta_-,\theta_+]$ by mean value theorem.
		The two-monotonicity interval structure follows from  item \textbf{(b)} of Proposition \ref{1215n} and Lemma \ref{C2}. Thus we finish the proof of \textbf{(a)}.
		
		If $E_{n+1}^{(2)} (\theta_s)=\mathcal{E}_{n+1}^{(2)} (\theta_s)$, by item \textbf{(c)} of Proposition \ref{1215n}, $E_{n+1}^{(2)}$ is strictly increasing and $\mathcal{E}_{n+1}^{(2)}$ is strictly decreasing. Thus $E_{n+1}^{(2)}>\mathcal{E}_{n+1}^{(2)}$ on the right of $\theta_s$ and $E_{n+1}<\mathcal{E}_{n+1}^{(2)}$ on the left  of $\theta_s$.  The proof of \textbf{(b)} is obvious.
	\end{proof}
	\subsubsection{Completion of the simple resonance case}In this part, we will complete the  proof of    \textbf{Subcase B} in  the  simple resonance case. In this subcase, $E_n$ belongs to \textbf{Type} \ref{t3}  and the simple resonant interval we consider is 
	\begin{equation}\label{SBI}
		J_{i_n}^{(1)}=   [\sup J (E_{n,<})-\frac{1}{2}\delta_{n-1}^{10}-3\delta_n,\inf J (  E_{n,>})+\frac{1}{2}\delta_{n-1}^{10}+3\delta_n]. 
	\end{equation} Recall that if $E_n$ belongs  to \textbf{Type} \ref{t3},  there exist two brother Rellich functions with  $0\leq \inf E_{n,>} (\theta)-\sup E_{n,<} (\theta)\leq 3\delta_n$.
	\begin{figure}[htp]
		\begin{tikzpicture}[>=latex, scale=1]
			\draw[domain=-4:0, samples=100] plot  ({\x},{\x*\x/5});
			\draw[domain=0:2, samples=100] plot  ({\x},{\x*\x/1.25});
			\draw[domain=-2:0, samples=100] plot  ({\x-0.5},{-0.5-\x*\x/1.25});
			\draw[domain=0:4, samples=100] plot  ({\x-0.5},{-0.5-\x*\x/5});
			\draw (0,0)node[below]{gap$\leq3\delta_n$}; 	
			
			\draw (1.1,2)node{$E_{n,>}$};  	 \draw (1,-2)node{$E_{n,<}$}; 
			\draw[dashed]  (-4,1)--  (4,1); 	\draw[dashed]  (-4,-1.5)--  (4,-1.5); 
			\draw[dashed,<->]  (2.5,1)--node[fill=white]{$J_{i_n}^{(1)}$}  (2.5,-1.5); 
			\node[rotate = 270] at  (-3, -0.25) {$\underbrace{\hspace{2.5cm}}$};
			\draw (-3,-0.25)node[left]{$\delta_{n-1}^{10}\sim$};  	
		\end{tikzpicture}
		\caption{A cartoon illustration of the interval  $J_{i_n}^{(1)}$ in \textbf{Subcase B}: $E_n$ belongs to \textbf{Type 3}, and the gap between the Rellich brothers is smaller than $3\delta_n$. $J_{i_n}^{(1)}$ with length $\sim\delta_{n-1}^{10}$  covers the nearby region of the critical values of the two Rellich brothers.}
	\end{figure}
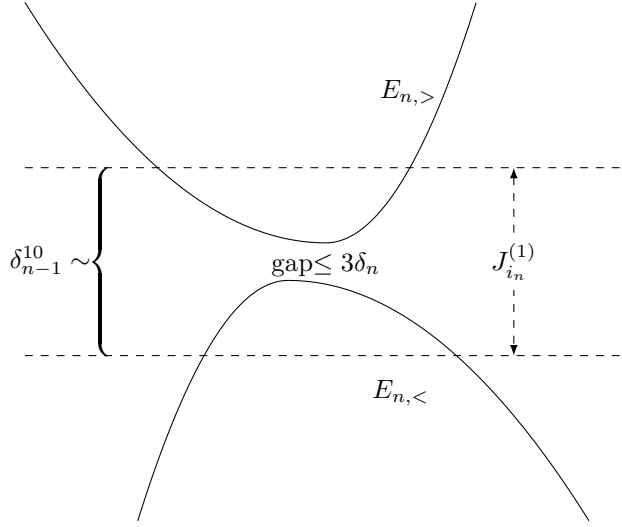
	
	\begin{lem}\label{Ta}
		$E_{n}$	must have  some ancestor $E_s\  (1\leq s\leq  n )$ generated from the double resonant interval of its parent $E_{s-1}$. 
	\end{lem}
	\begin{proof}
		For the sake of contradiction, assume that $E_s\in \mathcal{C}^{(1)} (E_{s-1})$ for all $1\leq s\leq n$. By the previous analysis of the simple resonance case  (c.f. Propositions  \ref{k1}, \ref{str1}, \ref{strn1}), $E_n$ must belong to  \textbf{Type} \ref{t1} or  \textbf{Type} \ref{t2}, a contradiction.
	\end{proof}
	Let $m$ be the largest integer  of $0\leq m\leq n-1$ such that $E_{m+1}$ is  generated from the double resonant interval of its parent  $E_m$ among the ancestors of $E_n$, that is, there exists some $k_m\in \Z^d$ with $10l_m<\|k_m\|_1\leq 10l_{m+1}$ satisfying  
	$$e_{k_m}:=E_m (\theta_{k_m,-})=E_m (\theta_{k_m,-}+k_m\cdot \omega )$$ 
	for some $\theta_{k_m,-}\in I (E_m)_-$ and   $E_{m+1,>},E_{m+1,<}\in \mathcal{C}^{(2)} (E_m)$  are the Rellich children generated from the double resonant interval $\bar{B}_{\delta_{m-1}^{10}} (e_{k_m})$. Thus the results of the  double resonance case from the previous section hold true with $m$ replacing $n$.
	\begin{lem}
		For all $m+1\leq s\leq n-1$, $E_s$ belongs to  \textbf{Type} \ref{t3} and  $E_{s+1}$ is generated from  the interval $J_{i_{s}}^{(1)}$ (\eqref{SBI} replacing $n$ by $s$) of $E_s$.
	\end{lem} 
	\begin{proof}
		By the definition of $m$, $E_{s+1}\in \mathcal{C}^{(1)} (E_s)$ for all $m+1\leq s\leq n-1$. If the lemma is not true, as the proof of Lemma \ref{Ta},  $E_n$ must  belong to  \textbf{Type} \ref{t1} or  \textbf{Type} \ref{t2}, a contradiction.
	\end{proof}
	Define the domains 
	$$I^{(1)}_{>}:= (E_{n,>})^{-1}\left (J_{i_n}^{(1)}\cap J (E_{n,>})\right), \ I^{(1)}_{<}:= (E_{n,<})^{-1}\left (J_{i_n}^{(1)} \cap J(E_{n,<})\right) $$ 
	and 
	$$I^{(1)}_{\cap}:=I^{(1)}_{>}\cap I^{(1)}_{<}, \ I^{(1)}_{\cup}:=I^{(1)}_{>}\cup I^{ (1)}_{<}.$$
	\begin{lem}\label{in}
		$I^{(1)}_{\cap}\neq \emptyset$. Thus $I^{(1)}_{\cup}$ is a single interval with $\delta_{n-1}^{10}\lesssim |I^{(1)}_{\cup}|\lesssim \delta_{n-1}^5$. Moreover, let  $\theta_{n,>}$ (resp. $\theta_{n,<}$) be  the critical point of $E_{n,>}$ (resp. $E_{n,<}$), then    $\theta_{n,>},\theta_{n,<}\in I^{(1)}_{\cap}$. 
	\end{lem}
	\begin{proof}
		We assume that $\theta_{n,>}\geq \theta_{n,<}$.  (The other case is analogous.) Recalling Hypothesis \ref{h4}, there exist two interlacing functions $\lambda_{n,\pm}$ satisfying 
		$$	E_{n,<} (\theta)\leq \lambda_{n,\pm} (\theta) \leq E_{n,>} (\theta),$$
		$$ \pm \lambda_{n,\pm}' (\theta) \geq \delta_{n-2}.$$  
		Thus \begin{align*}
			\delta_{n-2} (\theta_{n,>}-\theta_{n,<})&\leq  \lambda_{n,+} (\theta_{n,>})-\lambda_{n,+} (\theta_{n,<})\\&\leq E_{n,>} (\theta_{n,>})-E_{n,<} (\theta_{n,<})\leq 3\delta_n.
		\end{align*} 
		From the above inequality, we obtain 
		\begin{equation}\label{91}
			|\theta_{n,>}-\theta_{n,<}|\leq  3\delta_n \delta_{n-2}^{-1}\ll\delta_{n-1}^{10}.
		\end{equation}
		Since $|J_{i_n}^{(1)}\cap J (E_{n,>})|\geq \frac{1}{2} \delta_{n-1}^{10}$, by the uniform bound of $|v'|$, we have 
		\begin{equation}\label{92}
			B_{\frac{\delta_{n-1}^{10}}{2D}} (\theta_{n,>})\subset I^{(1)}_{>}.
		\end{equation}
		Similarly, 
		\begin{equation}\label{93}
			B_{\frac{\delta_{n-1}^{10}}{2D}} (\theta_{n,<})\subset I^{(1)}_{<}.
		\end{equation}
		The lemma follows from \eqref{91}, \eqref{92} and \eqref{93} immediately.
	\end{proof}
	We choose a block  a block $B_{n+1}^{(1)}$ satisfying 
	$$\Lambda_{l_{n+1}^{(1)}}\subset B_{n+1}^{(1)} \subset \Lambda_{l_{n+1}^{(1)}+50l_n}$$
	such that for all $\theta\in I^{(1)}_{\cup}$,  $B_{n+1}^{(1)}$ is $(n-1)$-strongly regular relative to $ (\theta,E)$ for all $E\in B_{2\delta_{n-1}} (E_{n,>} (\theta))$. 
	
	\begin{prop}\label{k2n+}
		The following holds  for $\theta\in I^{(1)}_{\cup}$:
		\begin{itemize}
			\item[\textbf{(a).}]  $H_{B_{n+1}^{(1)}} (\theta)$ has  two  eigenvalues $E_{n+1,>} (\theta)$ and $ E_{n+1,<} (\theta)$ with $E_{n+1,>} (\theta)\geq  E_{n+1,<} (\theta)$ satisfying 
			\begin{align}\label{ex}
				\begin{split}
					|E_{n+1,>} (\theta)- E_{n,>} (\theta)|\lesssim e^{-\frac{\gamma_0}{4}l_n}\ll\delta_n^{10},\\
					|E_{n+1,<} (\theta)- E_{n,<} (\theta)|\lesssim e^{-\frac{\gamma_0}{4}l_n}\ll\delta_n^{10}.
				\end{split}	
			\end{align}
			Moreover, any other $\hat{E}\in\sigma (H_{B_{n+1}^{(1)}} (\theta)) $ must obey $|\hat{E}-E_{n,>} (\theta)|\geq \delta_{n-1}$.
			\item [\textbf{(b).}] The corresponding   eigenfunction of $E_{n+1,>}${\rm\   (resp. $E_{n+1,<}$)}, $\psi_{n+1,>}${\rm\  (resp. $\psi_{n+1,<}$)} decays exponentially fast away from $o$ and $k_m$, i.e.,
			$$|\psi_{n+1,>} (x)|\leq e^{-\frac{\gamma_0}{4}\|x\|_1}+e^{-\frac{\gamma_0}{4}\|x-k_m\|_1},$$
			$$|\psi_{n+1,<} (x)|\leq e^{-\frac{\gamma_0}{4}\|x\|_1}+e^{-\frac{\gamma_0}{4}\|x-k_m\|_1}$$
			for $\operatorname{dist} (x,\{o,k_m\})\geq l_m^{6/7}.$	
			Moreover, the two eigenfunctions can be expressed as 
			\begin{align}\label{yn+}
				\begin{split}
					&\psi_{n+1,>}=A\psi_{m,-}+B\psi_{m,+}+O (\delta_m^{10}),\\
					&\psi_{n+1,<}=B\psi_{m,-}-A\psi_{m,+}+O (\delta_m^{10}),
				\end{split}	
			\end{align}
			where $A,B$ depend on $\theta$ satisfying $A^2+B^2=1$.
			\item [\textbf{(c).}] $\|G_{n+1}^{\perp\perp} (E_{n+1,>})\|,\|G_{n+1}^{\perp\perp} (E_{n+1,<})\|\leq2\delta_{n-1}^{-1}$, where $G_{n+1}^{\perp\perp}$ denotes the Green's function for $B_{n+1}^{(1)}$ on the orthogonal complement of the space spanned by $\psi_{n+1,>}$ and $\psi_{n+1,<}$.
		\end{itemize}
	\end{prop}
	\begin{proof}
		Fix  $\theta\in I_{\cup}^{(1)}$. Denote by $\psi_{n,>}$ (resp. $\psi_{n,<}$) the eigenfunction corresponding to $E_{n,>}$ (resp. $E_{n,<}$)  of the restriction operator $H_{B_n}$. By the exponential decay \eqref{Hdecay} of the eigenfunctions from the inductive hypothesis, we have 
		$$\| (H_{B_{n+1}^{(1)}}-E_{n,>})\psi_{n,>}\|=\|\Gamma_{B_{n}}\psi_{n,>}\|\lesssim e^{-\frac14 \gamma_0l_n}, $$ 
		$$\| (H_{B_{n+1}^{(1)}}-E_{n,<})\psi_{n,<}\|=\|\Gamma_{B_{n}}\psi_{n,<}\|\lesssim e^{-\frac14 \gamma_0l_n}.$$ 
		The above estimates together with Lemma \ref{trialcor} yield   $H_{B_{n+1}^{(1)}}$ must have an eigenvalue in $[E_{n,>}-O (e^{-\frac14 \gamma_0l_n}), E_{n,>}+O (e^{-\frac14 \gamma_0l_n})]$ and in $[E_{n,<}-O (e^{-\frac14 \gamma_0l_n}), E_{n,<}+O (e^{-\frac14 \gamma_0l_n})]$. If the  two intervals are disjoint, \eqref{ex} must hold.
		Otherwise, we have \begin{equation}\label{othern+}
			E_{n,>}-E_{n,<}\lesssim e^{-\frac14 \gamma_0l_n}.
		\end{equation} Thus we have two trial functions  for  $E_{n,>}$: 
		$$\| (H_{B_{n+1}^{(1)}}-E_{n,>})\psi_{n,>}\|\lesssim e^{-\frac14 \gamma_0l_n},$$
		$$\| (H_{B_{n+1}^{(1)}}-E_{n,>})\psi_{n,<}\|\leq |	E_{n,>}-E_{n,<}|+ \| (H_{B_{n+1}^{(1)}}-E_{n,<})\psi_{n,<}\| \lesssim e^{-\frac14 \gamma_0l_n}.$$ 
		By Lemma \ref{trialcor},  $H_{B_{n+1}^{(1)}}$ must have two eigenvalues in $[E_{n,>}-O (e^{-\frac14 \gamma_0l_n}), E_{n,>}+O (e^{-\frac14 \gamma_0l_n})].$  Combining \eqref{othern+}, we also obtain \eqref{ex}. Thus we finish  the existence part of \textbf{(a)}.
		
		We now prove \textbf{(b)}. Assume $E\in\sigma (H_{B_{n+1}}^{(1)} (\theta))$ with $|E-E_{n,>} (\theta)|\leq \delta_{n-1}$. Fix  $\operatorname{dist} (x,\{o,k_m\})\geq l_m^{6/7}$. Suppose 	 $\|x\|_1\leq l_{m+2}/2$,  restricting  the eigenfunction equation to the $(m-1)$-good set  $\Lambda:=B_{m+2}\setminus  (B_{m-1}\cup (B_{m-1}+k_m))$ together with the exponential decay of $G_\Lambda (E)$ yields the desired decay of eigenfunctions.  Suppose  $l_{s}/2\leq\|x\|_1\leq l_{s+1}/2$ for some  $m+2\leq s\leq  n-1$ or $l_{n}/2\leq\|x\|_1$ for $s=n$, restricting  the eigenfunction equation to the $(s-1)$-good set  $\Lambda:=B_{s+1}\setminus B_{s-1}$ together with the exponential decay of $G_\Lambda (E)$ yields the desired decay.  
		We explain that the above  set $B_{s+1}\setminus B_{s-1}$ is  $(s-1)$-nonresonant for $m+2\leq s\leq n$ since $I (E_{s-1})$ is a single interval satisfying   $| I (E_{s-1})|\lesssim\delta_{s-3}^5$, which together with Diophantine condition yields $\theta+x\cdot \omega \notin I (E_{s-1})$ for all $x\in B_{s+1}\setminus B_{s-1}$. The  $(s-1)$-regularity of $B_{s+1}\setminus B_{s-1}$ is confirmed by the construction.
		
		To prove \eqref{yn+}, it suffices to show $\psi_{n+1,>}$ and $\psi_{n+1,<}$ are close to a linear combination of $\psi_{m,-}$ and $\psi_{m,+}$ inside $B_{m,-}\cup B_{m,+}$.
		We restrict the equation $H_{B_{n+1}^{(1)}} (\theta)\psi_{n+1,>}=E_{n+1,>} (\theta)\psi_{n+1,>}$ on $B_{m,-}$ to get 
		$$\left(H_{B_{m,-}}-E_{n+1,>}\right)\psi_{n+1,>}=\Gamma_{B_{m,-}}\psi_{n+1,>}.$$
		Thus  $$\|P_{m,-}^\perp\psi_{n+1,>}\|=\|G_{B_{m,-}}^\perp (E_{n+1,>})P_{m,-}^\perp\Gamma_{B_{m,-}}\psi_{n+1,>}\|=O (\delta_{m-1}^{-1}e^{-\frac{1}{4}\gamma_0l_m})\leq\delta_m^{10},$$
		where $P_{m,-}^\perp=I-\langle\psi_{m,-}|\psi_{m,-}\rangle$ is the projection onto the orthogonal complement of $\psi_{m,-}$ and $G_{B_{m,-}}^\perp (E_{n+1,>})$ is the Green's function of  $B_{m,-}$ on $\operatorname{Image}P_{m,-}^\perp$ with upper bound $O (\delta_{m-1}^{-1})$ by $|E_{n+1,>} (\theta)-E_{m,-} (\theta)|\lesssim|J_{i_{m+1}}^{(1)}|\lesssim\delta_{m-1}^5$ together with  the eigenvalue separation  result of simple resonance case (c.f. Proposition \ref{kn} replacing $n$ by $m-1$). 
		Thus  
		$$\psi_{n+1,>}\chi_{B_{m,-}}=a\psi_{m,-}+O (\delta_m^{10}).$$
		Similarly, 
		$$\psi_{n+1,>}\chi_{B_{m,+}}=b\psi_{m,+}+O (\delta_m^{10}).$$
		The above  estimates together with the exponential  decay of the eigenfunction and orthogonality of  $\psi_{n+1,>}$ and $\psi_{n+1,<}$ yield \eqref{yn+}. 
		
		Since the  eigenfunction corresponding to $E$ in the interval $|E-E_{n,>} (\theta)|\leq  \delta_{n-1}$ has  the form \eqref{yn+}, there are only  two such eigenvalues $E_{n+1,>}$ and $ E_{n+1,<}$ or else it will violate the orthogonality.
		
		Finally,  \textbf{(c)} follows from \textbf{(a)} and  $|E_{n+1,>} (\theta)- E_{n,>} (\theta)|, |E_{n+1,<} (\theta)- E_{n,>} (\theta)| \lesssim\delta_{n-1}^5$ immediately.
	\end{proof}
	
	Define 
	$$P_{m,\pm}:=\langle\psi_{m,\pm}|\psi_{m,\pm}\rangle$$ and  $$Q_{s,\pm}:=\operatorname{Id}_{B_{s}}-P_{m,\pm},\ m+1\leq s\leq n+1. $$
	and consider the  compressed operators 
	$$H_{s,\pm} (\theta):=Q_{s,\pm} H_{B_s} (\theta) Q_{s,\pm}.$$
	{
		%	xiugai5
		In the following proposition, we will  use a multi-scale interlacing argument to construct the interlacing functions of $E_{s,>} (\theta)$ and $E_{s,<} (\theta)$ for all scales $m+1\leq s\leq n+1$. Unlike Lemma \ref{it} which starts  the construction at  the scale $n$, the following proposition  starts  the construction  at the scale $m$ to obtain the interlacing functions $\lambda_{s,\mp} (\theta)$ ($m+1\leq s\leq n+1$) scale by scale.}
	\begin{prop}\label{mul}
		For $m+1\leq s\leq n+1$ and  $\theta\in I^{(1)}_{\cup}$, the operators $H_{s,\pm} (\theta)$ have unique eigenvalues $\lambda_{s,\mp} (\theta)$   interlacing $E_{s,>} (\theta)$ and $E_{s,<} (\theta)$: 
		\begin{equation}\label{intern+}
			E_{s,<} (\theta)\leq \lambda_{s,\mp} (\theta) \leq E_{s,>} (\theta),
		\end{equation}
		and the eigenfunctions $\phi_{s,\mp}$ corresponding to $\lambda_{s,\mp}$ decay exponentially:
		\begin{equation}\label{interexp}
			|\phi_{s,\mp} (x)|\leq e^{-\frac{\gamma_0}{5}\|x\|_1}, \  
			\|x\|_1\geq l_s^{\frac{7}{8}}. 
		\end{equation}
		Moreover,  we have
		\begin{equation}\label{Epln+}
			|\lambda_{s,\pm}' (\theta)-E_{m,\pm}' (\theta)|\lesssim e^{-\frac{1}{5}\gamma_0l_m}\delta_{m-1}^{-1}\ll \delta_m^{10}.
		\end{equation}
		In particular,
		\begin{equation}\label{tranln+}
			\pm \lambda_{s,\pm}' (\theta)\gtrsim  \delta_{m-2}^{10}\gg \delta_{m-1}.
		\end{equation}
	\end{prop}
	\begin{proof}
		We prove \eqref{intern+}, \eqref{interexp} by induction. The case $s=m+1$ was proved in Lemma \ref{it} replacing $n$ by $m$. Assuming  that  $\lambda_{s,\mp}$ have been  constructed, we are going to construct $\lambda_{s+1,\mp}$.  First we claim that  $H_{s+1,\pm} (\theta)$ have  unique eigenvalues $\lambda_{s+1,\mp} (\theta)$ in  $[\lambda_{s,\mp} (\theta)-\delta_{s-1},\lambda_{s,\mp} (\theta)+\delta_{s-1}]$ satisfying 
		\begin{equation}\label{apl}
			|\lambda_{s+1,\mp} (\theta)-\lambda_{s,\mp} (\theta)|\lesssim e^{-\frac15 \gamma_0l_s}.
		\end{equation} Let $\phi_{s,\mp}\in \operatorname{Image}Q_{s,\pm}$ be the eigenfunctions of $H_{s,\pm} (\theta)$ corresponding to $\lambda_{s,\mp} (\theta)$.
		By the exponential decay \eqref{interexp} for $\phi_{s,\mp}$, we have 
		\begin{align*}
			\| (H_{s+1,\pm} (\theta)-\lambda_{s,\mp} (\theta))\phi_{s,\mp}\|&=\|Q_{s+1,\pm} ( H_{B_{s+1}} (\theta)-\lambda_{s,\mp} (\theta))\phi_{s,\mp}\| \\ 
			&=\|Q_{s+1,\pm}\Gamma_{B_{s}}\phi_{s,\mp}\|\lesssim e^{-\frac15 \gamma_0l_s}.
		\end{align*}
		The above estimate together with Lemma \ref{trialcor} yields the existence of $\lambda_{s+1,\mp}$.
		
		Next we show the uniqueness and  they are the desired interlacing curves.  Let $\hat{\lambda}\in \sigma (H_{s+1,\pm})$ satisfy   $|\hat{\lambda}-\lambda_{s,\mp}|\leq \delta_{s-1}$  with  corresponding eigenfunction $\hat{\phi}\in \operatorname{Image}Q_{s+1,\pm}$. Then we have 
		$$Q_{s+1,\pm} ( H_{B_{s+1}}-\hat{\lambda})\hat{\phi}=0.$$
		Thus$$  ( H_{B_{s+1}}-\hat{\lambda})\hat{\phi}=a\psi_{m,\pm}.$$
		Employing  the same argument as the proof of item \textbf{(b)} of Proposition \ref{k2n+},   together with the exponential decay of $\psi_{m,\pm}$, one can prove that 
		$$|\hat{\phi} (x)|\leq e^{-\frac{\gamma_0}{5}\|x\|_1}+e^{-\frac{\gamma_0}{5}\|x-k_m\|_1}$$ 
		for $\operatorname{dist} (x,\{o,k_m\})\geq l_m^{7/8}$ and 	 $\hat{\phi}=\psi_{m,\mp}+O ( e^{-\frac15 \gamma_0l_m}\delta_{m-1}^{-1})$ since  $\hat{\phi}\in \operatorname{Image}Q_\pm$.
		Thus such an eigenvalue must be unique or else it will violate the  orthogonality. 
		Since   $E_{s,<} (\theta)\leq \lambda_{s,\mp} (\theta) \leq E_{s,>} (\theta)$ and $E_{s+1,<},E_{s+1,>}$ are generated from the interval $J_{i_s}^{(1)}$ with $|J_{i_s}^{(1)}|\ll\delta_{s-1}$,  $\lambda_{s+1,\mp} (\theta)$ must lie between $E_{s+1,<} (\theta)$ and  $ E_{s+1,>} (\theta)$ by Theorem \ref{cahuchy}. Thus we finish the proof of induction.
		
		By the Feynman-Hellman formula \textbf{(1)} in Lemma \ref{daoshu}, we have 
		\begin{align*}
			|\lambda_{s+1,\mp}'-E_{m,\mp}'|&=|\left\langle\phi_{s+1,\mp}, (H_{s+1,\pm})' \phi_{s+1,\mp}
			\right\rangle-\left\langle\psi_{m,\mp},H' \psi_{m,\mp}\right\rangle|\\
			&=|\left\langle\phi_{s+1,\mp},V' \phi_{s+1,\mp}
			\right\rangle-
			\left\langle\psi_{m,\mp},V' \psi_{m,\mp}\right\rangle -2\langle\phi_{s+1,\mp}, (P_{m,\pm})'  H_{B_{s+1}} \phi_{s+1,\mp}
			\rangle|\\
			&\leq |\left\langle\phi_{s+1,\mp},V' \phi_{s+1,\mp}
			\right\rangle-
			\left\langle\psi_{m,\mp},V' \psi_{m,\mp}\right\rangle| +2|\langle\phi_{s+1,\mp}, (P_{m,\pm})'  H_{B_{s+1}} \phi_{s+1,\mp}
			\rangle|\\
			&\lesssim e^{-\frac15 \gamma_0l_m}\delta_{m-1}^{-1} + |\langle\phi_{s+1,\mp},Q_{s+1,\pm} (P_{m,\pm})'  H_{B_{s+1}} Q_{s+1,\pm}\phi_{s+1,\mp}\rangle|,
		\end{align*}
		where we use the estimate $\|\phi_{s+1,\mp}-\psi_{m,\mp}\|\lesssim  e^{-\frac15 \gamma_0l_m}\delta_{m-1}^{-1}$ to bound the first term on the third  line of the inequality.
		
		Thus it suffices to estimate $Q_{s+1,\pm}P_{m,\pm}'  H_{B_{s+1}} Q_{s+1,\pm}$.  
		Recall that 
		\begin{align*}
			P_{m,\pm}'&=\langle\psi_{m,\pm} |\psi_{m,\pm}'\rangle+\langle\psi_{m,\pm}'|\psi_{m,\pm}\rangle \\
			&=-G_{B_{m,\pm}}^\perp (E_{m,\pm})V'P_{m,\pm}-P_{m,\pm}V'G_{B_{m,\pm}}^\perp (E_{m,\pm}).
		\end{align*}
		Since $	Q_{s+1,\pm}P_{m,\pm}=P_{m,\pm}	Q_{s+1,\pm}=0$,  it follows that 
		\begin{equation}\label{699}
			Q_{s+1,\pm}P_{m,\pm}'  H_{B_{s+1}} Q_{s+1,\pm}=-Q_{s+1,\pm}G_{B_{m,\pm}}^\perp (E_{m,\pm})V'P_{m,\pm}  H_{B_{s+1}} Q_{s+1,\pm}.
		\end{equation}
		We can employ the  estimates 
		\begin{align*}
			\|	P_{m,\pm}  H_{B_{s+1}} Q_{s+1,\pm}	\|&=\|P_{m,\pm} (  H_{B_{s+1}}-E_{m,\pm} )Q_{s+1,\pm}\|	\\
			&\leq \|P_{m,\pm} (  H_{B_{s+1}}-E_{m,\pm} )\|	\\
			&\leq \|\Gamma_{B_{m,\pm}}\psi_{m,\pm}\|\lesssim e^{-\frac15 \gamma_0l_m}
		\end{align*} 
		and $\|G_{B_{m,\pm}}^\perp (E_{m,\pm})\|=O (\delta_{m-1}^{-1})$ to \eqref{699} and complete the proof of \eqref{Epln+}. 
		
		Finally, \eqref{tranln+} follows from $\pm E_{m,\pm}'\gtrsim \delta_{m-2}^{10}$  and \eqref{Epln+} immediately.
	\end{proof}
	
	\begin{lem}\label{crossn+}
		With the notation from Lemma \ref{in},  there is a unique point $\bar{\theta}\in I^{(1)}_{\cup}$ satisfying $|\bar{\theta}-\theta_{n,>}|\lesssim\delta_n\delta_{n-2}^{-2}$, such that 
		$$\lambda_{n+1,+} (\bar{\theta})=\lambda_{n+1,-} (\bar{\theta}).$$	
		Moreover, the Rellich children have  quantitative separation away from $\bar{\theta}$: 
		\begin{equation}\label{E1sn+}
			E_{n+1,>} (\theta)-E_{n+1,<} (\theta)\geq |\lambda_{n+1,+} (\theta)-\lambda_{n+1,-} (\theta)|\geq \delta_{m-1}|\theta-\bar{\theta}|.
		\end{equation}
		As a corollary, if  $E_{n+1,>} (\theta)=E_{n+1,<} (\theta)$, then $\theta=\bar{\theta}$.
	\end{lem}
	\begin{proof}
		We consider the difference function 
		$$d (\theta):=\lambda_{n+1,+} (\theta)-\lambda_{n+1,-} (\theta).$$
		By  \eqref{tranln+}, we have
		\begin{equation}\label{trand}
			d' (\theta)\geq2\delta_{m-1}.
		\end{equation}
		Recalling  \eqref{91}, it follows that  
		\begin{align*}
			|\lambda_{n,+} (\theta_{n,>})-\lambda_{n,-} (\theta_{n,>})|&\leq 	E_{n,>} (\theta_{n,>})-E_{n,<} (\theta_{n,>})\\ &\leq 	E_{n,>} (\theta_{n,>})-E_{n,<} (\theta_{n,<}) +|E_{n,<} (\theta_{n,<})-E_{n,<} (\theta_{n,>})|\\
			&\leq 3\delta_n+O (\delta_n\delta_{n-2}^{-1}).
		\end{align*}
		By the above estimate and approximation \eqref{apl}, we obtain 
		$$|d (\theta_{n,>})|=|\lambda_{n+1,+} (\theta_{n,>})-\lambda_{n+1,-} (\theta_{n,>})|\lesssim \delta_n\delta_{n-2}^{-1}.$$
		It follows that by mean value theorem  there exists  a unique  point $\bar{\theta}$ satisfying $|\bar{\theta}-\theta_{n,>}|\lesssim \delta_n\delta_{n-2}^{-1}\delta_{m-1 }^{-1}\leq \delta_n\delta_{n-2}^{-2}$, such that 
		$d (\bar{\theta})=0$, that is 	$$\lambda_{n+1,+} (\bar{\theta})=\lambda_{n+1,-} (\bar{\theta}).$$	
		Then \eqref{E1sn+} follows from \eqref{intern+}, \eqref{trand} and 
		$$ 	|\lambda_{n+1,+} (\theta)-\lambda_{n+1,-} (\theta)|=|d (\theta)|=|d (\theta)-d (\bar{\theta})|\geq \delta_{m-1}|\theta-\bar{\theta}|$$
		by mean value theorem.
	\end{proof}
	
	Suppose that $|E_{m,+}' (\theta_{k_m,-})|\geq|E_{m,-}' (\theta_{k_m,-})|$  (the argument is similar to  the opposite case), and define $1\leq r\lesssim  \delta_{m-1}^{-1}$ such that 
	$$    |E_{m,+}' (\theta_{k_m,-})|=r|E_{n,-}' (\theta_{k_m,-})|.$$
	Recalling Lemma \ref{chan},  for $\theta\in I_\cup^{(1)}$, we have 
	\begin{equation}\label{chan+}
		| (E_{m,+}'+rE_{m,-}') (\theta)|\leq \delta_{m-1}^2.
	\end{equation}
	
	\begin{prop}\label{1215n+}
		For $\theta\in I_{\cup}^{(1)}$, we have the following:
		\begin{itemize}
			\item[\textbf{(a).}] If $E_{n+1,>} (\theta)\neq E_{n+1,<} (\theta)$, then 
			\begin{align}
				(E_{n+1,>})'& = (A^2-rB^2) E_{m,-}'+O (\delta_{m-1}^2 ), \label{den+}\\
				(E_{n+1,<})'& = (B^2-rA^2) E_{m,-}'+O (\delta_{m-1}^2 )\label{deen+}, 
			\end{align}
			and 	
			\begin{align}
				( E_{n+1,>})'' & =\frac{2\left\langle\psi_{n+1}, V' \Psi_{n+1}\right\rangle^2}{E_{n+1,>}-E_{n+1,<}}+O (\delta_{n-1}^{-1} ),\label{dfn+} \\
				( E_{n+1,<})''& =\frac{2\left\langle\psi_{n+1}, V' \Psi_{n+1}\right\rangle^2}{E_{n+1,<}-E_{n+1,>}}+O (\delta_{n-1}^{-1} )\label{ddfn+}
			\end{align}
			with the notation from Proposition \ref{k2n+} and $r$ as above.
			\item [\textbf{(b).}] Suppose  $E_{n+1,>} (\theta)\neq E_{n+1,<} (\theta)$. If $| (E_{n+1,>})' (\theta)|\leq\delta_{m-1}^2$, then $ (E_{n+1,>})'' (\theta)\geq2$. Similarly, if   $| (E_{n+1,<})' (\theta)|\leq\delta_{m-1}^2$, then $ (E_{n+1,<})'' (\theta)\leq-2$.
			\item[\textbf{(c).}] If there is a level crossing  (thus $E_{n+1,>} (\bar{\theta})=E_{n+1,<} (\bar{\theta})$ by Lemma \ref{crossn+}),  then $E_{n+1,>},E_{n+1,<}$ are piecewise differentiable  (except the point $\bar{\theta}$), and for $\theta>\bar{\theta}$, 
			$$ (E_{n+1,>})' (\theta)>\delta_{m-1}^2, \  (E_{n+1,<})' (\theta)<-\delta_{m-1}^2,$$
			for $\theta<\bar{\theta}$,
			$$ (E_{n+1,>})' (\theta)<-\delta_{m-1}^2, \  (E_{n+1,<})' (\theta)>\delta_{m-1}^2.$$
		\end{itemize}
	\end{prop}
	\begin{proof}
		We are going to prove \eqref{den+}.  (The proof of \eqref{deen+} is analogous.) 	 By \eqref{yn+} and \eqref{chan+}, we refer to Feynman-Hellman formula \textbf{(1)} Lemma \ref{daoshu} to obtain 
		\begin{align*}
			( E_{n+1,>})'
			& =\left\langle\psi_{n+1}, V' \psi_{n+1}\right\rangle =A^2 \left\langle\psi_{m,-}, V' \psi_{m,-}\right\rangle +B^2  \left\langle\psi_{m,+}, V' \psi_{m,+}\right\rangle+O (\delta_m^{10} ) \\
			&=A^2  E_{m,-}'+B^2  E_{m,+}'+O (\delta_m^{10} ) \\
			& = (A^2-rB^2) E_{m,-}'+B^2  (E_{m,+}'+rE_{m,-}')+O (\delta_m^{10} )\\
			& = (A^2-rB^2) E_{m,-}'+O (\delta_{m-1}^2 ).
		\end{align*}
		To prove \eqref{dfn+} and \eqref{ddfn+}, we  use Feynman-Hellman formula \textbf{(2)} and \textbf{(3)} to obtain 	
		\begin{equation*}
			(E_{n+1,>})''=\left\langle\psi_{n+1}, V'' \psi_{n+1}\right\rangle+2 \frac{\left\langle\psi_{n+1}, V' \Psi_{n+1}\right\rangle^2}{E_{n+1,>}-E_{n+1,<}}-2\left\langle V' \psi_{n+1},G^{\perp \perp}_{n+1} (E_{n+1,>}) V' \psi_{n+1}\right\rangle .
		\end{equation*}
		The first term is bounded by $D$ and the  third  term is bounded by $$2\|G^{\perp \perp}_{n+1} (E_{n+1,>})\|\|V' \psi_{n+1}\|^2\lesssim \delta_{n-1}^{-1},$$ where we  use the estimate $\|G^{\perp \perp}_{n+1} (E_{n+1,>})\| \leq2\delta_{n-1}^{-1}$ from  item \textbf{(c)} of   Proposition \ref{k2n+}. Thus we finish the proof of item \textbf{(a)}.
		
		Now we are going to prove \textbf{(b)}. We will show the first term in \eqref{dfn+} is large if  $| (E_{n+1,>})' (\theta)|\leq\delta_{m-1}^2$. \\
		Assume $| (E_{n+1,>})' (\theta)|\leq\delta_{m-1}^2$,  by \eqref{den+}, we have 
		$$|A^2-rB^2| |E_{m,-}' (\theta)|\leq| ( E_{n+1,>})' (\theta)|+O (\delta_{m-1}^2)\lesssim  \delta_{m-1}^2.$$
		Since $$|E_{m,-}' (\theta)|\gtrsim \delta_{m-2}^{10} ,$$ it follows that 
		$$|A^2-rB^2|\lesssim \delta_{m-1}^2 \delta_{m-2}^{-10}<\frac{1}{100}$$
		Since $A^2+B^2=1$ and $r\geq1$, we obtain  
		$$ (1+r)B^2\geq A^2+B^2-|A^2-rB^2|\geq \frac{99}{100}$$ and 
		$$ (1+\frac{1}{r})A^2\geq A^2+B^2-\frac{1}{r}|A^2-rB^2|\geq \frac{99}{100}.$$
		Thus $$B^2\geq\frac{1}{4r},\ A^2\geq\frac{1}{4}.$$
		By \eqref{yn+}, \eqref{chan+}  and the previous lower bound for $A,B$,  we have 
		\begin{align*}
			|\langle\psi_{n+1}, V' \Psi_{n+1}\rangle|&=|AB (\langle\psi_{m,-}, V' \psi_{m,-}\rangle-\langle\psi_{m,+}, V' \psi_{m,+}\rangle)+O (\delta_m^{10} )| \\
			&=|AB   (E_{m,-}' -E_{m,+}')+O (\delta_m^{10})|\\
			&=|AB \left( (1+r) E_{m,-}'- (E_{m,+}'+r E_{m,-}')\right)+O (\delta_m^{10})|\\
			& \geq \frac{1+r }{4\sqrt{r}}  (|E_{m,-}'|-O (\delta_{m-1}^2 ))\\
			&\gtrsim \delta_{m-2}^{10}\gg \delta_{m-1}.
		\end{align*}
		The denominator  $0<E_{n+1,>}-E_{n+1,<}\lesssim \delta_{n-1}^5$ since $|J_{i_n}^{(1)}|\lesssim\delta_{n-1}^{10}$ and $|I_\cup^{(1)}|\lesssim\delta_{n-1}^5$. Combining  the previous estimate of numerator, by \eqref{dfn+},  we obtain  $$ (E_{n+1,>})'' (\theta)\gtrsim     \delta_{m-1}^{2}  \delta_{n-1}^{-5}-O (\delta_{n-1}^{-1})>2,$$ 
		and 
		$$- (E_{n+1,<})'' (\theta)\gtrsim     \delta_{m-1}^{2}  \delta_{n-1}^{-5}-O (\delta_{n-1}^{-1})>2.$$ 
		Supposing there is a level crossing,  then  $$E_{n+1,>} (\bar{\theta})=\lambda_{n+1,+} (\bar{\theta}) =\lambda_{n+1,-} (\bar{\theta})= E_{n+1,<} (\bar{\theta}).$$
		Thus for  $\theta>\bar{\theta}$, 
		$$E_{n+1,>} (\theta)\geq\lambda_{n+1,+} (\theta)>\lambda_{n+1,-} (\theta) \geq E_{n+1,<} (\theta).$$
		The above inequality together with \eqref{tranln+} yields  the left derivative estimates of $E_{n+1,>}$ and $E_{n+1,<}$, namely, 
		$$   (E_{n+1,>})'_{\operatorname{Left}} (\bar{\theta})\geq \lambda_{n+1,+}' (\bar{\theta}) >\delta_{m-1},\   (E_{n+1,<})'_{\operatorname{Left}} (\bar{\theta})\leq \lambda_{n+1,-}' (\bar{\theta})<-\delta_{m-1}.$$
		We next claim  the inequalities hold for all $\theta>\bar{\theta}$. If it is not true, then  
		$$\{\theta \in I_{\cup}^{(1)}:\  \theta>\bar{\theta}, \ \  (E_{n+1,>})' (\theta)\leq  \delta_{m-1}^2 \}\neq \emptyset.$$ 
		Let $\theta^*$ be its infimum. Since $E_{n+1,>}' (\theta)$ is continuous for $\theta\geq \bar{\theta}$ and $  (E_{n+1,>})'_{\operatorname{Left}} (\bar{\theta}) >\delta_{m-1}^2$, we have $\theta^*>\bar{\theta}$ and  $ (E_{n+1,>})' (\theta)>\delta_{m-1}^2\geq  (E_{n+1,>})' (\theta^*)$ for $\theta\in [\bar{\theta},\theta^*)$, which implies $  (E_{n+1,>})'' (\theta^*)\leq 0$. However, by item \textbf{(b)}, we get a contradiction that   $ (E_{n+1,>})'' (\theta^*)>2$. Thus we prove our claim. The case of $\theta<\bar{\theta}$ is analogous.
	\end{proof}
	\begin{rem}\label{xin}
		We can deduce  from the above proposition that  $E_{n+1,>}$ and $E_{n+1,<}$ have the two-monotonicity interval structure. Moreover, if the gap between the critical values of $E_{n+1,>}$ and $E_{n+1,<}$ is larger than $3\delta_{n+1}$, then each of them belongs to \textbf{Type} \ref{t2}, otherwise, $E_{n+1,>}$ along with  $E_{n+1,<}$ belong  to \textbf{Type} \ref{t3}.  We finish the proof of \textbf{Subcase B}.
	\end{rem}
	\subsubsection{A sketch of the domain adjustment}
	At the end of the construction, we employ the argument in Section \ref{adjm} to adjust  the domains of Rellich children  by $O (\delta_{n}^{10})$ so that every Rellich child $E_{n+1}$ has the same image on each of its monotone  interval. After the domain adjustment,    Hypothesis \ref{h5} still holds for $m=n+1$ because  the overlap of the adjacent codomains   ($\sim3\delta_n$) is  much larger  than the codomain contraction in the adjustment.  We henceforth denote by $I (E_{n+1})$ (resp. $J (E_{n+1})$) the  modified domain (codomain) of $E_{n+1}$.
	
	\subsubsection{Green's function estimates} In this part, we establish the Green's function estimates for $(n+1)$-good sets by a multi-scale analysis argument and finish the proof of  induction.
	
	Recall that $\tilde{J} (E_{n+1})$ is defined as:  
	For  $E_{n+1}$ belonging to   \textbf{Type} \ref{t1} or \ref{t2}, \\
	\textbf{(1).} if $E_{n+1}$ has no critical point in $ I (E_{n+1})$, then 
	$$\tilde{J} (E_{n+1}):=[\inf J (E_{n+1})+\frac{9}{8}\delta_{n},\sup J (E_{n+1})-\frac{9}{8}\delta_{n}], $$ 
	\textbf{(2).} if $E_{n+1}$  has  its minimum at a  critical point in $ I (E_{n+1})$, then 
	$$\tilde{J} (E_{n+1}):=[\inf J (E_{n+1})-2\delta_{n+1},\sup J (E_{n+1})-\frac{9}{8}\delta_{n}] ,$$
	\textbf{(3).}  if $E_{n+1}$ has its maximum at a  critical point in $ I (E_{n+1})$, then 
	$$\tilde{J} (E_{n+1}):=[\inf J (E_{n+1})+\frac{9}{8}\delta_{n}, \sup J (E_{n+1})+2\delta_{n+1}], $$
	and for  $E_{n+1}$ belonging to  \textbf{Type} \ref{t3},
	$$ \tilde{J} (E_{n+1}):=[\inf J (E_{n+1,<})+\frac{9}{8}\delta_{n},\sup J (E_{n+1,>})-\frac{9}{8}\delta_{n}].$$
	
	Fix $\theta^*,E^*$. Let  $E_{n+1}\in \mathcal{C}_{n+1}$ be the  Rellich function such that $E^*\in \tilde{J} (E_{n+1})$ (if exists). The set of $(n+1)$-resonant points  (relative to $\theta^*,E^*$ and $E_{n+1}$) is defined as: For $E_{n+1}$ belonging to   \textbf{Type} \ref{t1} or \ref{t2}, 
	$$	S_{n+1} (\theta^*,E^*):=	\{x\in \Z^d:\ \theta^*+x\cdot\omega \in I (E_{n+1}), |E_{n+1} (\theta^*+x\cdot \omega)-E^*|<\delta_{n+1}\},$$
	and for $E_{n+1}$ belonging to  \textbf{Type} \ref{t3},
	\begin{align*}
		&\ \ \  S_{n+1} (\theta^*,E^*)\\
		&:=\{x\in \Z^d:\ \theta^*+x\cdot\omega \in  I (E_{n+1,<}) \cup I (E_{n+1,>}), \min_{\bullet\in\{>,<\}} |E_{n+1,\bullet} (\theta^*+x\cdot \omega)-E^*|<\delta_{n+1}\}.
	\end{align*}
	We say that a set $\Lambda\subset \Z^d$ is $(n+1)$-nonresonant  (relative to $ (\theta^*,E^*)$ and $E_{n+1}$) if $\Lambda\cap S_{n+1} (\theta^*,E^*) =\emptyset$ and is  $(n+1)$-regular if  $ (\Lambda_{2l_{i+1}}+x)\subset\Lambda$ for any $x\in S_{i} (\theta^*,E^*)\cap\Lambda$ relative to  each ancestor $E_i  $ $  (0\leq i\leq n)$ of $E_{n+1}$. We say that a set is $(n+1)$-good if it is both $(n+1)$-nonresonant and $(n+1)$-regular. We are going to prove  Green's function estimates for $(n+1)$-good sets:
	\begin{thm}\label{rep}
		Fix $\theta^*, E^*$ and a finite set $\Lambda\subset \Z^d$. Let  $E_{n+1}\in \mathcal{C}_{n+1}$  satisfy  $E^*\in \tilde{J} (E_{n+1})$ (if exists).    If $\Lambda$ is $(n+1)$-good, then  for $|\theta-\theta^*|<\delta_{n+1}/ (10D), |E-E^*|<\delta_{n+1}/5$, 
		\begin{align*}
			\|G_\Lambda (\theta,E)\|&\leq10\delta_{n+1}^{-1},\\
			|G_\Lambda (\theta,E;x,y)|&\leq e^{-\gamma_{n+1}\|x-y\|_1}, \  \|x-y\|_1\geq l_{n+1}^{\frac{5}{6}}, 
		\end{align*}
		where $\gamma_{n+1}= (1-O (l_{n+1}^{-\frac{1}{30}}))\gamma_{n}\geq \frac{1}{2}\gamma_0$.
		The above estimates also hold if $\Lambda$ is $(n+1)$-regular and $E^*\notin  \tilde{J} (E_{n+1})$ for any  $E_{n+1}\in \mathcal{C}_{n+1}$.
	\end{thm}
	
	\begin{proof}
		If 	$E^*\notin  \tilde{J} (E_{n})$ for any  $E_{n}\in \mathcal{C}_{n}$, then the theorem follows from the induction hypothesis.  By the overlap of two   adjacent modified codomains  $\tilde{J}$ where the Rellich functions do not attain their  critical values (c.f. Hypothesis \ref{h5}), it follows that $\cup_{E_n\in \mathcal{C}_n}\tilde{J} (E_{n})= \cup_{E_n\in \mathcal{C}_n}\tilde{\tilde{J}} (E_{n})$. We may thus assume $E^*\in\tilde{\tilde{J}} (E_{n})$ for some $E_{n}\in \mathcal{C}_{n}$.
		
		First,  we assume that $E^*\in\tilde{J} (E_{n+1})\subset\tilde{J} (E_{n})$ for some $E_{n+1}\in \mathcal{C}^{(j)} (E_{n})$.	We define the set  $S_n^{\Lambda}:= S_n (\theta^*,E^*)\cap \Lambda$. If $S_n^{\Lambda}=\emptyset$, then $\Lambda$ is $n$-nonresonant and the theorem follows from  the induction  hypothesis. We thus assume $S_n^{\Lambda}$ is nonempty. We wish to  find a finite family $\{B_{n+1}^{(j)}+p\}_{p\in \bar{S}_n^{\Lambda} }$ of translations  of $B_{n+1}^{(j)}$ such that each $x\in S_n^{\Lambda}$ is near the center of $B_{n+1}^{(j)}+p$ for some $p\in \bar{S}_n^{\Lambda}$ and $\|G_{B_{n+1}^{(j)}+p} (\theta^*,E^*)\|\leq \delta_{n+1}^{-1}$
		for all $p\in \bar{S}_n^{\Lambda}$. We do so by case analysis:\\
		\textbf{(1).} $j=1$:	We  claim  that  for any  $x\in S_n^{\Lambda}$, if $E_{n+1}$ is generated from the interval in  \textbf{Subcase A}, then $\theta^*+x\cdot \omega \in  I (E_{n+1})$, and  if $E_{n+1}$ is generated from the interval  \eqref{SBI} from \textbf{Subcase B}, then    $\theta^*+x\cdot \omega \in  I (E_{n+1,<})\cup I (E_{n+1,>})$,  where $E_{n+1,<}$ and  $E_{n+1,>}$ are the two Rellich children generated from  the interval \eqref{SBI}. The proof proceeds also by case analysis:\\
		\textbf{Subcase A (a).} $E_{n+1}$ has no critical point.\\
		Since $E^*\in\tilde{J} (E_{n+1})$ and $x\in S_n^{\Lambda}$, by the definition of $ \tilde{J} (E_{n+1})$, we obtain 
		$$E_n (\theta^*+x\cdot \omega )\in [\inf J (E_{n+1})+\frac{1}{8}\delta_n,\sup J (E_{n+1})-\frac{1}{8}\delta_n].$$
		Since $j=1$, by approximation \eqref{C0n}, we have 
		$$[\inf J (E_{n+1})+\frac{1}{8}\delta_n,\sup J (E_{n+1})-\frac{1}{8}\delta_n]\subset  (E_n)_\pm (I (E_{n+1})_\pm ).$$ 
		Thus $E_n (\theta^*+x\cdot \omega )\in  (E_n)_\pm (I (E_{n+1})_\pm )$ and so $\theta^*+x\cdot \omega  \in  I (E_{n+1}) $.\\
		\textbf{Subcase A (b).}  $E_{n+1}$ has a  critical point.\\
		Suppose that this critical point is a minimum  (the other case is analogous). Then $E_n$ likewise achieves its minimum at a critical point in $ I (E_{n+1})$. As above,  we have 
		$$E_n (\theta^*+x\cdot \omega ) \in  [\inf J (E_n),\sup J (E_{n+1})-\frac{1}{8}\delta_n]$$
		and 
		$$ [\inf J (E_n),\sup J (E_{n+1})-\frac{1}{8}\delta_n]   \subset  (E_n)_\pm (I (E_{n+1})_\pm).$$
		Thus  $\theta^*+x\cdot \omega\in I (E_{n+1})$.   \\
		\textbf{Subcase B.}  $E_{n}$ belongs to \textbf{Type} \ref{t3} and $E_{n+1,<}$ and  $E_{n+1,>}$ are the two Rellich children generated from  the interval \eqref{SBI}. \\
		Since $E^*\in\tilde{J} (E_{n+1})$ and $x\in S_n^{\Lambda}$, by the definition of $ \tilde{J} (E_{n+1})$, we obtain 
		\begin{equation}\label{hold}
			E_{n,\bullet} (\theta^*+x\cdot \omega )\in [\inf J (E_{n+1,<})+\frac{1}{8}\delta_n,\sup J (E_{n+1,>})-\frac{1}{8}\delta_n]
		\end{equation}
		for some $\bullet\in\{>,<\}$. If \eqref{hold}  holds for $\bullet=<$, by the approximation \eqref{ex}, it follows that 
		$$E_{n,<} (\theta^*+x\cdot \omega )\geq \inf J (E_{n+1,<})+\frac{1}{8}\delta_n\geq \inf  ( E_{n,<})_\pm ( I (E_{n+1,<})), $$ %$$\sup J (E_{n+1,>})-\frac{1}{8}\delta_n\leq \sup  ( E_{n,>})_\pm ( I (E_{n+1,>})) .$$ 
		and so  $\theta^*+x\cdot \omega\in  I (E_{n+1,<}).$ 
		Likewise, if \eqref{hold}  holds for $\bullet=>$, then $\theta^*+x\cdot \omega\in  I (E_{n+1,>}).$ \\
		Thus we finish the proof of the claim.\\
		For $j=1$,  we define $\bar{S}_n^{\Lambda}=S_n^{\Lambda}$. \\
		%since $|E_n (\theta^*+x\cdot\omega)-E^*|\leq \delta_n$ and $\Lambda$ is $(n+1)$-nonresonant,\\
		For \textbf{Subcase A}, since $x\in S_n^{\Lambda}$,   $$|E_n (\theta^*+x\cdot\omega)-E^*|\leq \delta_n,$$
		and since $\Lambda$ is $(n+1)$-nonresonant,
		$$ \delta_{n+1}\leq |E_{n+1} (\theta^*+x\cdot\omega)-E^*|.$$
		For \textbf{Subcase B}, since $x\in S_n^{\Lambda}$,
		$$ \min_{\bullet\in\{>,<\}} |E_{n,\bullet} (\theta^*+x\cdot \omega)-E^*|\leq\delta_{n}, $$	and since $\Lambda$ is $(n+1)$-nonresonant,
		$$\delta_{n+1}\leq  \min_{\bullet\in\{>,<\}} |E_{n+1,\bullet} (\theta^*+x\cdot \omega)-E^*|.$$
		By the eigenvalue separation estimates from Propositions  \ref{kn}, \ref{k2n+}, it follows that $|\hat{E}-E^*|\geq \delta_{n+1}$ for any eigenvalue $\hat{E}$ of $H_{B_{n+1}^{(1)}+x} (\theta^*)$, and so 
		$$\|G_{B_{n+1}^{(1)}+x} (\theta^*,E^*)\|\leq \delta_{n+1}^{-1}.$$\\
		\textbf{(2).} $j=2$: Recalling that in double resonance case, $E_{n+1}$ has a brother. We  denote  by $E_{n+1,>}$  the larger one  and  $E_{n+1,<}$ the smaller one.   We   prove  that for any  $x\in S_n^{\Lambda}$, $\theta^*+p\cdot \omega \in  I (E_{n+1,<})\cup I (E_{n+1,>})$ for some $p\in \{x,x-k_n\}$, where the definition of $k_n$ is from \eqref{huiyi}. \\ 
		Since  $E^*\in\tilde{J} (E_{n+1})$ and $x\in S_n^{\Lambda}$, by the definition of $\tilde{J} (E_{n+1})$, we obtain 
		$$E_n (\theta^*+x\cdot \omega )\in [\inf J (E_{n+1,<})+\frac{1}{8}\delta_n,\sup J (E_{n+1,>})-\frac{1}{8}\delta_n].$$
		Since $j=2$, by approximation \eqref{apbpn}, we have 
		$$ [\inf J (E_{n+1,<})+\frac{1}{8}\delta_n,\sup J (E_{n+1,>})-\frac{1}{8}\delta_n]    \subset  (E_n)_-\left ([\inf I (E_{n+1,>}), \sup I (E_{n+1,<})]\right) $$
		and 
		$$[\inf J (E_{n+1,<})+\frac{1}{8}\delta_n,\sup J (E_{n+1,>})-\frac{1}{8}\delta_n]    \subset  (E_n)_+\left ([\inf I (E_{n+1,<}), \sup I (E_{n+1,>})]+k_n\cdot \omega \right).$$ 
		Thus $$\theta^*+x\cdot \omega\in [\inf I (E_{n+1,>}), \sup I (E_{n+1,<})]\cup [\inf I (E_{n+1,<}), \sup I (E_{n+1,>})]+k_n\cdot \omega.$$ It follows that $\theta^*+x\cdot \omega \in 
		[\inf I (E_{n+1,>}), \sup I (E_{n+1,<})]$ or $\theta^*+ (x-k_n)\cdot \omega \in 
		[\inf I (E_{n+1,<}), \sup I (E_{n+1,>})]$. Thus 
		\begin{align*}
			\theta^*+p\cdot \omega&\in  	[\inf I (E_{n+1,>}), \sup I (E_{n+1,<})]\cup[\inf I (E_{n+1,<}), \sup I (E_{n+1,>})]\\
			&=I (E_{n+1,<})\cup I (E_{n+1,>})
		\end{align*} for some $p\in \{x,x-k_n\}$. For $j=2$, we define $\bar{S}_n^{\Lambda}$ to be the set of all points  $p$ coming from  $x\in S_n^{\Lambda}$ as above. 
		Since $\Lambda$ is $(n+1)$-regular and $\|k_n\|_1\leq 10l_{n+1}^{(1)}<2l_{n+1}^{(2)}$, it follows that $\bar{S}_n^{\Lambda}\subset \Lambda$. Since $\Lambda$ is $(n+1)$-nonresonant,   by the eigenvalue separation estimates from Proposition \ref{k2n}, it follows that  $|\hat{E}-E^*|\geq \delta_{n+1}$ for any eigenvalue $\hat{E}$ of $H_{B_{n+1}^{(2)}+p} (\theta^*)$, and so 
		$$\|G_{B_{n+1}^{(2)}+p} (\theta^*,E^*)\|\leq \delta_{n+1}^{-1}.$$
		Now we assume  that  $E^*\notin\tilde{J} (E_{n+1})$ for all $E_{n+1}\in \mathcal{C} (E_{n})$. Recall that  $\tilde{\tilde{J}} (E_n)$ is almost  covered by $\cup_{E_{n+1}\in \mathcal{C} (E_{n})}\tilde{J} (E_{n+1})$ except the regions near a critical value of some $E_{n+1}\in \mathcal{C} (E_{n})$ (c.f. Hypothesis \ref{h5} for $m=n+1$). Thus the $$\theta\in\bigcup_{E_{n+1}\in \mathcal{C} (E_n) }I (E_{n+1})$$ minimizing  $$ \operatorname{dist} (E^*,E_{n+1} (\theta)),\ E_{n+1}\in \mathcal{C} (E_{n})$$ must attain at a critical point  of some $E_{n+1}\in \mathcal{C}^{(j)} (E_{n})$. By the same argument as above, it follows that for this $E_{n+1}$, 
		$\theta^*+p\cdot \omega \in  I (E_{n+1})$ for some $p\in \{x,x-k_n\}$.  By the definition of $\tilde{J} (E_{n+1})$, 
		we have  $$\delta_{n+1}\leq |E_{n+1} (\theta^*+p\cdot \omega)-E^*|. $$ Thus  we also have  $$\|G_{B_{n+1}^{(j)}+p} (\theta^*,E^*)\|\leq \delta_{n+1}^{-1}.$$ 
		For both $j=1,2$, by the $(n+1)$-regularity of $\Lambda$,  $(B_{n+1}+p)\subset \Lambda$ for all $p\in \bar{S}_n^{\Lambda}$. Moreover, by Diophantine condition, the separation of $p\in \bar{S}_n^{\Lambda}$ is larger than the size of $B_{n+1}$, thus the blocks $\{B_{n+1}+p\}_{p\in \bar{S}_n^{\Lambda} }$ are non-overlapping.

		Next,  we prove the case  that  $\Lambda=B_{n+1}+p$ is a  single $(n+1)$-nonresonant block. By the previous discussion, we have  
		\begin{equation*}
			\|G_{\Lambda} (\theta^*,E^*)\|\leq\delta_{n+1}^{-1}.
		\end{equation*}
		It follows that  by Neumann series argument, for $|\theta-\theta^*|<\delta_{n+1}/ (10D)$ and  $|E-E^*|<\frac{2}{5}\delta_{n+1}$, 
		\begin{equation*}\label{L2}
			\|G_{\Lambda } (\theta,E)\|\leq2\delta_{n+1}^{-1}.
		\end{equation*}
		Let $x,y\in \Lambda $ satisfy  $\|x-y\|_1\geq l_{n+1}^\frac{4}{5}$. Since $G_{\Lambda }$ is self-adjoint, we may  assume $\|x-p\|_1\geq l_{n+1}^\frac{3}{4}$. Let $\Omega_{n+1,p}$ be an $l_{n+1}^\frac{2}{3}$-size block centered at $p$ such that  $\Lambda \setminus \Omega_{n+1,p}$ is $n$-good. Recalling  the estimates for $n$-good sets from Hypothesis \ref{h6}, by   resolvent identity, we obtain 
		\begin{align*}
			|G_{\Lambda } (\theta,E;x,y)|&=|G_{\Lambda \setminus \Omega_{n+1,p}} (x,y)\chi (y)+\sum_{z,z'}G_{\Lambda \setminus \Omega_{n+1,p}} (x,z)\Gamma_{z,z'}G_{\Lambda } (z',y)|\\
			&\leq e^{-\gamma_n\|x-y\|_1}+C (d)\sup_{z,z'}e^{-\gamma_n\|x-z\|_1}|G_{\Lambda } (z',y)|\\
			&\leq e^{-\gamma_n\|x-y\|_1}+C (d)\sup_{z,z'}e^{-\gamma_n\|x-z\|_1}e^{-\gamma_n (\|z'-y\|_1-l_{n+1}^\frac{3}{4})}\delta_{n+1}^{-1}\\
			&\leq e^{-\gamma'_n\|x-y\|_1}
		\end{align*}
		with  $\gamma'_n= (1-O (l_{n+1}^{-\frac{1}{30}}))\gamma_n$, where we have used if $\|z'-y\|_1\leq l_{n+1}^\frac{3}{4}$, then 
		\begin{equation*}
			|G_{\Lambda} (z',y)|\leq 	\|G_{\Lambda}\|\leq2\delta_{n+1}^{-1}\leq 2e^{-\gamma_n (\|z'-y\|_1-l_{n+1}^\frac{3}{4})}\delta_{n+1}^{-1}, 
		\end{equation*}
		if $\|z'-y\|_1\geq l_{n+1}^\frac{3}{4}$, then 
		\begin{align*}
			|G_{\Lambda} (z',y)|=	|G_{\Lambda} (y,z')|&\leq 	\sum_{w,w'}|G_{\Lambda \setminus \Omega_{n+1,p}} (y,w)\Gamma_{w,w'}G_{\Lambda} (w',z')|\\
			&\leq C (d)e^{-\gamma_n\|y-w\|_1}\|G_{\Lambda}\| \\
			&\leq C (d)	e^{-\gamma_n (\|y-z'\|_1-l_{n+1}^\frac{3}{4})}\delta_{n+1}^{-1}
		\end{align*} and  $\delta_{n+1}^{-1}=e^{l_{n+1}^\frac{2}{3}}\ll e^{\gamma_n\|x-y\|_1}$ to bound the second term. Thus  we finish  the case that  $\Lambda$ is a single $(n+1)$-nonresonant block.\\  
		Now,  assume $\Lambda$ is an arbitrary $(n+1)$-good set. We must show that $G_\Lambda$ dose exist. By Schur's test, it suffices to prove 
		\begin{equation}\label{Schurn}
			\sup_x\sum_{y}|G_\Lambda (\theta,E+io;x,y)|<C<\infty.
		\end{equation}
		Define  $$\Lambda':=\Lambda\setminus\bigcup_{p\in \bar{S}_n^{\Lambda}} \Omega_{n+1,p}.$$ Then $\Lambda'$ is $n$-nonresonant  since $S_n^{\Lambda}\subset\cup_{p\in \bar{S}_n^{\Lambda}} \Omega_{n+1,p}$ and is $n$-regular since $\Lambda$ is $(n+1)$-regular and each $(B_{n+1}+p)\setminus\Omega_{n+1,p}$ is $n$-regular.  Let $\tilde{\Omega}_{n+1,p}$ be  a  $2l_{n+1}^\frac{2}{3}$-size cube centered at $p$.  For $x\in \Lambda\setminus\cup_{p\in \bar{S}_n^{\Lambda}} \tilde{\Omega}_{n+1,p}$, by resolvent identity, we have
		\begin{align*}
			\sum_y|G_\Lambda (x,y)|&\leq \sum_y|G_{\Lambda'} (x,y)|+\sum_{z,z',y}|G_{\Lambda'} (x,z)\Gamma_{z,z'}G_{\Lambda} (z',y)|\\
			&\leq	C (d)\delta_n^{-2}+	C (d)e^{-l_{n+1}^\frac{2}{3}}\sup_{z'}\sum_y|G_{\Lambda} (z',y)|.
		\end{align*}
		For $x\in \tilde{\Omega}_{n+1,p}$, by resolvent identity, we have
		\begin{align*}
			\sum_y|G_\Lambda (x,y)|&\leq \sum_y|G_{B_{n+1}+p} (x,y)|+\sum_{z,z',y}|G_{B_{n+1}+p} (x,z)\Gamma_{z,z'}G_{\Lambda} (z',y)|\\
			&\leq \delta_{n+1}^{-2}+C (d)e^{-\frac{1}{2}l_{n+1}}\sup_{z'}\sum_y|G_{\Lambda} (z',y)|.
		\end{align*}
		Taking supremum for $x$ on the left hand side of the above two inequalities, we get $$\sup_x\sum_y|G_\Lambda (x,y)|\leq \delta_{n+1}^{-2}+\frac{1}{2}\sup_x\sum_y|G_\Lambda (x,y)|,$$ thus 
		$$\sup_x\sum_y|G_\Lambda (x,y)|\leq 2\delta_{n+1}^{-2},$$
		which gives \eqref{Schurn}.
		Since for $|\theta-\theta^*|<\delta_{n+1}/ (10D)$ and $|E-E^*|<\frac{2}{5}\delta_{n+1}$, $G_\Lambda (\theta,E)$ does exist,  we obtain   $\operatorname{dist} (\sigma (H_\Lambda (\theta)),E^*)\geq \frac{2}{5}\delta_{n+1}$ and hence $\operatorname{dist} (\sigma (H_\Lambda (\theta)),E)\geq \frac{1}{5}\delta_{n+1}$ for all $|E-E^*|<\frac{1}{5}\delta_{n+1}$. Thus we get the desired  operator norm 
		$$\|G_\Lambda (\theta,E)\|=\|\left ( H_\Lambda (\theta)-E\right) ^{-1}\|=\frac{1}{\operatorname{dist} (\sigma (H_\Lambda (\theta)),E)}\leq10\delta_{n+1}^{-1}.$$
		It remains to  prove the off-diagonal decay of $G_\Lambda$.
		Let $x,y\in \Lambda$ satisfy $\|x-y\|_1\geq l_{n+1}^\frac{5}{6}$. We define 
		\[B_x:=\left\{\begin{aligned}
			&\Lambda_{\l_{1}^\frac{1}{2}} (x)\cap\Lambda  \quad \text{if }   x\in \Lambda\setminus\bigcup_{p\in \bar{S}_0^{\Lambda}} \tilde{\Omega}_{1,p}, \\
			&B_m+p_m\ \quad \text{if }  x \in\tilde{\Omega}_{m,p_m},
		\end{aligned}\right. \]
		where $p_m\in \bar{S}_{m-1}^{\Lambda}$ and $m$ satisfying    $1\leq m\leq n+1$ is the smallest integer  such that  $p_m\notin S_m (\theta^*,E^*)$.
		The set $B_x$ has the following two properties: \textbf{(1).} $B_x$ is $m$-good for some $0\leq m\leq n+1$; \textbf{(2).} $x$ is closed  to the center of $B_x$ and away from its relative boundary $\partial_\Lambda B_x$. We can iterate the resolvent identity  to obtain 
		\begin{align}
			|G_\Lambda (x,y)|&\leq\prod_{s=0}^{L-1}  (C (d) l_{m_s}^d e^{-\gamma_{m_s-1}'\|x_{s}-x_{s+1}\|_1})|G_\Lambda (x_L,y)|\nonumber \\
			&\leq e^{-\gamma_n''\|x-x_L\|_1}|G_\Lambda (x_L,y)|, \label{728n}
		\end{align}
		where $x_0:=x$, $x_{s+1}\in \partial B_{x_s}$  with  $\|x_{s}-x_{s+1}\|_1\geq \frac{1}{2}l_{m_s}\gg l_{m_s}^{\frac{5}{6}}$. We stop the iteration until  $y\in B_{x_L}$. 
		By resolvent identity again, we get 
		\begin{align}
			|G_\Lambda (x_L,y)|&\leq|G_{B_{x_L}} (x_L,y)|+\sum_{z,z'}|G_{B_{x_L}} (x_L,z)\Gamma_{z,z'}G_{\Lambda} (z',y)|\nonumber\\
			&\leq C (d) e^{-\gamma_n' (\|x_L-y\|_1-l_{n+1}^\frac{4}{5})}\delta_{n+1}^{-1},\label{728.n}
		\end{align}
		where we have used the off-diagonal exponential decay of $G_{B_{x_L}}$ and   $\|G_{\Lambda}\|\leq10\delta_{n+1}^{-1}$. Since $\|x-y\|_1\geq l_{n+1}^\frac{5}{6}$, \eqref{728n} and \eqref{728.n} give the desired off-diagonal estimate
		$$|G_\Lambda (x,y)|\leq e^{-\gamma_{n+1}\|x-y\|_1}$$
		with $\gamma_{n+1}= (1-O (l_{n+1}^{-\frac{1}{30}}))\gamma_n$.  
		
		Thus we finish the proof.
	\end{proof}
	
	\section{Proof of main Theorem}\label{mainp}
	In this section, we will relate the Rellich functions we have constructed in the previous section to spectral and dynamical properties  of  $H$ and give the proof of Theorem \ref{main}.
	\subsection{Anderson localization}
	
	We first construct the set where Anderson localization can occur.
	
	For each function $E_n\in \mathcal{C}_n\  (n\geq 1)$,  noticing that  $\delta_{n-2}^{10}\ll (l_{n+2})^{-2\tau}$, one can  employ the argument  from the proof of  Lemma \ref{sepn} to  show that there  exists at most one $\tilde{k}_n$ with $0\neq \|{\tilde{k}_n}\|_1\leq 10l_{n+2}$, such that there is some $\theta_{{\tilde{k}_n},-}\in  I (E_n)_-$ satisfying 
	$$\theta_{{\tilde{k}_n},-}+{\tilde{k}_n}\cdot \omega  \in  I (E_n)_+, \ e_{\tilde{k}_n}:= E_n (\theta_{{\tilde{k}_n},-})=E_n (\theta_{{\tilde{k}_n},-}+{\tilde{k}_n}\cdot \omega). $$ 
	Moreover, if such $\tilde{k}_n$ exists, then $E_n$ belongs to \textbf{Type} \ref{t1}. 
	If the  above $\tilde{k}_n$ exists, we define 
	\begin{equation}\label{ALDL}
		J^{DR} (E_n):=B_{\delta_{n-1}^{10}} (e_{\tilde{k}_n}).
	\end{equation}
	Now define 
	\begin{equation}\label{bad}
		K_n:=\bigcup_{E_n\in\mathcal{C}_n}\bigcup_{\|k\|_1\leq l_{n+1}} (E_n^{-1} (J^{DR} (E_n))+k\cdot \omega )
	\end{equation}
	and 
	\begin{equation}\label{ALset}
		K:=\bigcap_{s\geq 1}\bigcup_{n\geq s }K_n.   
	\end{equation}
	
	Denote by $|\cdot|$ the Lebesgue measure of a set. We have
	\begin{lem}\label{ALLem}
		$|K|=0.$	
	\end{lem}
	\begin{proof}
		By construction (c.f. \textbf{(5)} of Proposition \ref{covern}), for any $E_s\in \mathcal{C}_s$, $|\mathcal{C} (E_s)|\leq 2\delta_{s-1}^{-10}$. Thus for $n\geq 1$, 
		\begin{equation}\label{gu1}
			|\mathcal{C}_n|\leq \prod_{s=0}^{n-1} (2\delta_{s-1}^{-10})\leq \delta_{n-2}^{-11}.
		\end{equation}
		Fix $E_n\in \mathcal{C}_n$.  Since each  $|J^{DR} (E_n)|=2\delta_{n-1}^{10}$  and  $E_n$ satisfies the Morse condition with the  two-monotonicity interval structure, we can employ Lemma \ref{C2} to obtain 
		\begin{equation}\label{gu2}
			|E_n^{-1} (J^{DR} (E_n))|\leq 2\sqrt{2|J^{DR} (E_n)|}\leq 4\delta_{n-1}^5.
		\end{equation} Combining the estimates \eqref{gu1} and \eqref{gu2}, we get 
		$$|K_n|\lesssim  \delta_{n-2}^{-11}  (l_{n+1})^d \delta_{n-1}^5\leq \delta_{n-1}$$
		and then $$\sum_{n=1}^{+\infty }|K_n| <+\infty.$$ Thus by Borel-Cantelli Lemma, $|K|=0$.
	\end{proof}

	Thus $\Theta:=\mathbb{T}\setminus K$ is a set of full measure. For the remainder of this section, we fix $\theta \in \Theta$. Thus there exists some $N_1$ such that $\theta\notin K_n$ for all $n\geq N_1$.

	To prove $H (\theta)$ has Anderson localization using Schnol's lemma,     it suffices  to show that every generalized eigenfunction $\psi$ of $H (\theta)$ that grows at most polynomially    (i.e., $ |\psi (x)|\leq  ( 1+\|x\|_1)^d$) in fact decays exponentially. For the remainder of this section, we fix a generalized eigenvalue $E$ and its generalized eigenfunction $\psi$. 
	
	\begin{lem}\label{ALL}
		There exists some $N_2$ such that for all $n \geq N_2$, $E\in \tilde{J} (E_n)$ for some $E_n\in \mathcal{C}_n$. Moreover, there exists an $x_n\in \Z^d$ satisfying $\|x_n\|_1\leq 100l_n$ such that $x_n\in S_n (\theta,E)$ relative to this  $E_n$. 
	\end{lem}
	\begin{proof}
		If not,   there exists an increasing integer sequence $\{n_i\}_{i=0}^{\infty}$ such that $\Lambda_{100l_{n_i}}$ is $n_i$-nonresonant and  $\{U_i\}_{i=0}^{\infty}$  such that $U_i$ is $n_i$-regular and satisfies   $$\Lambda_{50l_{n_i}}\subset 	U_i\subset \Lambda_{100l_{n_i}}. $$ Thus $U_i$ satisfies the condition of Hypothesis \ref{h6} at scale $n_i$. It follows that $G_{U_i} (\theta,E;x,y)$ decays exponentially for $\|x-y\|_1\geq l_{n_i}^\frac{5}{6}$.  For any $\|x\|_1\leq \frac{1}{3}n_i$, one has $\operatorname{dist} (x,\partial U_i)\geq4l_{n_i}>  l_{n_i}^\frac{5}{6} $.  Thus applying the Poisson formula gives 
		$$|\psi (x)|\leq\sum_{z,z'}|G_{U_i} (x,z)\Gamma_{z,z'} \psi (z')|\lesssim l_{n_i}^{2d}e^{-2\gamma_0l_{n_i}}\leq e^{-\gamma_0l_{n_i}} .$$
		Fixing $x\in \Z^d$, we have $\|x\|_1\leq \frac{1}{3}n_i$, as $i\to \infty$.  Thus we must have $\psi (x)=0$ by letting $i\to\infty$, a contradiction since $\psi\neq 0$.  
	\end{proof}
	
	We are ready to prove Anderson localization. 
	\begin{proof}[Proof of Anderson localization]
		First,  we  note that  fixing $E_n\in \mathcal{C}_n$, if $E_n$ belongs to \textbf{Type} \ref{t3}, then for any $\theta\in I (E_{n,>})\cup I (E_{n,<})$ and $0<\|y\|_1\leq  10l_{n+2}$,  we have $\theta+y\cdot \omega \notin I (E_{n,>})\cup I (E_{n,<})$ since  $I (E_{n,>})\cup I (E_{n,<})$ is a single interval with length $\ll  l_{n+2}^{-\tau}$, and if $\theta\in  I (E_n)\setminus E_n^{-1} (J^{DR} (E_n))$, then $\theta$ is in the simple resonant interval. An analogous argument in  the proof of Proposition \ref{SRnpro} yields that   for $0<\|y\|_1\leq  10l_{n+2}$ with  $\theta+y\cdot \omega \in I (E_n)$, 
		\begin{equation}\label{ALL1}
			|E_n (\theta)-E_n (\theta+y\cdot \omega )|\geq 3\delta_n.
		\end{equation}
		
		Let  $n\geq \max  (N_1,N_2):=N$. 	By Lemma \ref{ALL} and the definition of $S_n (\theta,E)$, there exists some $x_n$ satisfying $\|x_n\|_1\leq 100l_n\leq l_{n+1}$, such that $\theta+x_n\cdot \omega \in I (E_n)$ and 
		$$|E_n (\theta+x_n\cdot \omega )-E|\leq \delta_n.$$
		Since  $\theta\notin K_n$, by the definition of $K_n$, we have   $\theta+x_n\cdot \omega\notin  E_n^{-1} (J^{DR} (E_n))$. Thus by \eqref{ALL1} (with $\theta+x_n\cdot \omega$ in the role of $\theta$), for any $x$ with $0<\|x-x_n\|_1\leq 10l_{n+2}$ and  $ \theta+x\cdot \omega \in I (E_n)$,
		$$|E_n (\theta+x\cdot \omega )-E|\geq |E_n (\theta+x\cdot \omega)-E_n (\theta+x_n\cdot \omega )|- |E_n (\theta+x_n\cdot \omega )-E|\geq \delta_n.$$
		It follows that the annulus $\Lambda_{9l_{n+2}}\setminus \Lambda_{l_{n+1}}$ is $n$-nonresonant.  Thus there exists an $n$-good set $\Lambda$, which is  a deformation of  $\Lambda_{8l_{n+2}}\setminus \Lambda_{2l_{n+1}}$ satisfying 
		$$  (\Lambda_{8l_{n+2}}\setminus \Lambda_{2l_{n+1}})\subset   \Lambda \subset (\Lambda_{9l_{n+2}}\setminus \Lambda_{l_{n+1}}).$$
		Let  $x\in \Lambda_{6l_{n+2}}\setminus \Lambda_{4l_{n+1}}.$ Then $$\operatorname{dist} (x,\partial\Lambda)\geq \frac{1}{5}\|x\|_1>l_n^{\frac{5}{6}}. $$
		Thus 
		\begin{align*}
			|\psi (x)|&\leq\sum_{z,z'}|G_{\Lambda} (x,z)\Gamma_{z,z'} \psi (z')|\\
			&\lesssim  (l_{n+2})^{2d}e^{-\frac{1}{10}\gamma_0\|x\|_1}\\
			&\leq e^{-\frac{1}{20}\gamma_0\|x\|_1}.	
		\end{align*}
		Since $$\left(\Z^d\setminus\bigcup_{n\geq N} (\Lambda_{6l_{n+2}}\setminus \Lambda_{4l_{n+1}} )\right)\subset \Lambda_{4l_{N+1}}, $$
		the above estimate shows  the exponential decay of $\psi (x)$ for $\|x\|_1> 4l_{N+1}$. 
		
		We  finish the proof of Anderson localization for $\theta\in \Theta$.
	\end{proof}

	\subsection{Dynamical  localization}
	We first construct the set where dynamical localization can  occur.
	
	Fix $E_n\in \mathcal{C}_n\  (n\geq 1)$.  Recalling the definition of $J^{DR} (E_n)$ in the previous section (c.f. \eqref{ALDL}), we  define  analogously   $J^{DR}_\rho$ for $\rho>0$,
	$$J^{DR}_\rho (E_n):=   B_{\rho\delta_{n-1}^{9}} (\tilde{k}_n), $$
	and 
	$$K_{\rho}:=\bigcup_{n\geq 1}\bigcup_{E_n\in\mathcal{C}_n}\bigcup_{\|k\|_1\leq l_{n+1}} (E_n^{-1} (J^{DR}_\rho  (E_n))+k\cdot \omega ). $$
	As the proof of Lemma \ref{ALLem}, one can prove  $|K_{\rho}|\leq \rho^\frac{1}{2}.$ Recalling the set $K$ (c.f. \eqref{ALset}), we have  $K\subset K_{\rho}$ for any $\rho>0$. Thus Anderson localization holds for $\theta\in \T\setminus K_\rho$ by the previous theorem. We will show dynamical localization also occurs on this set. 
	
	For the remainder of this section,  we  fix $\theta\in  \T\setminus K_\rho$ and  denote by  $\{\lambda_\alpha,\varphi_\alpha\}_{\alpha\in\N}$ the  complete set of  eigenvalues and corresponding eigenfunctions  of $H (\theta)$.
	Then 
	$${\bm e}_0=\sum_\alpha\varphi_\alpha (0)\varphi_\alpha,$$ where 
	$\varphi_\alpha (0)=\langle \varphi_\alpha, {\bm e}_0\rangle.$
	Hence 
	$$e^{itH(\theta)}{\bm e}_0=\sum_\alpha e^{it\lambda_\alpha}\varphi_\alpha (0)\varphi_\alpha.$$
	To control $$\sum_{x\in \mathbb{Z}^d} (1+\|x\|_1)^q|\langle e^{itH (\theta)}{\bm e}_0, {\bm e}_x\rangle|,$$ it suffices to estimate 
	\begin{equation*}
		\sum_\alpha\left(\sum_x  (1+\|x\|_1)^q |\varphi_\alpha (x)|\right) |\varphi_\alpha (0)|.
	\end{equation*}
	Let $A_0=\emptyset$ and $A_j=\{\alpha:\ |\varphi_\alpha (0)|>e^{-\gamma_0l_j}\}\   (j\geq1).$
	Then we can rewrite the above summation as \begin{equation}\label{kuai}
		\sum_{j=1}^{+\infty}\sum_{\alpha\in A_j\setminus A_{j-1}}\left(\sum_x  (1+\|x\|_1)^q |\varphi_\alpha (x)|\right) |\varphi_\alpha (0)|.
	\end{equation}
	{\begin{lem}
			Let  $\alpha\in A_j$.  Then for all  $n\geq j$, there exists some $E_n\in \mathcal{C}_n$  such that  $\lambda_\alpha\in \tilde{J} (E_n)$. Moreover, there exists an $x_n\in \Z^d$ satisfying $\|x_n\|_1\leq 100l_n$ such that $x_n\in S_n (\theta,\lambda_\alpha)$ relative to this $E_n$. 
		\end{lem}
		\begin{proof} The proof is similar  to  that of Lemma \ref{ALL}.
	\end{proof}}
	Let $m\geq 1$ be the smallest integer such that $\rho \geq \delta_{m-1}$. From  the above lemma,  as in  the proof of  Anderson localization (via the Poisson formula argument), one can prove that for $\alpha\in A_j,x\in \Z^d $ satisfying $\|x\|_1>4\max  (l_{m+1}, l_{j+1}):=L_{\rho,j}$, 
	$$	|\varphi_\alpha (x)|\leq e^{-\frac{1}{20}\gamma_0\|x\|_1}.	$$
	By the exponential decay of the eigenfunction, we employ the Hilbert-Schmidt argument to obtain   \begin{align*}
		L_{\rho,j}^d&\gtrsim\sum_{\|x\|_1\leq L_{\rho,j}}\sum_\alpha|\varphi_\alpha (x)|^2\\
		&\geq\sum_{\alpha\in A_j}\sum_{\|x\|_1\leq L_{\rho,j}}|\varphi_\alpha (x)|^2\\
		&=\sum_{\alpha\in A_j} (1-\sum_{\|x\|_1>L_{\rho,j} }|\varphi_\alpha (x)|^2)\\
		&\geq\frac{1}{2}|A_j|.
	\end{align*}
	Thus   $|A_j|\lesssim L_{\rho,j}^{d}$. 
	
	Using the estimate   	$|\varphi_\alpha (x)|\leq e^{-\frac{1}{20}\gamma_0\|x\|_1}$ for $\alpha \in A_m$ and  $\|x\|_1> L_{\rho,m} =4l_{m+1}$, we get 
	\begin{align}\label{jizhong}\sum_{j=1}^{m}&\sum_{\alpha\in A_j\setminus A_{j-1}}\left(\sum_x (1+\|x\|_1)^q |\varphi_\alpha (x)|\right) |\varphi_\alpha (0)|\leq\sum_{\alpha\in A_m}\left (\sum_x (1+\|x\|_1)^q |\varphi_\alpha(x)|\right)\nonumber\\
		&\leq|A_m|\sup_{\alpha\in A_m}\left(\sum_{\|x\|_1\leq4l_{m+1}}+\sum_{\|x\|_1>4l_{m+1}}\right) (1+\|x\|_1)^q |\varphi_\alpha (x)|\nonumber\\
		&\leq C ({q,d}) l_{m+1}^{q+2d}.
	\end{align}
	Since  $|\varphi_\alpha (0)|\leq e^{-\gamma_0l_{j-1}} $ for $\alpha\notin A_{j-1}$, using the estimate  	$|\varphi_\alpha (x)|\leq e^{-\frac{1}{20}\gamma_0\|x\|_1}$ for $  \alpha \in A_j \  (j> m)$ and $\|x\|_1>L_{\rho,j}= 4l_{j+1}$,    we get 
	\begin{align*}
		&\ \ \ \sum_{\alpha\in A_j\setminus A_{j-1}}\left(\sum_x  (1+\|x\|_1)^q |\varphi_\alpha (x)|\right) |\varphi_\alpha (0)|\\
		&\leq|A_j| e^{-\gamma_0l_{j-1}}\sup_{\alpha\in A_j}\left(\sum_{\|x\|_1\leq4l_{j+1}}+\sum_{\|x\|_1>4l_{j+1}}\right) (1+\|x\|_1)^q |\varphi_\alpha (x)|\\
		&\leq C (q,d)l_{j+1}^{q+2d	}e^{-\gamma_0l_{j-1}}.
	\end{align*}
	Summing  up the above estimates for  $j$ from $m+1$ to $+\infty$ implies  
	\begin{align}\label{zai}
		\sum_{j=m+1}^{+\infty}\sum_{\alpha\in A_j\setminus A_{j-1}}\left(\sum_x  (1+\|x\|_1)^q |\varphi_\alpha (x)|\right) |\varphi_\alpha (0)|\leq C (q,d)e^{-\frac{1}{2}\gamma_0l_{m}}.
	\end{align}
	From \eqref{jizhong} and \eqref{zai}, we obtain 
	\begin{equation}
		\begin{aligned}
			\eqref{kuai}&\leq C ({q,d})l_{m+1}^{q+2d}\\ 
			&\leq C ({q,d})\max (|\log \delta_{m-2}|^{100 (q+2d)},l_{2}^{q+2d})\\ 
			&\leq C ({q,d})\max (|\log \rho|,|\log\varepsilon_0|)^{100 (q+2d)}  <+\infty, \label{dede}
		\end{aligned}
	\end{equation}
	where  on the third line, we use  $\rho\leq \delta_{m-2}$ for $m\geq 2$ by the definition of $m$.  Hence we finish  proving $H (\theta)$ satisfies   dynamical localization  on the set  $\T\setminus K_\rho$ for all $\rho>0$.  Since $ |K_\rho|\leq \rho^{\frac{1}{2}}$, it follows that  dynamical localization holds for Lebesgue almost $\theta\in \T$.
	
	Next,  we  prove the expectation version of  dynamical localization. For this purpose, we employ the estimate from \eqref{dede} to get 
	\begin{align*}
		&\ \ \ \left(\int_{\T\setminus K_{\delta_1}}	+\sum_{n=1}^{+\infty}	\int_{K_{\delta_{n}}\setminus K_{\delta_{n+1}}}\right)\sup_{t\in \mathbb{R}}\sum_{x\in \mathbb{Z}^d} (1+\|x\|_1)^q|\langle e^{itH (\theta)}{\bm e}_0, {\bm e}_x\rangle| d\theta\\
		&\leq C ({q,d})\left (|\log\delta_1|^{100 (q+2d)}+\sum_{n=1}^{+\infty}|\log\delta_{n+1}|^{100 (q+2d)}|K_{\delta_{n}}|\right)\\
		&<+\infty.
	\end{align*}
	
	\subsection{$ (\frac{1}{2}-)$-H\"older regularity of IDS} In this section, 
	we fix $\theta \in \T ,E\in \R$ and $\eta>0$. We are going to estimate  the number of eigenvalues of  $H_{\Lambda_N} (\theta)$ belonging to $[E-\eta,E+\eta]$ as $N\to +\infty$.   
	For this purpose, we first introduce a lemma  relating  the  operator norm of Green's function  to the number of eigenvalues  inside the energy  interval $[E-\eta,E+\eta]$.
	%Before proving this Theorem, we prove a lemma which 
	{\begin{lem}\label{IDSL}
			Let $H$ be a  self-adjoint operator on $\Z^{d}$  and $\Lambda\subset \Z^d$  be a finite set. Fix $E\in \R$. Assume there exists a set   $\Lambda'\subset\Z^d$ such that $|\Lambda\Delta\Lambda'|\leq M$ and $\|G_{\Lambda'} (E)\|\leq (2\eta)^{-1}$,  where $G_{\Lambda'} (E)= (H_{\Lambda'}-E)^{-1}$. Then the number of eigenvalues of $H_\Lambda$  inside $[E-\eta,E+\eta]$ does  not exceed  $3M$. 
		\end{lem}
		\begin{rem}
			This lemma only works for self-adjoint matrices. 
		\end{rem}
		\begin{proof}
			We refer to \cite{CSZ23a} for the  detailed proof (c.f. page 553 of \cite{CSZ23a}). 

	\end{proof}}
	
	Let  $N$ be  sufficiently large depending on $\eta$. The proof proceeds by case analysis:\\
	\textbf{(1).}  $\varepsilon_0 =\delta_0^{20}<\eta$.\\ We define   $$S_\eta:=\{x\in \Lambda_N:\ |v (\theta+x\cdot \omega)-E|\leq (2d+2) \eta\}.$$ 
	Since $v$ has a two-monotonicity interval structure and satisfies the Morse condition, 
	$$|v^{-1}\left( [E- (2d+2) \eta,E+ (2d+2) \eta]\right)|\leq C (d) \eta^\frac{1}{2}.$$
	From the uniform distribution of $\{x\cdot \omega\}_{x\in \Z^d}$, we deduce that  $|S_\eta|\leq C (d)\eta^{\frac12}|\Lambda_N|$ for sufficiently large $N$. Let  $\Lambda'=\Lambda_N\setminus S_\eta$. Then by Neumann series argument,  %the absolute value of the  diagonal elements of  $H_{\Lambda'}-E$ have a lower bound $ (2d+2)\eta$.
	%Thus
	$$\|G_{\Lambda'} (E;\theta)\|\leq  ( (2d+2)\eta-\varepsilon\|\Delta\|)^{-1}  \leq ( (2d+2)\eta-2d\varepsilon)^{-1}\leq (2\eta)^{-1}.$$
	By Lemma \ref{IDSL}, $H_{\Lambda_N} (\theta)$ has at most $C (d)\eta^{\frac12}|\Lambda_N|$ eigenvalues in $[E-\eta,E+\eta]$. Thus  $$\mathcal{N}_{\Lambda_N} (E+\eta;\theta)-\mathcal{N}_{\Lambda_N} (E-\eta;\theta)\leq C (d)\eta^{\frac12}.$$
	\textbf{(2).} $\eta\in  (\delta_{n+1}^{20},\delta_n^{20}]$ for some $n\geq0$.\\  Replacing $\delta_{n+1}$ with $\eta$ in the proof of Theorem \ref{rep} where we essentially use the relation  $|\log\delta_{n}|\ll|\log\delta_{n+1}|\ll l_{n+1}$, one can obtain  the following lemma:
	\begin{lem}\label{fff}
		Let $E_{n+1}\in \mathcal{C}_{n+1}$ be such that $E\in \tilde{J} (E_{n+1})$ or $E_{n+1}$ minimizes $\operatorname{dist} (\tilde{J} (\textbf{E}),E)$ among $\textbf{E}\in\mathcal{C}_{n+1}$. Let $\Lambda\subset\Z^d$ be an $(n+1)$-regular set satisfying    $$|E_{n+1} (\theta+x\cdot \omega)-E|\geq 20\eta$$ for any $x\in \Lambda$ with $\theta+x\cdot \omega \in  I (E_{n+1})$.  Then $$\|G_\Lambda (\theta,E)\|\leq \frac{1}{2}\eta^{-1}.$$
	\end{lem} 
	We find  an $(n+1)$-regular set $\Lambda$  such that $$\Lambda_N\subset\Lambda\subset\Lambda_{N+50l_{n+1}}$$
	and  define   
	$$S_\eta:=\{ x\in \Lambda  : \  \theta+x\cdot \omega \in  I (E_{n+1}),\  |E_{n+1} (\theta+x\cdot \omega)-E|< 20\eta\}.$$
	Then the set $\Lambda':=\Lambda \setminus (\bigcup_{x\in S_\eta}B_{n+1}+x)$ satisfies the assumption of Lemma \ref{fff} and so \begin{equation}\label{1241}
		\|G_{\Lambda'} (\theta,E)\|\leq \frac{1}{2}\eta^{-1}.
	\end{equation}
	As done above,  we have for sufficiently large $N$,  $$|S_\eta|\leq C (d)\eta^{\frac12}|\Lambda_{N+50l_{n+1}}|\leq C (d)\eta^{\frac12}|\Lambda_{N}|.$$ It  follows that 
	\begin{equation}\label{1242}
		\begin{aligned}
			|\Lambda_N\setminus\Lambda'|+|\Lambda'\setminus\Lambda_N|
			&\leq|\Lambda\setminus\Lambda'|+|\Lambda\setminus\Lambda_N|\\
			&\leq C (d) (l_{n+1}^d |S_\eta|  +l_{n+1}N^{d-1})\\
			&\leq C (d) (l_{n+1}^d\eta^{\frac12}|\Lambda_N|+l_{n+1}N^{d-1}).
		\end{aligned}
	\end{equation}
	Combining \eqref{1241}, \eqref{1242} and  Lemma \ref{IDSL}  implies
	\begin{align*}
		\mathcal{N}_{\Lambda_N} (E+\eta;\theta)-\mathcal{N}_{\Lambda_N} (E-\eta;\theta)&\leq C (d) (l_{n+1}^d\eta^{\frac12}+l_{n+1}/N)\\ 
		&\leq C (d)\eta^{\frac{1}{2}}|\log\eta|^{16d}
	\end{align*}
	provided $N$ sufficiently large,  where we use   $$l_{n+1}\leq l_n^8\leq|\log \delta_n|^{16}\leq|\log\eta|^{16}$$ on the second line.
	
	We  finish the proof by combining the above two cases and letting $N\to +\infty$.

	\appendix
	\section{Proof of Lemma \ref{ll}}
	\begin{proof}
		Without loss of generality, we may assume $E^*=0$. It suffices to show the smallest $m$ eigenvalues of the  positive semidefinite operator   $H^2$,   $0\leq\lambda_1\leq\cdots\leq\lambda_m$ satisfy
		$$	\sum_{k=1}^m \lambda_k\leq m\delta^2.$$
		Denote by $P$ the orthogonal projection on the space spanned by $\psi_k\  (1\leq k\leq m)$. Thus the compressed operator $PH^2P$ has $m$  eigenvalues $0\leq\mu_1\leq\cdots\leq\mu_m$ satisfying $\lambda_k\leq\mu_k$ by min-max principle. Thus we obtain \begin{align*}
			\sum_{k=1}^m \lambda_k\leq\sum_{k=1}^m \mu_k&=\operatorname{Trace} (PH^2P)\\
			&\leq \sum_{k=1}^{m}\langle\psi_k,PH^2P\psi_k\rangle\\
			&=\sum_{k=1}^m\|H\psi_k\|^2\\
			&\leq m\delta^2 
		\end{align*} and finish the proof.
	\end{proof}
	
	\section{Proof of Lemma \ref{C2}}
	\begin{proof}
		Without loss of generality, we may consider the case $\left|E'\right| \leq \delta$ implies $E^{\prime \prime} \geq 2$ and $\theta_s\leq \theta_1\leq \theta_2$. Since  $E'\left(\theta_s\right)=0$,  by assumption  $E^{\prime \prime}\left(\theta_s\right) \geq 2$. Let $\theta_d$ be the largest point with the following property:
		$$
		E^{\prime \prime} (\theta) \geq 2 \text { for } \theta_s \leq \theta \leq \theta_d .
		$$
		This implies that $E (\theta)$ is an increasing function near the  right of $\theta_s$. By definition of $\theta_d$, we have $E^{\prime \prime}\left (\theta_d+\Delta \theta_n\right)<2$ for a sequence $\Delta \theta_n\to 0^+$; therefore $E'\left (\theta_d+\Delta \theta_n\right) >\delta$. This inequality must hold for every $\theta>\theta_d$ or else we would have a point $\theta^*>\theta_d$ where $E' (\theta^*)=\delta,E' (\theta^*-\Delta\theta)>\delta$ for small $\Delta\theta>0$,  but   $E^{\prime \prime} (\theta^*) \geq 2>0$ by  $|E' (\theta^*)|\leq\delta$ and the assumption of $E$. This is impossible. Therefore
		$$
		E' (\theta) \geq \delta \text { for } \theta>\theta_d .
		$$
		There are three cases.\\
		{\it Case} 1: $\theta_s\leq\theta_1<\theta_2\leq\theta_d$.  
		\begin{align*}
			E (\theta_2)-E (\theta_1)= E' (\theta_1)  (\theta_2-\theta_1)+\frac{1}{2}E'' (\xi) (\theta_2-\theta_1)^2\geq (\theta_2-\theta_1)^2.
		\end{align*}
		{\it Case}  2: $\theta_d\leq\theta_1<\theta_2$.
		\begin{align*}
			E (\theta_2)-E (\theta_1)= E' (\xi)  (\theta_2-\theta_1)\geq \delta (\theta_2-\theta_1)\geq  (\theta_2-\theta_1)^2.
		\end{align*}
		{\it Case}  3: $\theta_s\leq\theta_1\leq\theta_d\leq\theta_2$. 
		\begin{align*}
			E (\theta_2)-E (\theta_1)&= E (\theta_2)-E (\theta_d)+E (\theta_d)-E (\theta_1) \\
			&\geq  (\theta_2-\theta_d)^2+ (\theta_d-\theta_1)^2\\
			&\geq\frac{1}{2} (\theta_2-\theta_1)^2.
		\end{align*}
		
		To prove the second inequality, we assume $\theta_s\leq\theta$  (The case $\theta\leq\theta_s$ is analogous).  We need to consider two cases. \\
		{\it Case} 1: $\theta_s \leq \theta \leq \theta_d$.
		$$
		E' (\theta)=E' (\theta_s)+E'' (\xi) (\theta-\theta_s) \geq \theta-\theta_s .
		$$
		{\it Case}  2: $\theta_d \leq \theta$. In this case we have $E' (\theta) \geq \delta$. \\
		
		Hence we finish the proof. 
	\end{proof} 
	\section{Proof of Lemma \ref{daoshu}}
	\begin{proof}
		The local  $C^2$ smoothness of $E (\theta)$ is inherited from  the operators since $E (\theta^*)$ is a simple eigenvalue. 	Noting  that $\langle\psi,\psi\rangle\equiv1$, we have $\langle\psi,\psi'\rangle=0$.
		Differentiating  the equation  $E= \langle\psi,H\psi\rangle$ yields 
		\begin{equation}\label{yijie}
			E'=\langle\psi,H'\psi\rangle+2\langle\psi',H\psi\rangle=\langle\psi,H'\psi\rangle,
		\end{equation} where we use   $H\psi=E\psi$ and $\langle\psi,\psi'\rangle=0$.
		
		Now we prove \textbf{(2)}.
		Differentiating  \eqref{yijie} gives 
		$E''=\langle\psi,H''\psi\rangle+2\langle \psi',H'\psi\rangle$. Thus it suffices to show 
		$$\langle \psi',H'\psi\rangle=-\langle H'\psi,G^\perp (E)H'\psi\rangle.$$
		Since $\psi'$ is orthogonal to $\psi$, we have 
		\begin{align*}
			\langle \psi',H'\psi\rangle =	\langle P^\perp \psi',H'\psi\rangle=\langle G^\perp (E) (H-E)\psi',H'\psi\rangle=-\langle H'\psi,G^\perp (E)H'\psi\rangle,
		\end{align*}
		where we use  $G^\perp (E)\psi=0 $ and  $ (H'-E')\psi=- (H-E)\psi'$ since $ (H-E)\psi\equiv0$.

		Finally, \textbf{(3)} follows from $$G^\perp (E)= ( (H-E)^\perp)^{-1}=\sum_{E'\neq E}\frac{1}{E'-E}P_{E'}$$
		and 
		$$G^{\perp \perp} (E)= ( (H-E)^{\perp \perp})^{-1}=\sum_{E'\neq E,\mathcal{E}}\frac{1}{E'-E}P_{E'}$$
		immediately.
	\end{proof}
	\section{Lemmas for constructing the  desired regular deformation}\label{chouti}
	The desired deformation can be made by  the following two lemmas:
	\begin{lem}
		Let $\theta^*\in\T,E^*\in\R$ and $E_n\in\mathcal{C}_n$. Then any three different points $x_1,x_2,x_3\in S_n (\theta^*,E^*)$ relative to $E_n$ satisfy
		$$\max (\|x_1-x_2\|_1,\|x_2-x_3\|_1,\|x_1-x_3\|_1)\gtrsim \delta_n^{-\frac{1}{2\tau}}\gg l_n.$$ 
		The above separation also holds uniformly for $\theta$ and $E$ belonging to a small neighborhood (of  order $\delta_n$) of $ (\theta^*,E^*)$.
	\end{lem}
	\begin{proof}
		By the two-monotonicity interval structure and Morse condition of $E_n$, it follows that  
		the preimage $E_n^{-1} (B_{\delta_n} (E^*))$ has at  most two connected components with length $\lesssim \delta_n^\frac{1}{2}$. The lemma follows from the  Diophantine condition and pigeonhole principle immediately.
	\end{proof}
	\begin{lem}
		Let $L>1$. Assume a set  $S\subset\Z^d$ satisfies the following condition:\\ Any three different points $x_1,x_2,x_3\in S$ satisfy 
		$$\max (\|x_1-x_2\|_1,\|x_2-x_3\|_1,\|x_1-x_3\|_1)> 100L.$$
		Then for any set $B\subset\Z^d$, there exists a deformation $B^*$ satisfying
		$$B\subset B^* \subset (B+\Lambda_{50L}) =\{x\in \Z^d:\ \operatorname{dist} (x,B)\leq 50L\},   $$
		such that $(\Lambda_L+x)\subset  B^*$ for any $x\in S$ with $(\Lambda_L+x)\cap B^*\neq \emptyset$.
	\end{lem}
	\begin{proof}
		We inductively define $B_0:=B$, 
		\begin{equation}\label{449}
			B_{0} \subsetneqq B_{1}\subsetneqq \cdots \subsetneqq B_{N} 
		\end{equation} satisfying 
		$$B_{n+1}=B_n\bigcup\left(\bigcup_{\{x\in S: \ (\Lambda_L+x)\cap B_n\neq \emptyset\} }  (\Lambda_L+x)  \right)  $$ for $n\geq  0$,  
		where $N$ is the  largest integer such that  the ``$\subsetneqq$'' relationship holds in \eqref{449}. We claim that $N\leq 2$. Otherwise, there exist three  different points  $x_1,x_2,x_3\in S$ such that 
		$$(\Lambda_L+x_1)\cap B_0\neq \emptyset,\  (\Lambda_L+x_2)\cap(\Lambda_L+x_1)\neq \emptyset, \  (\Lambda_L+x_3)\cap(\Lambda_L+x_2)\neq \emptyset.$$
		Thus $\max (\|x_1-x_2\|_1,\|x_2-x_3\|_1,\|x_1-x_3\|_1)\leq 4L,$ a contradiction. Then $B^*:=B_N$ is the desired deformation.
	\end{proof}
	\section*{Acknowledgments}
	H. Cao  and Z. Zhang are partially supported by the National Key R\&D Program of China under Grant 2023YFA1008801.  Y. Shi is  partially supported by NSFC  (12271380) and Z. Zhang is  partially supported by NSFC  (12288101).   The authors are very grateful to the handling editor and the anonymous referees for their helpful suggestions. 
	\section*{Data Availability}
	The manuscript has no associated data.
	\section*{Declarations}
	{\bf Conflicts of interest} \ The authors  state  that there is no conflict of interest.
	\bibliographystyle{alpha}
	%\bibliography{CSZ20240605}

\end{document}